\documentclass[11pt,aps,pra,notitlepage,tightenlines,nofootinbib,superscriptaddress]{revtex4-1}
\pdfoutput=1

\usepackage{newpxtext,newpxmath}

\let\coloneqq\relax

\usepackage[utf8]{inputenc}

\usepackage{amsthm}
\usepackage{amssymb}
\usepackage{amsmath}
\usepackage{bbm}
\usepackage[pdftex, backref=page]{hyperref}
\usepackage{braket}
\usepackage{mathdots}
\usepackage{mathtools}
\usepackage{enumerate}
\usepackage[shortlabels]{enumitem}
\usepackage{csquotes}
\usepackage[cal=boondox]{mathalfa}
\usepackage{graphicx}
\usepackage{stackengine}
\usepackage{scalerel}
\usepackage{tensor}       
\usepackage[dvipsnames,svgnames]{xcolor}
\usepackage{array}
\usepackage{makecell}
\newcolumntype{x}[1]{>{\centering\arraybackslash}p{#1}}
\usepackage{tikz}
\usepackage{pgfplots}
\usetikzlibrary{shapes.geometric, shapes.misc, positioning, arrows, arrows.meta, decorations.pathreplacing, decorations.pathmorphing, patterns, angles, quotes, calc}
\usepackage{booktabs}
\usepackage{xfrac}
\usepackage{siunitx}
\usepackage{centernot}
\usepackage{comment}
\usepackage{chngcntr}
\usepackage{caption}
\usepackage{subcaption}

\newtheorem{thm}{Theorem}
\newtheorem*{thm*}{Theorem}
\newtheorem{prop}[thm]{Proposition}
\newtheorem*{prop*}{Proposition}
\newtheorem{lemma}[thm]{Lemma}
\newtheorem*{lemma*}{Lemma}
\newtheorem{cor}[thm]{Corollary}
\newtheorem*{cor*}{Corollary}
\newtheorem{cj}[thm]{Conjecture}
\newtheorem*{cj*}{Conjecture}
\newtheorem{Def}[thm]{Definition}
\newtheorem*{Def*}{Definition}

\newtheorem*{question*}{Question}

\newtheorem*{problem*}{Problem}

\makeatletter
\def\thmhead@plain#1#2#3{%
  \thmname{#1}\thmnumber{\@ifnotempty{#1}{ }\@upn{#2}}%
  \thmnote{ {\the\thm@notefont#3}}}
\let\thmhead\thmhead@plain
\makeatother

\theoremstyle{definition}
\newtheorem{rem}[thm]{Remark}

\hypersetup{
    bookmarksnumbered=true, 
    breaklinks=true,
    unicode=false, 
    pdfstartview={FitH}, 
    pdfnewwindow=true, 
    colorlinks=true, 
    allcolors=MidnightBlue!70!black!70!TealBlue
}

\newcommand{\bb}{\begin{equation}\begin{aligned}\hspace{0pt}}
\newcommand{\bbb}{\begin{equation*}\begin{aligned}}
\newcommand{\ee}{\end{aligned}\end{equation}}
\newcommand{\eee}{\end{aligned}\end{equation*}}
\newcommand\floor[1]{\left\lfloor#1\right\rfloor}
\newcommand\ceil[1]{\left\lceil#1\right\rceil}
\newcommand{\eqt}[1]{\stackrel{\mathclap{\mbox{\scriptsize #1}}}{=}}
\newcommand{\leqt}[1]{\stackrel{\mathclap{\mbox{\scriptsize #1}}}{\leq}}

\newcommand{\geqt}[1]{\stackrel{\mathclap{\mbox{\scriptsize #1}}}{\geq}}
\newcommand*{\spr}[2]{\langle #1 | #2 \rangle}
\newcommand{\ketbra}[1]{\ket{#1}\!\!\bra{#1}}

\newcommand{\ketbraa}[2]{\ket{#1}\!\!\bra{#2}}

\newcommand{\sumno}{\sum\nolimits}

\newcommand{\e}{\varepsilon}
\renewcommand{\epsilon}{\varepsilon}

\newcommand{\id}{\mathbbm{1}}

\DeclareMathOperator{\Tr}{Tr}

\DeclareMathAlphabet{\pazocal}{OMS}{zplm}{m}{n}

\DeclareMathOperator{\supp}{supp}

\newcommand{\HH}{\pazocal{H}}

\newcommand{\NN}{\pazocal{N}}

\newcommand{\XX}{\mathcal{X}}
\newcommand{\YY}{\mathcal{Y}}

\newcommand{\FF}{\mathcal{F}}

\newcommand{\lsmatrix}{\left(\begin{smallmatrix}}
\newcommand{\rsmatrix}{\end{smallmatrix}\right)}

\stackMath

\stackMath

\makeatletter
\newcommand*\rel@kern[1]{\kern#1\dimexpr\macc@kerna}
\newcommand*\widebar[1]{%
  \begingroup
  \def\mathaccent##1##2{%
    \rel@kern{0.8}%
    \overline{\rel@kern{-0.8}\macc@nucleus\rel@kern{0.2}}%
    \rel@kern{-0.2}%
  }%
  \macc@depth\@ne
  \let\math@bgroup\@empty \let\math@egroup\macc@set@skewchar
  \mathsurround\z@ \frozen@everymath{\mathgroup\macc@group\relax}%
  \macc@set@skewchar\relax
  \let\mathaccentV\macc@nested@a
  \macc@nested@a\relax111{#1}%
  \endgroup
}

\counterwithin*{equation}{part}
\counterwithin*{thm}{part}
\counterwithin*{figure}{part}

\tikzset{meter/.append style={draw, inner sep=10, rectangle, font=\vphantom{A}, minimum width=30, line width=.8, path picture={\draw[black] ([shift={(.1,.3)}]path picture bounding box.south west) to[bend left=50] ([shift={(-.1,.3)}]path picture bounding box.south east);\draw[black,-latex] ([shift={(0,.1)}]path picture bounding box.south) -- ([shift={(.3,-.1)}]path picture bounding box.north);}}}
\tikzset{roundnode/.append style={circle, draw=black, fill=gray!20, thick, minimum size=10mm}}
\tikzset{squarenode/.style={rectangle, draw=black, fill=none, thick, minimum size=10mm}}

\definecolor{Blues5seq1}{RGB}{239,243,255}
\definecolor{Blues5seq2}{RGB}{189,215,231}
\definecolor{Blues5seq3}{RGB}{107,174,214}
\definecolor{Blues5seq4}{RGB}{49,130,189}
\definecolor{Blues5seq5}{RGB}{8,81,156}

\definecolor{Greens5seq1}{RGB}{237,248,233}
\definecolor{Greens5seq2}{RGB}{186,228,179}
\definecolor{Greens5seq3}{RGB}{116,196,118}
\definecolor{Greens5seq4}{RGB}{49,163,84}
\definecolor{Greens5seq5}{RGB}{0,109,44}

\definecolor{Reds5seq1}{RGB}{254,229,217}
\definecolor{Reds5seq2}{RGB}{252,174,145}
\definecolor{Reds5seq3}{RGB}{251,106,74}
\definecolor{Reds5seq4}{RGB}{222,45,38}
\definecolor{Reds5seq5}{RGB}{165,15,21}

\allowdisplaybreaks

\usepackage[most,breakable]{tcolorbox}
\renewenvironment{boxed}[1][white]%
{\expandafter\ifstrequal\expandafter{#1}{filled}{\begin{tcolorbox}[colback=gray!3,colframe=gray!20,breakable=false,enhanced,left=5.75pt,right=5.75pt,grow sidewards by=10pt]}{\begin{tcolorbox}[colback=MidnightBlue!70!black!70!TealBlue!4!white,colframe=MidnightBlue!70!black!70!TealBlue!50!white,breakable=false,enhanced,left=5.75pt,right=5.75pt,grow sidewards by=10pt]}}%
  {\end{tcolorbox}}

\newcommand{\uml}[1]{\ddot{#1}}

\newcommand{\opnsa}[1]{\ensuremath{E_0^{\rm NS,a}(#1)}} 
             
\newcommand{\errnsa}[1]{\ensuremath{E^{\rm NS,a}(0^+,#1)}}  
\newcommand{\errpl}[1]{\ensuremath{E^{\emptyset}(0^+,#1)}}             

\newcommand{\errL}[1]{\ensuremath{E_L^{\emptyset}(0^+,#1)}}

\usepackage{tikz}
\newcommand{\old}[1]{}
\newcommand{\new}[1]{#1}

\begin{document}

\title{Quantum umlaut information}

\author{Filippo Girardi}
\email{filippo.girardi@sns.it}
\affiliation{Scuola Normale Superiore, Piazza dei Cavalieri 7, 56126 Pisa, Italy}
\affiliation{QuSoft, Science Park 123, 1098 XG Amsterdam, The Netherlands}
\affiliation{Korteweg--de Vries Institute for Mathematics, University of Amsterdam, Science Park 105-107, 1098 XG Amsterdam, The Netherlands}

\author{Aadil Oufkir}
\email{oufkir@physik.rwth-aachen.de}
\affiliation{Institute for Quantum Information, RWTH Aachen University,  Germany}

\author{Bartosz Regula}
\affiliation{Mathematical Quantum Information RIKEN Hakubi Research Team, RIKEN Pioneering Research Institute (PRI) and RIKEN Center for Quantum Computing (RQC), Wako, Saitama 351-0198, Japan}

\author{Marco~Tomamichel}
\affiliation{Centre for Quantum Technologies, National University of Singapore, Singapore}
\affiliation{Department of Electrical and Computer Engineering, National University of Singapore, Singapore}

\author{Mario Berta}
\affiliation{Institute for Quantum Information, RWTH Aachen University,  Germany}

\author{Ludovico Lami}
\affiliation{Scuola Normale Superiore, Piazza dei Cavalieri 7, 56126 Pisa, Italy}
\affiliation{QuSoft, Science Park 123, 1098 XG Amsterdam, The Netherlands}
\affiliation{Korteweg--de Vries Institute for Mathematics, University of Amsterdam, Science Park 105-107, 1098 XG Amsterdam, The Netherlands}

\begin{abstract}
We study the quantum umlaut information, a correlation measure defined for bipartite quantum states $\rho_{AB}$ as a reversed variant of the quantum mutual information: $U(A;B)_\rho = \min_{\sigma_B} D(\rho_A\otimes \sigma_B\|\rho_{AB})$ in terms of the quantum relative entropy $D$. As in the classical case~\href{https://arxiv.org/abs/2503.18910}{[Girardi et al., arXiv:2503.18910]}, this definition allows for a closed-form expression and has an operational interpretation as the asymptotic error exponent in the hypothesis testing task of deciding whether a given bipartite state is product or not. We generalise the umlaut information to quantum channels, where it also extends the notion of `oveloh information'~\href{https://arxiv.org/abs/2404.16101}{[Nuradha et al., J.\ Phys.\ A 
58, 165304 (2025)]}. We prove that channel umlaut information is additive for classical-quantum channels, while we observe additivity violations for fully quantum channels. Inspired by recent results in entanglement theory, we then show as our main result that the regularised umlaut information constitutes a fundamental measure of the \emph{quality} of classical information transmission over a quantum channel\,---\,as opposed to the capacity, which quantifies the \emph{quantity} of information that can be sent. This interpretation applies to coding assisted by activated non-signalling correlations, and the channel umlaut information is in general larger than the corresponding expression for unassisted communication as obtained by Dalai for the classical-quantum case \href{https://arxiv.org/abs/1201.5411}{[IEEE Trans.\ Inf.\ Theory 59, 8027 (2013)]}. Combined with prior works on non-signalling--assisted zero-error channel capacities, our findings imply a dichotomy between the settings of zero-rate error exponents and zero-error communication.
While our results are single-letter only for classical-quantum channels, we also give a single-letter bound for fully quantum channels in terms of the `geometric' version of umlaut information.
\end{abstract}

\maketitle


\section{Introduction}

\subsection{Motivation}

The quantification of the information transmission capabilities of quantum channels is a fundamental problem in quantum information theory, with prominent applications in quantum technologies. Most of the literature focused on the question of computing the asymptotic \emph{capacities} of channels, i.e.\ the maximum amount of information (either classical or quantum) that can be sent across a channel with asymptotically vanishing error.
In particular, in the study of communicating classical information through quantum channels, this has led to the expressions for the unassisted classical capacity of a channel as the regularised Holevo information~\cite{Holevo-S-W, H-Schumacher-Westmoreland}, and for entanglement- or non-signalling--assisted classical capacity as the quantum mutual information of the channel~\cite{entanglement-assisted, Bennett2002, leung_2015}. This provides a fundamental operational interpretation to these correlation measures.

Recently, however, there has been a resurgence of interest in a different approach to benchmarking the performance of asymptotic information-theoretic protocols. The main idea is to move the focus away from the \emph{quantity} of the resource that has been obtained (achievable yield with asymptotically vanishing error) to the \emph{quality} of the output (maximum error exponent achievable with arbitrarily high yield). An application of this quality-focused approach to the problem of quantum entanglement distillation by some of us~\cite{lami2024asymptotic} revealed a curious property: the shift of focus can lead to simpler asymptotic expressions, where despite the asymptotic yield rates being given by effectively uncomputable regularised expressions, the asymptotic exponent of the error is given by a single-letter formula. It is in fact given by a rather intriguing quantity: a relative entropy of entanglement where the arguments are reversed in comparison to conventional notions of entanglement entropies~\cite{Vedral1997}. 

This motivates two questions: first, whether this shift of perspective from quantity to quality can lead to new insights and simplifications in other quantum Shannon-theoretic tasks, and second, if other variants of such reversed quantum relative entropies could find operational applications in quantifying the performance of asymptotic protocols.
As for the first question, applications of this mindset can in fact be traced back to early seminal works in classical information theory, where investigation of the zero-rate limit of the error exponent was initiated by Shannon, Gallager, and Belerkamp~\cite{SHANNON1967522, Berlekamp1964BlockCW, Gallager1965}. Similar applications to quantum communications have been more sparse, however --- with two notable exceptions by Holevo~\cite{Holevo2000} and Dalai~\cite{Dalai_2013} --- and many questions remain unanswered. Concerning the second question, the study of reversed entropies has previously appeared in some contexts, including entanglement theory~\cite{eisert_reverse,lami2024asymptotic} as well as quantum statistics, where notably~\cite{Nuradha2024,ji_2024} introduced reverse variants of the Holevo information and other entropic quantities, although not in the context of exact expressions for asymptotic exponents in quantum Shannon theory.


\subsection{Approach}

Before we can investigate quantum channels, it is useful to take a step back and look at quantum states instead. A key quantifier of quantum correlations for quantum states is the quantum mutual information. Mathematically, for a quantum state $\rho_{AB}$ this can be expressed as 
\bb
I(A\!:\!B)_{\rho} = D(\rho_{AB}\| \rho_A \otimes \rho_B) = \min_{\sigma_B} D(\rho_{AB}\| \rho_A \otimes \sigma_B) \, ,
\label{mutual_information_intro}
\ee
where $D(\rho\|\sigma)\coloneqq \Tr \left[ \rho (\log \rho - \log \sigma) \right]$ is the Umegaki quantum relative entropy, and the minimisation is over all quantum states $\sigma_B$ on the $B$ system. The fact that the mutual information appears naturally in both the Holevo quantity~\cite{Holevo-S-W} and in the formula for the entanglement-assisted capacity~\cite{entanglement-assisted} suggests that it is the right starting point for us. Notably, the same quantity also governs the simulation cost of quantum channels~\cite{Bennett2014, Berta2011}, it determines the maximum amount of information that can be securely transmitted using a quantum one-time pad~\cite{Schumacher2006Oct}, and it can be used to bound the generalisation error of learning algorithms~\cite{Caro2024Jun}.

Motivated in particular by the results of~\cite{lami2024asymptotic} in entanglement theory, we can now ask ourselves: what happens if we flip the arguments of the relative entropy in~\eqref{mutual_information_intro}? In~\cite{Filippo25}, we studied this problem in the classical case, inspired also by the earlier ideas of Palomar and Verd\'{u}~\cite{Lautum_08}. There, we considered a measure of correlation between classical random variables $X$ and $Y$ called umlaut information,  defined as
\bb
U(X;Y) = \min_{Q_Y} D(P_X Q_Y \| P_{XY})\, ,
\ee 
where $P_{XY}$ is the joint probability distribution of marginal $P_X$, and $Q_Y$ is an arbitrary probability distribution on the alphabet where $Y$ is defined. The classical umlaut information is additive and has operational interpretation in hypothesis testing and channel coding. In this work, we generalise the umlaut information to the quantum setting and study its properties.  

For a bipartite state $\rho_{AB}$, the \emph{quantum umlaut information} is defined as
\bb \label{intro:umlaut}
U(A;B)_{\rho}&\coloneqq \min_{\sigma_B}D(\rho_A\otimes\sigma_B\|\rho_{AB})\, ,
\ee
where $\sigma_B$ is a quantum state, $\rho_A=\Tr_B[\rho_{AB}]$, and $D$ denotes the Umegaki quantum relative entropy~\cite{Umegaki1962, Hiai1991}. The geometric umlaut information $U_{BS}$ is defined similarly using the Belavkin--Staszewski quantum relative entropy~\cite{Belavkin-Staszewski}. Conceptually similar quantities have previously appeared in~\cite{Nuradha2024}, where a reversed variant of the Holevo information~\cite{Holevo1973} was introduced for a certain class of classical-quantum states and given the name of `oveloh information', as well as in~\cite{ji_2024}, where variants of the umlaut information based on the Petz--R\'enyi $\alpha$-divergences (for $\alpha\neq 1$) have been defined as the Petz--R\'enyi lautum information. Here, we prefer to avoid the terminology `lautum information', which, according to the original formulation by Palomar and Verd\'u~\cite{Lautum_08}, suggests the idea of an \emph{unoptimised} quantity where one sets $\sigma_B = \rho_B$ in~\eqref{intro:umlaut}. The corresponding lautum information and our umlaut information (optimised over $\sigma_B$) behave differently operationally, as already observed in the classical case~\cite{Filippo25}.


\subsection{Overview of results}

The quantum umlaut information can be shown to admit the closed-form expression 
\bb
U(A;B)_{\rho} = -S(\rho_A)-\log\Tr\left[\exp\left(\Tr_A[(\rho_A\otimes\id_B)\log\rho_{AB}]\right)\right] ,
\ee
where $S$ is the von Neumann entropy (Proposition~\ref{prop:closed_lautum}). This formulation implies that the quantum umlaut information is additive under tensor products. In Theorem~\ref{thm:interp_L2}, we show that the quantum umlaut information has an operational interpretation in a composite hypothesis testing problem  that consists in deciding whether a bipartite quantum system is in a product state or not. More precisely, for a quantum state $\rho_{AB}$, we establish that the quantum umlaut information $U(A;B)_{\rho}$ governs the exponential decaying rate of the type-I error while maintaining a type-II error smaller than a constant $\epsilon>0$ for the problem of testing the two hypotheses $H_0: \tau_{A^nB^n}=\rho_{AB}^{\otimes n}$ vs $H_1: \tau_{A^nB^n}=\rho_{A}^{\otimes n}\otimes \sigma_{B^n}$ where $\sigma_{B^n}$ is an arbitrary quantum state. 

For a quantum channel $\pazocal{N}_{A\to B}$, the quantum umlaut information is defined by optimising the corresponding state quantity over all input states
\bb
U(\pazocal{N}) \coloneqq \sup_{\Psi_{A'A}} U(A';B)_{(\mathrm{Id}\otimes\NN)(\Psi)} = \sup_{\rho_{A'}} \min_{\sigma_B} D\left(\rho_{A'}\otimes \sigma_B \,\middle\|\, \rho_{A'}^{1/2} J_{A'B}^{(\pazocal{N})} \rho_{A'}^{1/2}\right) ,
\ee
where $\Psi_{A'A}$ is an arbitrary input state, which can, without loss of generality, assumed to be pure, the states $\rho_{A'}, \sigma_B$ in the rightmost expression are in general mixed, with $\rho_{A'}$ representing the $A'$ marginal of $\Psi_{A'A}$, and $J_{A'B}^{(\pazocal{N})}$ is the (un-normalised) Choi matrix of $\pazocal{N}$ (Definition~\ref{def:channel_lautum}). Extending the classical results of Ref.~\cite{Filippo25}, we show that the channel quantum umlaut information stays additive for classical-quantum channels --- that is, $U(\pazocal{M} \otimes \pazocal{M}) = 2 U(\pazocal{M})$ for all channels $\pazocal{M}$ that take as an input a classical probability distribution but output quantum states. However, the umlaut information is in general non-additive for fully quantum channels. In particular, while super-additivity, expressed by the inequality $U(\pazocal{N}\otimes \pazocal{N}) \geq 2 U (\pazocal{N}) $, always holds (Lemma~\ref{lem:super_add}), in Proposition~\ref{prop:non_add} we exhibit a quantum channel $\pazocal{N}$ such that 
\bb
U(\pazocal{N} \otimes \pazocal{N}) > 2U(\pazocal{N})\, . 
\ee
We then give an operational interpretation of the regularised channel quantum umlaut information in noisy channel coding. For a quantum channel $\pazocal{N}$ and a rate $r$ strictly below the channel capacity $C(\pazocal{N})$, the activated non-signalling error probability of sending a number of messages $M=\floor{\exp(nr)}$ through $n$ i.i.d.\ copies of the channel decays exponentially fast~\cite{Oufkir2024Oct} as
\bb
\epsilon\big(\floor{\exp(nr)}\!,\, \pazocal{N}^{\otimes n}\big) \approx \exp\left[-n E^{\rm NS, a}(r, \pazocal{N})+o(n)\right] ,
\ee
where $E^{\rm NS, a}(r, \pazocal{N})$ denotes the error exponent. As our main result, we show that in the zero-rate regime $\frac{1}{n}\log M \to 0$, the activated non-signalling error exponent is exactly equal to the regularised channel quantum umlaut information $U^\infty(\pazocal{N})\coloneqq\lim_{n\to\infty}\frac{1}{n}U(\pazocal{N}^{\otimes n})$ (Theorem~\ref{thm:exact_tho}). The regularisation is not needed for classical-quantum channels, and in the general case the exponent can also be bounded by the single letter channel geometric umlaut information $U_{BS}(\pazocal{N})$ (Proposition~\ref{prop:hierarchy_umlauts}). We also show that our results characterising the setting of zero-rate communication can be closely connected, in a complementary fashion, with the setting of zero-error non-signalling--assisted coding~\cite{duan_2016}: the umlaut information $U(\pazocal{N})$ of a general quantum channel diverges if and only if its zero-error capacity is positive.

The remainder of our manuscript is structured as follows. In Section~\ref{sec:states} we present the umlaut information and its operational interpretation for bipartite states of a quantum system; in Section~\ref{sec:channels} we discuss the channel umlaut information and its operational interpretation; finally, in Section~\ref{sec:outlook} we discuss some features of our results and raise relevant open questions.


\section{Quantum umlaut information}\label{sec:states}

\subsection{Notation and Preliminaries}

We will denote by $D(p\|q)$ the classical relative entropy between the probability distribution $p$ and $q$, and we will also denote by $D(\rho\|\sigma)$ the Umegaki quantum relative entropy between the quantum states $\rho$ and $\sigma$, defined as
\bb
    D(\rho\|\sigma)=\begin{cases}
        \Tr[\rho(\log\rho-\log\sigma)] & {\rm supp}(\rho)\subseteq {\rm supp}(\sigma)\, ,\\
        +\infty & \text{otherwise.}
    \end{cases}
\ee
The classical mutual information between two random variables $X$ and $Y$ is defined as 
\bb
I(X\!:\!Y)\coloneqq D(P_{XY}\| P_XP_Y)\, , \qquad X\sim P_X\, , \ \ Y\sim P_Y\, , \ \ (X,Y)\sim P_{XY}\, ,
\ee
where $P_X$ and $P_Y$ are the marginal probability distributions of $P_{XY}$. 
Its quantum generalisation for a bipartite state $\rho_{AB}\in\mathcal{D}(\mathcal{H}_A\otimes\mathcal{H}_B)$ is given by
\begin{equation}
I(A\!:\!B)_{\rho}\coloneqq D(\rho_{AB}\| \rho_A\otimes\rho_B)  \qquad\qquad \rho_A=\Tr_B[\rho_{AB}],\, \rho_B=\Tr_A[\rho_{AB}].
\end{equation}
Inspired by the classical mutual information, Palomar and Verd\'{u}~\cite{Lautum_08} defined the lautum information, given by
\begin{equation}
L(X\!:\!Y)\coloneqq D(P_XP_Y\|P_{XY})\, , \qquad X\sim P_X\,,\ \ Y\sim P_Y\,,\ \ (X,Y)\sim P_{XY}\, ,
\end{equation}
and discussed its properties. They also noticed that, although the mutual information can be equivalently rewritten in the variational forms
\bb
    I(X\!:\!Y)= D(P_{XY}\| P_XP_Y)=\min_{Q_Y}D(P_{XY}\| P_XQ_Y)=\min_{Q_X}\min_{Q_Y}D(P_{XY}\| Q_XQ_Y)\, ,
\ee
this property is no longer true for the lautum information, as~\cite{Lautum_08}
\bb
    L(X\!:\!Y)>\min_{Q_Y}D(P_XQ_Y\| P_{XY}).
\ee
for some $P_{XY}$.
In~\cite{Filippo25}, we study the properties and the operational interpretations of the asymmetric version given by
\bb
    U(X;Y)=\min_{Q_Y}D(P_XQ_Y\| P_{XY})\, ,
\ee
which was not considered in detail by Palomar and Verd\'{u}. The aim of this paper is to generalise our classical asymmetric formulation to quantum systems and to explore the differences emerging when we take into account the non-commutative nature of the quantum setting. Such formulation will provide a new tile in the mosaic of quantum communication; in particular, it will allow us to characterise the reliability function of a quantum channel with non-signalling assistance in the limit of vanishing rate of communication.


\subsection{Definition and basic properties}

In what follows, for a Hilbert space $\mathcal{H}$ we will denote by $\mathcal{D}(\mathcal{H})$ the set of density operators on $\mathcal{H}$. A density operator is any positive semi-definite operator with unit trace. All Hilbert spaces we consider, with the exception of those appearing in Section~\ref{sec:gaussian}, are assumed to be finite dimensional.

\begin{Def}[(Quantum umlaut information)]\label{def:lautum}
Given a bipartite state $\rho_{AB}\in\mathcal{D}(\mathcal{H}_A\otimes\mathcal{H}_B)$, the {\rm quantum umlaut information} is defined as
\bb\label{eq:def_umlaut}
U(A;B)_{\rho}&\coloneqq \min_{\sigma_B}D(\rho_A\otimes\sigma_B\|\rho_{AB}),
\ee
where $\sigma_B\in \mathcal{D}(\mathcal{H}_B)$ and $\rho_A=\Tr_B[\rho_{AB}]$.
\end{Def}

The straightforward generalisation to the quantum setting of the Palomar--Verd\'u lautum information would be
    \bb
        L(A\!:\!B)_{\rho}&\coloneqq D(\rho_A\otimes\rho_B\|\rho_{AB})
    \ee
where $\rho_A=\Tr_B[\rho_{AB}]$ and $\rho_B=\Tr_A[\rho_{AB}]$. It is clear that $U\leq L$. Furthermore, while $L$ is symmetric under exchange of $A$ and $B$, in general $U$ is not: because of this, for the latter we use the semicolon instead of the colon to separate the two systems. Some elementary properties of the quantum umlaut information are collected in the proposition below.

\begin{lemma} \label{lem:quantum_lautum_prop}
The quantum umlaut information satisfies the following properties:
\begin{enumerate}[(1)]
\item Positive definiteness: $U(A;B)_\rho \geq 0$ for all states $\rho_{AB}$, with equality if and only if $\rho_{AB} = \rho_A \otimes \rho_B$.
\item Boundedness: $U(A;B)_\rho < \infty$ if and only if there exists a pure state $\psi_B = \ketbra{\psi}_B$ such that $\supp(\rho_A \otimes \psi_B) \subseteq \supp(\rho_{AB})$. 
\item Data-processing inequality: for all states $\rho_{AB}$ and all pairs of quantum channels $\Lambda_{A\to A'}$ and $\Lambda_{B\to B'}$, setting $\omega_{A'B'} \coloneqq (\Lambda_{A\to A'}\otimes\Lambda_{B\to B'})(\rho_{AB})$ we have
\bb\label{eq:data_proc_u}
U(A';B')_{\omega}\leq U(A;B)_{\rho}\, .
\ee
\end{enumerate}    
\end{lemma}

\begin{proof}
Non-negativity follows from the definition~\eqref{eq:def_umlaut} via the non-negativity of the Umegaki relative entropy. Since $D(\rho\|\sigma)=0$ if and only if $\rho=\sigma$, we also deduce that $U(A;B)_\rho=0$ if and only if $\rho_{AB} = \rho_A\otimes \sigma_B$ for some state $\sigma_B$ on $B$. Taking the partial trace on $A$ of this equality shows that in fact $\rho_B = \sigma_B$, completing the proof of~(1). Property~(2) follows from the observation that $D(\rho\|\sigma) < \infty$ if and only if $\supp(\rho) \subseteq \supp(\sigma)$, once one notes that $\supp(\rho_A \otimes \sigma_B) \subseteq \supp(\rho_{AB})$ for some $\sigma_B$ if and only if this happens for some pure state $\sigma_B = \ketbra{\psi}_B$, namely, any $\ket{\psi}_B\in \supp(\sigma_B)$. Finally, for~(3) we write
\bb
U(A';B')_{(\Lambda\otimes\Lambda)(\rho)}
&=\min_{\sigma_{B'}}D\big(\Lambda_{A\to A'}(\rho_A)\otimes\sigma_{B\to B'}\,\big\|\,(\Lambda_{A\to A'}\otimes\Lambda_{B\to B'})(\rho_{AB})\big)\\
&\leqt{(i)}\min_{\sigma_B}D\big(\Lambda_{A\to A'}(\rho_A)\otimes\Lambda_{B\to B'}(\sigma_B)\,\big\|\,(\Lambda_{A\to A'}\otimes\Lambda_{B\to B'})(\rho_{AB})\big)\\
&\leqt{(ii)}\min_{\sigma_B}D\big(\rho_A\otimes\sigma_B\,\big\|\,\rho_{AB}\big)\\
&=U(A;B)_{\rho},
\ee
where in~(i) we have restricted the minimisation to the set of states $\sigma_{B'}$ which can be written as $\sigma_{B'}=\Lambda_{B\to B'}(\sigma_B)$ for an arbitrary $\sigma_B\in\mathcal{D}(\mathcal{H}_B)$, and in~(ii) we have leveraged the data processing inequality for the quantum relative entropy~\cite{Lindblad-monotonicity}.
\end{proof}

As a simple consequence of Lemma~\ref{lem:quantum_lautum_prop}(2), we observe the (elementary) fact that $U(A;B)_\rho < \infty$ holds for all full-rank states $\rho_{AB}$, and in particular almost everywhere in a measure-theoretic sense.

Definition~\ref{def:lautum} can be generalised by means of the Petz--R\'{e}nyi $\alpha$-relative entropies. The quantity below has been also recently considered in~\cite[Eq.~(E3)]{ji_2024}.

\begin{Def}[(Petz--R\'{e}nyi $\alpha$-umlaut information)]\label{def:petz_lautum} 
Let $\alpha\in(0,1)\cup(1,\infty)$. Given a bipartite state $\rho_{AB}\in\mathcal{D}(\mathcal{H}_A\otimes\mathcal{H}_B)$, we define its {\rm Petz--Rényi} $\alpha$-{\rm umlaut information} as
\bb
U_\alpha(A;B)_{\rho}&\coloneqq \min_{\sigma_B}D_\alpha(\rho_A\otimes\sigma_B\|\rho_{AB})
\ee
where $\sigma_B\in \mathcal{D}(\mathcal{H}_B)$ and $\rho_A=\Tr_B[\rho_{AB}]$ and $D_\alpha$ is the Petz--Rényi $\alpha$-relative entropy:
\bb
    D_{\alpha}(\rho\|\sigma)\coloneqq\frac{1}{\alpha-1}\log\Tr\left[\rho^\alpha\sigma^{1-\alpha}\right].
\ee
\end{Def}

The data processing inequality for the umlaut information in eq.~\eqref{eq:data_proc_u} can be generalised to the Petz--Rényi $\alpha$-umlaut information for $\alpha\in(0,1)\cup(1,2)$ with the same strategy leveraging the data processing inequality for the Petz--Rényi $\alpha$-relative entropy~\cite{Tomamichel_2009}.

\begin{lemma}\label{lem:alpha_to_1} Given a bipartite state $\rho_{AB}\in\mathcal{D}(\mathcal{H}_A\otimes\mathcal{H}_B)$, it holds that
    \bb
        \limsup_{\alpha\to 1^-}U_\alpha(A;B)_{\rho}=U(A;B)_{\rho}.
    \ee
\end{lemma}

\begin{proof}
    Since $\alpha\mapsto D_\alpha(\rho\|\sigma)$ is a monotonically increasing function, we can write
    \bb
        \limsup_{\alpha\to 1^-}U_\alpha(A;B)_{\rho}=\sup_{\alpha<1}\min_{\sigma_B}D_\alpha(\rho_A\otimes\sigma_B\|\rho_{AB})
    \ee
    By the Mosonyi--Hiai minimax theorem~\cite[Corollary A2]{MosonyiHiai}, we can rewrite
    \bb
        \limsup_{\alpha\to 1^-}U_\alpha(A;B)_{\rho}=\min_{\sigma_B}\sup_{\alpha<1}D_\alpha(\rho_A\otimes\sigma_B\|\rho_{AB})\eqt{(i)}\min_{\sigma_B}D(\rho_A\otimes\sigma_B\|\rho_{AB})=U(A;B)_{\rho},
    \ee
    where in (i) we have used that the Petz--Rényi relative entropies converge to the Umegaki quantum relative entropy as $\alpha\to 1^-$.
\end{proof}


\subsection{Closed form and optimal marginal}

It turns out that the variational problem defining the quantum umlaut information has an explicit solution. We can therefore provide an analytical formula for $U$. As an immediate corollary, we can prove that $U$ is additive, a property which was not evident looking just at the definition of $U$. In order to do this, we need to recall an elementary inequality and introduce a definition which will be often used in the rest of the paper. 

\begin{lemma}[(Gibbs variational principle)] \label{lemma:Gibbs}
Let $H$ be a self-adjoint operator. For any quantum state $\gamma$:
    \begin{equation}
       -S(\gamma) +\Tr[\gamma H] \geq - \log \Tr[\exp[-H]],
    \end{equation}
    with equality if and only if
    \begin{equation}
        \gamma=\frac{\exp[-H]}{\Tr[\exp[-H]]}.
    \end{equation}
\end{lemma}

\begin{Def}[(Umlaut-marginal)]
Given a bipartite state $\rho_{AB}\in\mathcal{D}(\mathcal{H}_A\otimes\mathcal{H}_B)$, the umlaut-marginal $\uml{\rho}_{B}\in\mathcal{D}(\mathcal{H}_B)$ of the state $\rho_{AB}$ on the system $B$ is defined as
    \bb
        \uml{\rho}_{B}\coloneqq \frac{\exp[\Tr_A[(\rho_A\otimes\id_B)\log\rho_{AB}]]}{\Tr\left[\exp[\Tr_A[(\rho_A\otimes\id_B)\log\rho_{AB}]]\right]}.
    \ee
\end{Def}

In the following theorem we are going to show that the umlaut-marginal of a bipartite state is the optimiser of the variational problem appearing in the definition of $U$.

\begin{prop}[(A closed-form expression for the quantum umlaut information)]\label{prop:closed_lautum}
Given a bipartite state $\rho_{AB}\in\mathcal{D}(\mathcal{H}_A\otimes\mathcal{H}_B)$, its quantum umlaut information can be written as
    \bb\label{eq:formula}
        U(A;B)_{\rho}=-S(\rho_A)-\log\Tr\left[\exp\left(\Tr_A[(\rho_A\otimes\id_B)\log\rho_{AB}]\right)\right].
    \ee
    In particular, the unique minimiser in the definition of $U$ is the umlaut-marginal $\uml{\rho}_{B}$ of the state $\rho_{AB}$ on the system $B$, i.e.
    \bb
        U(A;B)_{\rho}=D(\rho_A\otimes\uml{\rho}_{B}\|\rho_{AB}).
    \ee
\end{prop}

\begin{proof}
This proof generalises the one provided in the classical case~\cite[Proposition 7]{Filippo25}. Let us start by rewriting
\bb
    \Tr[(\rho_A\otimes\sigma_B)\log\rho_{AB}]&=\Tr[(\id_A\otimes\sigma_B)(\rho_A\otimes\id_B)\log\rho_{AB}]\\
    &=\Tr\left[\sigma_B\Tr_A[(\rho_A\otimes\id_B)\log\rho_{AB}]\right]\\
    &=\Tr\left[\sigma_B X_B\right],
\ee
where we have defined $X_B\coloneq \Tr_A[(\rho_A\otimes\id_B)\log\rho_{AB}]$.
Now, let us compute
\bb
    U(A;B)_{\rho}&= \inf_{\sigma_B}D(\rho_A\otimes\sigma_B\|\rho_{AB})\\
    &=-S(\rho_A)+\inf_{\sigma_B}\left\{ \Tr[\sigma_B\log\sigma_B]-\Tr[\sigma_B X_B] \right\}\\
    &=-S(\rho_A)-S(\uml{\rho}_{B})-\Tr[\uml{\rho}_{B} X_B]\\
    &=-S(\rho_A)-\log \Tr[\exp[X_B]],
\ee
where we have used the Gibbs variational principle (Lemma~\ref{lemma:Gibbs}) and introduced 
\bb
    \uml{\rho}_{B}\coloneq\frac{\exp[X_B]}{\Tr[\exp[X_B]]}.
\ee
The uniqueness is ensured by the statement of the Gibbs variational principle. This concludes the proof.
\end{proof}

Using the explicit formula for the quantum umlaut information, it is elementary to prove that it is additive.

\begin{cor}[(Additivity of  the quantum umlaut information)]\label{cor:additivity}
Given two bipartite states $\rho_{AB}\in\mathcal{D}(\mathcal{H}_A\otimes\mathcal{H}_B)$ and $\sigma_{A'B'}\in\mathcal{D}(\mathcal{H}_{A'}\otimes\mathcal{H}_{B'})$, it holds that
    \begin{equation}
        U(AA';BB')_{\rho\otimes\sigma}=U(A;B)_{\rho}+U(A';B')_{\sigma}.
    \end{equation}
\end{cor}

\begin{proof}
Similarly to~\cite[Corollary 8]{Filippo25}, knowing the closed-form expression provided by Proposition~\ref{prop:closed_lautum}, it is straightforward to directly show that the umlaut information is additive: indeed, it is sufficient to use the additivity of the logarithm under tensor products in both the terms appearing in~\eqref{eq:formula}.
As an alternative proof, it is easy to see that the umlaut-marginal $\xi_{BB'}$ of ${\rho_{AB}\otimes\sigma_{A'B'}}$ factorises, once again due to the additivity of the logarithm:
\bb
    \xi_{BB'}&=\frac{\exp\left[\Tr_{AA'}[(\rho_A\otimes\sigma_{A'})\log(\rho_{AB}\otimes\sigma_{A'B'})]\right]}{\Tr\left[\exp\left[\Tr_{AA'}[(\rho_A\otimes\sigma_{A'})\log(\rho_{AB}\otimes\sigma_{A'B'})]\right]\right]}\\
    &=\frac{\exp\left[\Tr_{AA'}[(\rho_A\otimes\sigma_{A'})(\log\rho_{AB}+\log\sigma_{A'B'})]\right]}{\Tr\left[\exp\left[\Tr_{AA'}[(\rho_A\otimes\sigma_{A'})(\log\rho_{AB}+\log\sigma_{A'B'})]\right]\right]}\\
    &=\uml{\rho}_{B}\otimes\uml{\sigma}_{B'},
\ee
where $\uml{\rho}_{B}$ and $\uml{\sigma}_{B'}$ are the umlaut-marginals of $\rho_{AB}$ and $\sigma_{A'B'}$. Therefore,
\bb
    U(AA';BB')_{\rho\otimes\sigma}&=D(\rho_A\otimes\sigma_B\otimes\xi_{BB'}\|\rho_{AB}\otimes\sigma_{A'B'})\\
    &=D(\rho_A\otimes\sigma_B\otimes\uml{\rho}_{B}\otimes\uml{\sigma}_{B'}\|\rho_{AB}\otimes\sigma_{A'B'})\\
    &=D(\rho_A\otimes\uml{\rho}_{B}\|\rho_{AB})+D(\sigma_{A'}\otimes\uml{\sigma}_{B'}\|\sigma_{A'B'})\\
    &=U(A;B)_{\rho}+U(A';B')_{\sigma}
\ee
which proves the claim of additivity.
\end{proof}

Also for the Petz--Rényi $\alpha$-umlaut information one can prove a closed-form expression which turns out to be additive. The existence of such a formula was established in~\cite[Lemma~35]{ji_2024}, and one can deduce additivity from there by direct inspection. 

\begin{prop}[{\cite[Lemma~38]{ji_2024}}]
\label{prop:closed_alpha}
    For any $\alpha\in(0,1)$, given a bipartite state $\rho_{AB}\in\mathcal{D}(\mathcal{H}_A\otimes\mathcal{H}_B)$, it holds that
    \bb
        U_\alpha(A;B)_{\rho}= D_\alpha\left(\rho_A\otimes\uml{\rho}^{(\alpha)}_B\,\middle\|\,\rho_{AB}\right) = -\log \left\|\Tr_A\left[\rho_A^\alpha\rho_{AB}^{1-\alpha}\right]\right\|_{\frac{1}{1-\alpha}}^{\frac{1}{1-\alpha}}\,,
    \ee
    where $\rho_A=\Tr_B[\rho_{AB}]$, the Schatten $p$-norm is defined for $p\geq 1$ by the formula $\|X\|_p \coloneqq \left( \Tr\left[ (X^\dag X)^{p/2}\right] \right)^{1/p}$, and
    \bb
        \uml{\rho}^{(\alpha)}_B\coloneqq \frac{\left(\Tr_A\left[\rho_A^\alpha\rho_{AB}^{1-\alpha}\right]\right)^{\frac{1}{1-\alpha}}}{\Tr\left[\left(\Tr_A\left[\rho_A^\alpha\rho_{AB}^{1-\alpha}\right]\right)^{\frac{1}{1-\alpha}}\right]}.
    \ee
    Furthermore, given two bipartite states $\rho_{AB}\in\mathcal{D}(\mathcal{H}_A\otimes\mathcal{H}_B)$ and $\sigma_{A'B'}\in\mathcal{D}(\mathcal{H}_{A'}\otimes\mathcal{H}_{B'})$, it holds that
    \begin{equation}
        U_\alpha(AA';BB')_{\rho\otimes\sigma}=U_\alpha(A;B)_{\rho}+U_\alpha(A';B')_{\sigma}.
    \end{equation}
\end{prop}

For completeness, we provide a self-contained proof of Proposition~\ref{prop:closed_alpha} in Appendix~\ref{app:proof_closed_alpha}.


\subsection{Operational interpretation in composite quantum hypothesis testing}

Let $\mathcal{H}$ be a Hilbert space and let $\mathcal{D}(\mathcal{H})$ be the set of quantum states of $\mathcal{H}$. Let us consider two fixed states $\rho$ and $\sigma$ in $\mathcal{D}(\mathcal{H})$. Let us suppose to have a device that produces $n$ copies of an unknown states $\tau\in \mathcal{D}(\mathcal{H})$ according to one of two possible hypotheses: if the null hypothesis $H_0$ holds, then $\tau=\rho$; if the alternative hypothesis $H_1$ holds, then $\tau =\sigma$.  In the task of simple asymmetric hypothesis testing, one performs a binary measurement on the available copies of $\tau$ in order to decide whether $H_0$ or $H_1$ holds. The asymmetry stems from the role of the two hypotheses: for instance, we may want to detect $H_1$ as efficiently as possible when it is the case, provided that under $H_0$ the probability of a false alarm is under a fixed threshold, or vice-versa. Let us call $\{M_n,\id-M_n\}$ the POVM representing the measurement chosen to guess the right hypothesis. In particular,
\begin{itemize}
    \item $\Tr[M_n\tau^{\otimes n}]$ is the probability of guessing $H_0$ given $n$ copies of $\tau$,
    \item $\Tr[(\id-M_n)\tau^{\otimes n}]$ is the probability of guessing $H_1$ given $n$ copies of $\tau$.
\end{itemize}
Two kinds of error could occur:
\begin{itemize}
    \item $H_0$ holds, but the outcome of the measurement is $H_1$ (error of type I, or false positive), which happens with probability $\Tr[(\id-M_n)\rho^{\otimes n}]$
    \item $H_1$ holds, but the outcome of the measurement is $H_0$ (error of type II, or missed detection), which happens with probability $\Tr[M_n\sigma^{\otimes n}]$
\end{itemize}
Given any $\epsilon\in (0,1)$. The hypothesis testing relative entropy is defined as
\bb
    D^{\epsilon}_H(\rho\|\sigma)\coloneqq -\log\min\left\{\Tr[M\sigma]\,:\,0\leq M\leq \id,\, \Tr[(\id-M)\rho]\leq \epsilon\right\},
    \label{D_H}
\ee
where $\{M,\id-M\}$ is a POVM for the states of $\mathcal{H}$.
The type-II error exponent will asymptotically decay as~\cite{Hiai1991}
\bb
    \lim_{\epsilon\to 0}\liminf_{n\to \infty}\frac{1}{n}D^\epsilon_H(\rho^{\otimes n}\|\sigma^{\otimes n}) = D(\rho\|\sigma),
\ee
which is known as the quantum Stein's lemma, with the exponent on the left-hand side often called Stein's exponent. 
This means that, when the type-I error is below a fixed threshold, the optimal type-II error asymptotically decays as $\exp(-nD(\rho\|\sigma))$, i.e.\ the distinguishability between $\rho$ and $\sigma$ in this asymmetric setting is quantified by their relative entropy. On the other hand, when we require the type-II error to be below a fixed threshold and we minimise the type-I error, we consider the hypothesis testing relative entropy with the role of the quantum states reversed, so that
\bb
    D^{\epsilon}_H(\sigma\|\rho)= -\log\min\left\{\Tr[(\id-M)\rho]\,:\,0\leq M\leq \id,\, \Tr[M\sigma]\leq \epsilon\right\},
\ee
is the quantity determining the exponential decay of the type-I error.

\begin{figure}
    \centering
    \includegraphics[width=0.45\linewidth]{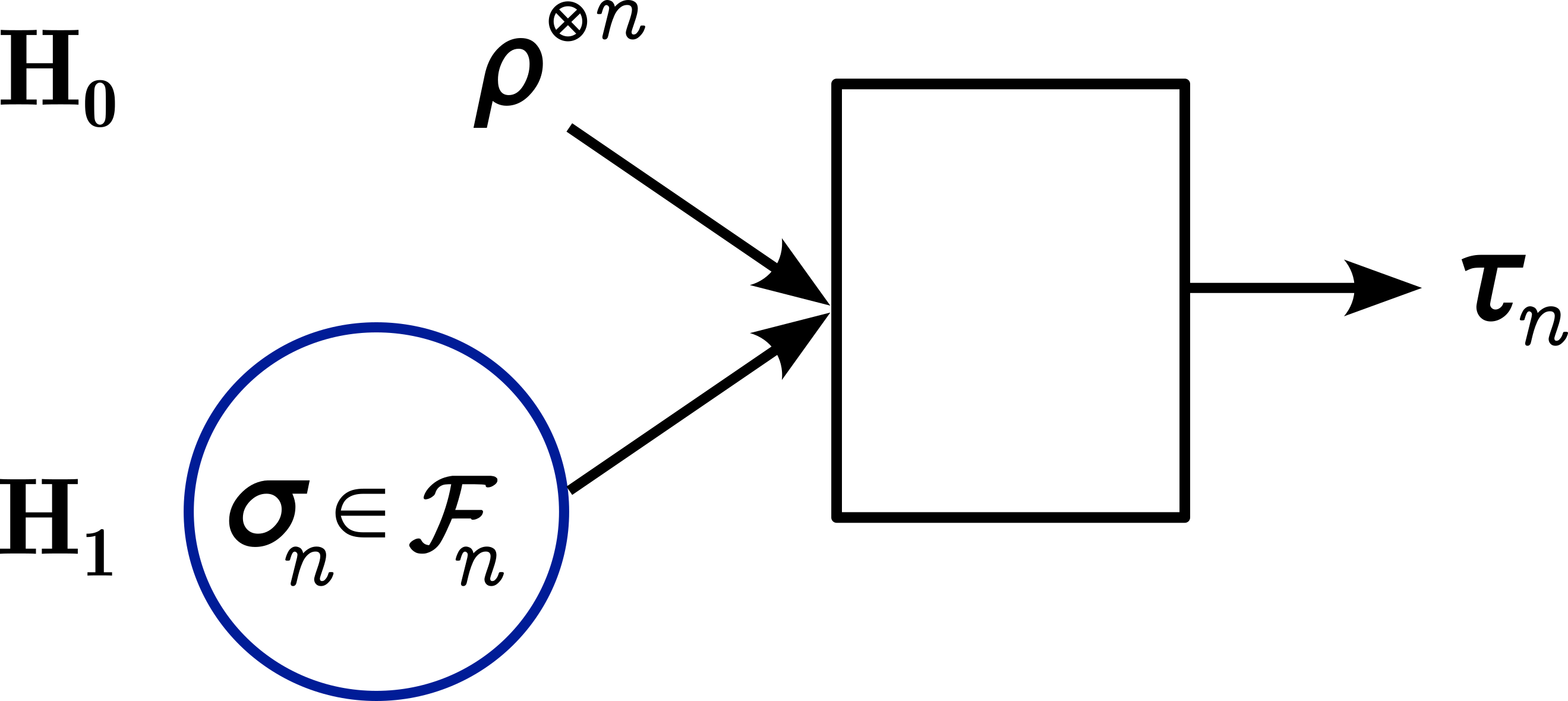}
    \caption{Composite hypothesis testing}
    \label{fig:composite_hp}
\end{figure}

In this paper, we are interested in the following extension of the hypothesis testing setting, where the state $\tau_n\in\mathcal{D}(\mathcal{H}^{\otimes n})$ in consideration is not necessarily an i.i.d.\ state. More precisely, let us suppose that the alternative hypothesis $H_1$ is composite, i.e.\ the output is a state $\sigma_n$ belonging to a family $\mathcal{F}_n\subseteq \mathcal{D}(\mathcal{H}^{\otimes n})$. On the contrary, according to the null hypothesis $H_0$, the output of the device are $n$ copies of a fixed (and known) state $\rho\in\mathcal{D}(\mathcal{H})$, as in Figure~\ref{fig:composite_hp}.  We will call $\mathcal{F}=(\mathcal{F}_n)_{n\in\mathbb{N}}$ the sequence of families related to the alternative hypothesis. 
We will be interested in the optimal decay rate of the type-I error probability (under a fixed constraint on the type-II error probability), which --- following the naming conventions of~\cite{Bjelakovic2005,Noetzel_2014}, recently adapted also in~\cite{lami2024asymptotic} (but see the discussion in~\cite{Hayashi24}) --- we call the (generalised) quantum Sanov exponent $\text{Sanov}(\rho\|\mathcal{F})$. Precisely,
\bb\label{eq:Sanov_def}
    \text{Sanov}(\rho\|\mathcal{F})\coloneqq \lim_{\epsilon\to 0}\liminf_{n\to \infty}\frac{1}{n}D^\epsilon_H(\mathcal{F}_n\|\rho^{\otimes n}),
\ee
where
\bb
    D^\epsilon_H(\mathcal{F}_n\|\rho^{\otimes n})\coloneqq \min_{\sigma_n\in\mathcal{F}_n}D^\epsilon_H(\sigma_n\|\rho^{\otimes n}).
\ee
 Were the roles of the null and alternative hypotheses to be reversed (i.e., if we focused on the decay rate of the type-II error), this would instead correspond to the generalised quantum Stein's exponent~\cite{Brandao2010}.

We can also consider the strong converse Sanov exponent, given by a modified version of~\eqref{eq:Sanov_def} in which we allow the type II error to be arbitrarily close to $1$, instead of arbitrarily small, and we replace the limit inferior in $n$ with a limit superior:
\bb\label{eq:stein_sc_def}
    \mathrm{Sanov}^\dag(\rho\|\FF)\coloneqq \lim_{\epsilon\to 1^-}\limsup_{n\to \infty}\frac{1}{n} D^\epsilon_H(\FF_n\|\rho^{\otimes n}) .
\ee
Clearly, in general $\mathrm{Sanov}^\dag(\rho\|\FF) \geq \mathrm{Sanov}(\rho\|\FF)$. In many interesting cases, however, equality holds; equivalently, $\lim_{n\to \infty}\frac{1}{n} D^\epsilon_H(\FF_n\|\rho^{\otimes n})$ exists for all $\e\in (0,1)$, and its value is independent of $\e$. This holds, for example, when both hypotheses are simple and i.i.d., in which case we have indeed~\cite{Hiai1991,Ogawa2000}
\bb
\label{eq:strong_converse_property_iid}
D(\sigma\|\rho) = \mathrm{Sanov}(\rho \| \sigma) = \mathrm{Sanov}^\dagger(\rho\|\sigma ) = \lim_{n\to \infty}\frac{1}{n}D^\epsilon_H(\sigma^{\times n}\|\rho^{\otimes n})\qquad \forall\ \e\in (0,1)\, .
\ee

The operational interpretation of the umlaut information can be found in the following composite hypothesis testing problem. Given a fixed bipartite state $\rho_{AB}\in\mathcal{D}(\mathcal{H}_A\otimes\mathcal{H}_B)$, for any $n\geq 1$ the device produces a state $\tau_{A^nB^n}\in\mathcal{D}\big((\mathcal{H}_A\otimes\mathcal{H}_B)^{\otimes n}\big)$ according to one of the following hypotheses:\bb
    &H_0\,:\, \tau_{A^nB^n}=(\rho_{AB})^{\otimes n}\\
    &H_1\,:\, \tau_{A^nB^n}=\rho_{A}^{\otimes n}\otimes \sigma_{B^n}\quad \text{where}\quad\sigma_{B^n}\in \mathcal{D}(\mathcal{H}_B^{\otimes n}).
\ee
The sequence of families defining the alternative hypothesis is
\bb
    \mathcal{F}^{\rho_A}_n\coloneqq &\left\{\sigma_{A^nB^n} \in \mathcal{D}\big((\mathcal{H}_A\otimes\mathcal{H}_B)^{\otimes n}\big)\,:\,
    \sigma_{A^nB^n}=\rho_{A}^{\otimes n}\otimes\sigma_{B^n}, \,\sigma_{B^n}\in\mathcal{D}(\mathcal{H}_B^{\otimes n})\right\}.
\ee
The Stein setting of this problem was considered in~\cite{tomamichel_2018,hayashi_2016-1}.
Introducing the corresponding Sanov exponent as
\bb
    \text{Sanov}(\rho_{AB}\|\mathcal{F}^{\rho_A})\coloneqq \lim_{\epsilon\to 0}\liminf_{n\to \infty}\frac{1}{n}D^\epsilon_H(\mathcal{F}^{\rho_A}_n\|\rho_{AB}^{\otimes n}),
\ee
the following theorem identifies such exponent with the umlaut information of $\rho_{AB}$.

\begin{boxed}{}
\begin{thm}[(Operational interpretation of the quantum umlaut information)] \label{thm:interp_L2} 
Let $\rho_{AB}$ be a bipartite state in $\mathcal{D}(\mathcal{H}_A\otimes\mathcal{H}_B)$. Then,
\bb
U(A;B)_\rho=\lim_{n\to \infty}\frac{1}{n} D^\epsilon_H\big(\FF^{\rho_{A}}_n\,\big\|\, \rho_{AB}^{\otimes n}\big) \qquad \forall\ \e\in (0,1)\, ;
\ee
equivalently,
\bb
U(A;B)_\rho=\mathrm{Sanov}\big(\rho_{AB}\,\big\|\,\FF^{\rho_A}\big) = \mathrm{Sanov}^\dag\big(\rho_{AB}\,\big\|\,\FF^{\rho_A}\big)\, .
\ee
where $\rho_{A}=\Tr_B[\rho_{AB}]$ and $\mathcal{F}^{\rho_A}$ is the family given by $\mathcal{F}_n^{\rho_A} = \left\{\rho_{A}^{\otimes n}\otimes \sigma_{B^n}\,:\,\sigma_{B^n}\in \mathcal{D}(\mathcal{H}_B^{\otimes n})\right\}$.
\end{thm}
\end{boxed}

\begin{proof} 
    Let $\alpha\in(0,1)$. It is known that~\cite{Hayashi_2007,Audenaert2012_quantum}
    \bb\label{eq:lowerDalpha}
        D^\epsilon_H(\rho\|\sigma)\geq D_\alpha(\rho\|\sigma)+\frac{\alpha}{1-\alpha}\log\frac{1}{\epsilon}.
    \ee
    We can use that inequality as follows:
    \bb
         \text{Sanov}(\rho\|\mathcal{F}^{\rho_A})&= \lim_{\epsilon\to 0^+}\liminf_{n\to \infty}\frac{1}{n}\min_{\sigma_{B^n}}D^\epsilon_H(\rho_{A}^{\otimes n}\otimes\sigma_{B^n}\|\rho_{AB}^{\otimes n})\\
         &\geq\lim_{\epsilon\to 0^+}\liminf_{n\to \infty}\frac{1}{n}\min_{\sigma_{B^n}}\left(D_\alpha(\rho_{A}^{\otimes n}\otimes\sigma_{B^n}\|\rho_{AB}^{\otimes n})+\frac{\alpha}{1-\alpha}\log\frac{1}{\epsilon}\right)\\
         &=\liminf_{n\to \infty}\frac{1}{n}U_\alpha(A^n;B^n)_{\rho^{\otimes n}}\\
         &\eqt{(i)}U_\alpha(A;B)_{\rho},
    \ee
    where in (i) we have leveraged the additivity of $U_\alpha$ (as in Proposition~\ref{prop:closed_alpha}). In particular,
    \bb
         \text{Sanov}(\rho\|\mathcal{F}^{\rho_A})&\geq \limsup_{\alpha\to 1^-}U_\alpha(A;B)_{\rho}\eqt{(ii)}U(A;B)_{\rho},
    \ee
    where in (ii) we have used Lemma~\ref{lem:alpha_to_1}. For the upper bound we consider the ansatz $\sigma_{B^n}=\sigma_B^{\otimes n}$, where $\sigma_B$ is an arbitrary fixed state. Then
    \bb
        \text{Sanov}^\dagger(\rho\|\mathcal{F}^{\rho_A})&= \lim_{\epsilon\to 1^-}\liminf_{n\to \infty}\frac{1}{n}\min_{\sigma_{B^n}}D^\epsilon_H(\rho_{A}^{\otimes n}\otimes\sigma_{B^n}\|\rho_{AB}^{\otimes n})\\
        &\leq \lim_{\epsilon\to 1^-}\liminf_{n\to \infty}\frac{1}{n}D^\epsilon_H\left((\rho_{A}\otimes\sigma_B)^{\otimes n}\|\rho_{AB}^{\otimes n}\right)\\
        &\eqt{(iii)}D(\rho_A\otimes\sigma_B\|\rho_{AB}),
    \ee
    where (iii) follows from the strong converse to the quantum Stein's lemma~\cite{Hiai1991,Ogawa2000}. Minimising over $\sigma_B\in\mathcal{D}(\mathcal{H}_B)$ yields
\bb
\text{Sanov}^\dag(\rho\|\mathcal{F}^{\rho_A}) \leq \min_{\sigma_B}D(\rho_A\otimes\sigma_B\|\rho_{AB})=U(A;B)_{\rho}.
    \ee
    This concludes the proof.
\end{proof}

As an alternative proof of Theorem~\ref{thm:interp_L2}, we could verify that the specific testing problem considered in this section is part of a general class of composite asymmetric hypothesis testing problems that was solved in~\cite{lami2024asymptotic}.


\subsection{Umlaut information of Gaussian states}\label{sec:gaussian}

Although in the rest of the paper we restrict ourselves to finite-dimensional spaces, let us take a brief jaunt into infinite-dimensional spaces to showcase an illustrative example of how our definitions can be generalised to Gaussian states, which have a wide range of application in quantum communication, computing, sensing, and are commonly used in quantum optics  laboratories~\cite{BUCCO}.

In this subsection we are thus going to consider the natural extension of Definition~\ref{def:lautum} to continuous variable system,\footnote{We refer to~\cite{Lindblad1973} for a formal definition of the Umegaki relative entropy for pairs of states on a separable Hilbert space.}. We first briefly recall the basic notions about continuous variable systems and Gaussian state. Let $m$ be a positive integer. An $m$-mode bosonic system is a quantum system with Hilbert space $\mathcal{H}_m\coloneqq L^2(\mathbb{R}^m)$. The momentum and position operators acting on the mode $j\in\{1,\dots,n\}$, respectively denoted as $x_j$ and $p_j$, are called quadratures and satisfy the canonical commutation relation $[x_j,p_k]=i\delta_{jk}$. The quadrature vector $R$ is defined as $R=(x_1,p_1,\dots,x_m,p_m)^\intercal$. An $m$-mode state $\rho$ is called a Gaussian state if it can be written as a a Gibbs state of a quadratic Hamiltonian of the form $\frac{1}{2}(R-m)^\intercal H (R-m)$ for some vector $m\in\mathbb{R}^{2m}$ and some symmetric matrix $H\in\mathbb{R}^{2n\times 2n}$, namely $\rho\propto\exp(-\frac{1}{2}(R-m)^\intercal H (R-m))$, or as a limit of such states. Any Gaussian state $\rho$ is fully characterised by its mean vector $m=\Tr[R\rho]$ and its covariance matrix $V=\Tr[\{R-m,(R-m)^\intercal\}\,\rho]$, where $\{\,\cdot\,,\,\cdot\,\}$ is the anticommutator. We will denote by $\rho(m,V)$ the Gaussian state uniquely identified by the mean vector $m$ and the covariance matrix $V$. Given $k\in\mathbb{R}^{2m}$, we call $D_k=\exp(-iR^\intercal k)$ the displacement operator. The characteristic function of a state $\rho$ is defined as $\chi_{\rho}(k)=\Tr\left[\rho D_k\right]$. In particular, the form of the characteristic function of a Gaussian state is known and it involves the mean vector and the covariance matrix: $\chi_{\rho(m,V)}(k)=\exp\left(-\frac{1}{2}k^\intercal V k-ik^\intercal m\right)$. The covariance matrix $V$ can be directly written in terms of the matrix $H$ according to the identity $V(H)=\coth\left(\frac{i\Omega H}{2}\right)i\Omega$ (see, e.g.~\cite[Problem 3.2]{BUCCO}), where $\Omega$ is the symplectic form: 
\bb
\Omega=\bigoplus_{k=1}^m\begin{pmatrix} &1\\ -1 &\end{pmatrix}.
\ee
Now we have all the notions in order to state and prove the following proposition.

\begin{prop}[(Umlaut-marginal of a Gaussian state)] Let
\bb
    \rho_{AB}=\frac{\exp\left[-\frac{1}{2}(R-m)^\intercal H(R-m)\right]}{\Tr\left[\exp\left[-\frac{1}{2}(R-m)^TH(R-m)\right]\right]}
\ee
be a bipartite Gaussian state with
\bb
    R=(R_A,R_B),\qquad m=(m_A,m_B) \qquad \text{and}\qquad  H=\begin{pmatrix}H_{AA} & H_{AB} \\ H_{AB} & H_{BB}\end{pmatrix}.
\ee
Then the umlaut-marginal of $\rho_{AB}$ on the system $B$ is given by
\bb
    \uml{\rho}_{B}=\frac{\exp\left[-\frac{1}{2}(R_B-m_B)^\intercal H_{BB}(R_B-m_B)\right]}{\Tr\left[\exp\left[-\frac{1}{2}(R_B-m_B)^\intercal H_{BB}(R_B-m_B)\right]\right]}.
\ee
\end{prop}

\begin{rem}
    It is interesting to notice that, in the Gaussian state, the reduced state of $\rho_{AB}=\rho(m,V(H))$ is given by the reduced covariance matrix, i.e.\ $\rho_B=\rho(m_B,V_{BB}(H))$, while the umlaut-marginal is given by the reduced Hamiltonian, i.e.\ $\uml{\rho}_{B}=\rho(m_B,V_B(H_{BB}))$.
\end{rem}

\begin{proof}
It is known that the mean of this state is $m$ and the covariance matrix $V$ is given by
\bb
    V(H)=\begin{pmatrix}V_{AA}(H) & V_{AB}(H) \\ V_{AB}(H) & V_{BB}(H)\end{pmatrix}=\coth\left(\frac{i\Omega_{AB} H}{2}\right)i\Omega_{AB},
\ee
where $\Omega_{AB}$ is the symplectic form on the system $AB$.
Now we compute the reduced state $\rho_A$ via its characteristic function, defining $k\coloneqq(k_A,k_B)$.
\bb
    \chi_{\rho_A}(k_A)&=\Tr_A\left[\rho_A\exp\left[-iR_A^\intercal k_A\right]\right]\\
    &=\Tr\left[\rho_{AB}\exp\left[-iR_A^\intercal k_A\right]\right]\Big|_{k_B=0}\\
    &=\chi_{\rho_{AB}}(k_A, k_B=0)\\
    &=\exp\left[-\frac{1}{2}k^\intercal V(H)k-ik^Tm\right]\bigg|_{k_B=0}\\
    &=\exp\left[-\frac{1}{2}k_A^TV_{AA}(H)k_A-ik_A^\intercal m_A\right].
\ee
Therefore $\rho_A$ is the Gaussian state with mean $m_A$ and covariance $V_{AA}(H)$. Defining
\bb
    h_A\coloneqq 2i\Omega_A\;\text{arc\,coth}\left(V_{AA}(H) i\Omega_A\right),
\ee
where $\Omega_A$ is the symplectic form on the system $A$, we have
\bb
    \rho_{A}=\frac{\exp\left[-\frac{1}{2}(R_A-m_A)^\intercal h_A(R_A-m_A)\right]}{\Tr\left[\exp\left[-\frac{1}{2}(R_A-m_A)^\intercal h_A(R_A-m_A)\right]\right]}.
\ee
Now we compute
\bb
        \uml{\rho}_{B}\coloneqq \frac{\exp\left[\Tr_A[(\rho_A\otimes\id_B)\log\rho_{AB}]\right]}{\Tr\left[\exp\left[\Tr_A[(\rho_A\otimes\id_B)\log\rho_{AB}]\right]\right]}.
\ee
In particular,
\bb
    \Tr_A[(\rho_A\otimes\id_B)\log\rho_{AB}]&=-\frac{1}{2}\Tr_A[(\rho_A\otimes\id_B)(R-m)H(R-m)]+\text{const.}\\
    &=-\frac{1}{2}\Tr_A\left[\rho_A\otimes \big((R_B-m_B)H_{BB}(R_B-m_B)\big)\right]\\
    &\quad-\Tr_A\left[(\rho_A\otimes\id_B)\big((R_A-m_A)H_{AB}(R_B-m_B)\big)\right]+\text{const.}\\
    &=-\frac{1}{2}(R_B-m_B)^\intercal H_{BB}(R_B-m_B)\\
    &\quad-(\Tr_A[R_A\rho_A]-m_A)^\intercal H_{AB}(R_B-m_B)+\text{const.}\\
    &=-\frac{1}{2}(R_B-m_B)^\intercal H_{BB}(R_B-m_B)+\text{const.}
\ee
where in the last line we have used that $\Tr_A[R_A\rho_A]=m_A$. Therefore,
\bb
        \uml{\rho}_{B}=\frac{\exp\left[-\frac{1}{2}(R_B-m_B)^\intercal H_{BB}(R_B-m_B)\right]}{\Tr\left[\exp\left[-\frac{1}{2}(R_B-m_B)^\intercal H_{BB}(R_B-m_B)\right]\right]},
\ee
which has a covariance matrix
\bb
        \uml{V}_B=V_B(H_{BB})=\coth\left(\frac{i\Omega_{B} H_{BB}}{2}\right)i\Omega_{B},
\ee
and this concludes the proof.
\end{proof}


\section{Channel quantum umlaut information}\label{sec:channels}

\subsection{Definition and basic properties}

\begin{figure}
    \centering
    \includegraphics[width=0.6\linewidth]{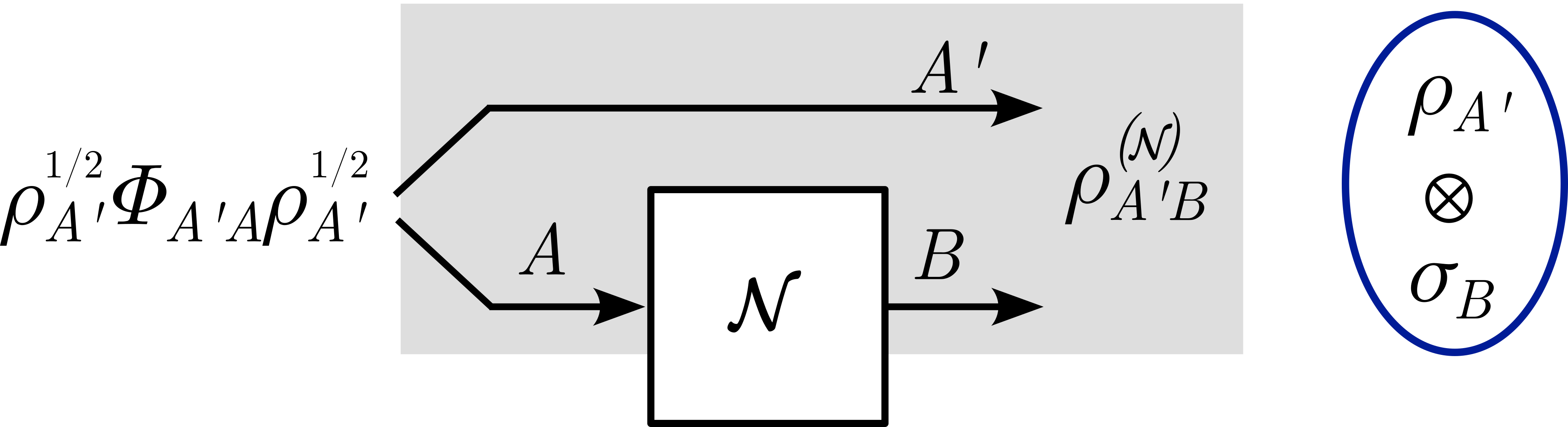}
    \caption{The umlaut information of the channel $\pazocal{N}$ is the largest umlaut information of $\rho^{(\pazocal{N})}_{A'B}$ by varying $\rho_{A'}$.}
    \label{fig:channel_lautum}
\end{figure}

Here, we discuss the notion of umlaut information of a quantum channel, which generalises the umlaut information of a classical channel defined in~\cite{Filippo25}. A pictorial representation is given in Figure~\ref{fig:channel_lautum}.

\begin{Def}[(Umlaut information of a channel)]\label{def:channel_lautum}
Let $\NN :\mathcal{L}(\mathcal{H}_A)\to \mathcal{L}(\mathcal{H}_B)$ be a quantum channel and let $\mathcal{H}_{A'}$ be an auxiliary Hilbert space isomorphic to $\mathcal{H}_{A}$. Then the quantum umlaut information of $\NN$ is defined as
\bb
U(\NN) \coloneqq \sup_{\Psi_{A'A}} U(A';B)_{(\mathrm{Id}\otimes\NN)(\Psi)} = \sup_{\rho_{A'}} \min_{\sigma_B} D\left(\rho_{A'}\otimes \sigma_B \,\middle\|\, \rho_{A'}^{1/2} J_{A'B}^{(\NN)} \rho_{A'}^{1/2}\right) ,
\label{channel_umlaut}
\ee
where $\sigma_B$ is an arbitrary mixed state of $\mathcal{H}_B$ and $\Psi_{A'A}$ is an arbitrary input state of $\mathcal{H}_{A'}\otimes\mathcal{H}_A$, which can be assumed to be pure, without loss of generality, due to Lemma~\ref{lem:quantum_lautum_prop}(3); the state $\rho_{A'}$ in the rightmost expression is the $A'$ marginal of $\Psi_{A'A}$, and $J_{A'B}^{(\NN)}$ is the (un-normalised) Choi--Jamio\l kowski matrix of $\pazocal{N}$, defined as
\bb
J^{(\pazocal{N})}_{A'B}\coloneq (\mathrm{Id}_{A'}\otimes\pazocal{N}_{A\to B})(\Phi_{A'A}),
\label{unnormalised_Choi}
\ee
with $\Phi_{A'A}\coloneqq\sum_{ij}\ketbraa{ii}{jj}_{A'A}$ being the (un-normalised) maximally entangled state between $A'$ and $A$.
\end{Def}

According to the specific context, it will be convenient to use different equivalent notations for the channel quantum umlaut information. We introduce $\rho_{A'B}^{(\pazocal{N})}$ as the output state of $\mathcal{N}_{A\to B}$ when as input we consider the standard purification of $\rho_{A'}$, namely
    \bb
        \rho_{A'B}^{(\pazocal{N})}\coloneqq ({\rm Id}_{A'}\otimes\pazocal{N}_{A\to B})\left(\left(\rho_{A'}^{1/2}\otimes\id_A\right)\Phi_{A'A}\left(\rho_{A'}^{1/2}\otimes\id_A\right)\right)=\left(\rho_{A'}^{1/2}\otimes\id_A\right)\,J^{(\pazocal{N})}_{A'B}\,\left(\rho_{A'}^{1/2}\otimes\id_A\right).
        \label{eq:output_state}
    \ee
In order to unburden the notation, whenever there is no risk of ambiguity, we will imply the tensor products with identity operators or identity channels. Therefore, all the following forms of the channel umlaut information will have to be considered equivalent:
    \bb
        U(\pazocal{N})&= \sup_{\rho_{A'}}\min_{\sigma_B}D\left(\rho_{A'}\otimes\sigma_B\,\middle\|\,\rho_{A'}^{1/2}\,J_{A'B}^{(\pazocal{N})}\,\rho_{A'}^{1/2}\right)\\ &=\sup_{\rho_{A'}}\min_{\sigma_B}D\left(\rho_{A'}\otimes\sigma_B\,\middle\|\,\pazocal{N}_{A\to B}\Big(\rho_{A'}^{1/2}\,\Phi_{A'A}\,\rho_{A'}^{1/2}\Big)\right)\\
        &=\sup_{\rho_{A'}}\min_{\sigma_B}D\left(\rho_{A'}\otimes\sigma_B\,\middle \|\,\rho_{A'B}^{(\pazocal{N})}\right).
    \ee
The aim of this section is to generalise the operational interpretation for the umlaut information of a classical channel $U(\pazocal{W})$ that we provided in~\cite{Filippo25}. In the classical setting it turned out that the channel umlaut information quantifies the error exponent of channel coding for non-signalling--assisted codes at arbitrarily small communication rates. Using activated non-signalling--assisted codes, we are going to extend this result to classical-quantum channels, while for generic quantum channels we need to regularise the umlaut information in order to have the interpration as an error exponent.

We start the exploration of the properties of the channel umlaut information by collecting a few basic facts.

\begin{lemma} \label{lem:channel_umlaut_prop}
The umlaut information of the quantum channel $\pazocal{N} = \pazocal{N}_{A\to B}$ satisfies the following properties:
\begin{enumerate}[(1)]
\item Positive definiteness: $U(\pazocal{N}) \geq 0$, with equality if and only if $\NN(\cdot) = \Tr_A[\cdot] \sigma_B$ is a replacer channel, i.e.\ it maps any input state to a fixed output state.
\item Boundedness: denoting with $Q_{A'B}$ the projector onto $\ker\big(J_{A'B}^{\NN}\big)$, we have $U(\NN) < \infty\ \Longleftrightarrow\ \ker\big(\Tr_{A'} Q_{A'B}\big)\neq \{0\}$, i.e.\ $\Tr_{A'} Q_{A'B}$ is not of full rank.
\item Boundedness, CQ case: if $\NN: x\mapsto \rho_x$ is classical-quantum, then $U(\NN) < \infty$ if and only if $\bigcap_x \supp(\rho_x) \neq \{0\}$.
\end{enumerate}
\end{lemma}

\begin{proof}
Claim~(1) follows easily from Lemma~\ref{lem:quantum_lautum_prop}: the inequality $U(\pazocal{N}) \geq 0$ is an immediate consequence of the same claim for states, and equality holds if and only if, for all input states $\Psi_{A'A}$, we have $\NN_{A\to B}(\Psi_{A'A}) = \Psi_{A'}\otimes \NN_{A\to B}(\Psi_A)$. Taking $\Psi_{A'A}$ as the maximally entangled state, we conclude that the Choi--Jamio\l kowski state of $\NN$ is a product state, which implies that $\NN$ is indeed a replacer channel.

We now move on to~(2). We claim that $U(\NN) < \infty$ is equivalent to the existence of a state $\sigma_B$ with the property that $\supp(\id_{A'}\otimes \sigma_B) \subseteq \supp\big(J_{A'B}^{(\NN)}\big)$, where $J_{A'B}^{(\NN)}$ is defined by~\eqref{unnormalised_Choi}. That this condition is necessary can be seen by taking the maximally entangled state as an ansatz for $\Psi_{A'A}$ in the first expression in~\eqref{channel_umlaut}, and then using Lemma~\ref{lem:quantum_lautum_prop}(2). Vice versa, if such a state exists then for some real number $M>0$ we have $\id_{A'}\otimes \sigma_B \leq M\, J_{A'B}^{(\NN)}$, implying that for all states $\rho_{A'}$ the operator inequality $\rho_{A'}\otimes \sigma_B \leq M\, \rho_{A'}^{1/2} J_{A'B}^{(\NN)} \rho_{A'}^{1/2}$ holds. Together with the operator monotonicity of the logarithm, this entails that
\bb
D\left(\rho_{A'}\otimes \sigma_B \,\middle\|\, \rho_{A'}^{1/2} J_{A'B}^{(\NN)} \rho_{A'}^{1/2}\right) \leq \log M\, ,
\ee
so that
\bb
U(\NN) = \sup_{\rho_{A'}} \inf_{\sigma'_B} D\left(\rho_{A'}\otimes \sigma'_B \,\middle\|\, \rho_{A'}^{1/2} J_{A'B}^{(\NN)} \rho_{A'}^{1/2}\right) \leq \log M < \infty\, .
\ee

Now, the condition that $\supp(\id_{A'}\otimes \sigma_B) \subseteq \supp\big(J_{A'B}^{(\NN)}\big)$ is equivalent, via considering orthogonal complements, to $\ker\big(J_{A'B}^{(\NN)}\big) \subseteq \ker(\id_{A'}\otimes \sigma_B)$. Calling $Q_{A'B}$ the projector onto $\ker\big(J_{A'B}^{(\NN)}\big)$ and setting $Q_B \coloneqq \Tr_{A'} Q_{A'B}$, the above inclusion can be reformulated as the equality
\bb
0 = \Tr_{A'B} Q_{A'B} (\id_{A'}\otimes \sigma_B) = \Tr_B Q_{B} \sigma_B\, ,
\ee
which needs to be obeyed for some state $\sigma_B$. This is possible if and only if $Q_B$ is not of full rank, showing~(2).

Finally, (3)~follows directly from~(2). For a CQ channel $\NN_{X\to B}$, we have $J_{XB}^{(\NN)} = \sum_x \ketbra{x}_X \otimes \rho_x$. Calling $P_x$ the projector onto $\ker(\rho_x)$, the projector onto $\ker\big(J_{XB}^{(\NN)}\big)$ can be written as
\bb
Q_{XB} = \sum_x \ketbra{x} \otimes P_x\, .
\ee
Hence, $\Tr_{X} Q_{XB} = \sum_x P_x$, which is not of full rank if and only if
\bb
\{0\} \neq \bigcap_x \ker(P_x) = \bigcap_x \supp(\rho_x)\, .
\ee
This completes the proof.
\end{proof}

Similarly to what happens for states, also in the channel case we have $U(\NN) < \infty$ whenever $\NN$ belongs to the interior of the set of quantum channels, i.e.\ when $\NN(\rho)>0$ for all input states $\rho$. 

In analogy with~\cite[Proposition~13]{Filippo25}, we now establish some convexity/concavity properties of the functional that defines the quantum channel umlaut information.

\begin{prop} \label{prop:concavity}
The functional 
\bb
\big(\rho_{A'}, \sigma_B, \NN\big) \mapsto D\left(\rho_{A'}\otimes\sigma_B\,\middle\|\,\rho_{A'B}^{(\NN)}\right) ,
\ee
where $\rho_{A'B}^{(\NN)}$ is defined by~\eqref{eq:output_state}, is concave in $\rho_{A'}$ and jointly convex in $\sigma_B$ and $\NN$. In particular, for all quantum channels $\NN = \NN_{A\to B}$ the map
\bb
\rho_{A'}\mapsto \min_{\sigma_B} D\left(\rho_{A'}\otimes\sigma_B\,\middle\|\,\rho_{A'B}^{(\NN)}\right)
\label{eq:concavity_umlaut_input_state}
\ee
is concave.
\end{prop}

\begin{proof}
Concavity in $\rho_{A'}$ follows from~\cite[Proposition~7.83]{KHATRI} applied to the replacer channel $\omega_A \mapsto \Tr_A[\omega_A] \sigma_B$ (first argument) and the channel $\NN$ (second argument). Joint convexity in $\sigma_B$ and $\NN$ comes from the joint convexity of the quantum relative entropy. Finally, the concavity of the map~\eqref{eq:concavity_umlaut_input_state} descends from the observation that the point-wise minimum of concave functions is concave.
\end{proof}

Equipped with these tools, we can now connect the umlaut information with the well-studied notion of channel relative entropy. We recall that the channel relative entropy between $\pazocal{M}_{A\to B}$ and $\pazocal{N}_{A\to B}$ is defined as
\bb
D(\pazocal{M}\|\pazocal{N})\coloneqq \sup_{\rho_{A'}}D\left(\left(\mathrm{Id}_{A'}\otimes\pazocal{M}_{A\to B}\right)\left(\rho_{A'}^{1/2}\,\Phi_{A'A}\,\rho_{A'}^{1/2}\right)\,\middle\|\,\left(\mathrm{Id}_{A'}\otimes\pazocal{N}_{A\to B}\right)\left(\rho_{A'}^{1/2}\,\Phi_{A'A}\,\rho_{A'}^{1/2}\right)\right),
\ee
where $\Phi_{A'A}\coloneqq\sum_{ij}\ketbraa{ii}{jj}_{A'A}$ is the (unnormalised) maximally entangled state between $A'$ and $A$.

\begin{lemma}[(Channel umlaut information and channel relative entropy)]\label{lem:replacer} 
Given $\sigma_B\in \mathcal{D}(\mathcal{H}_B)$, let $\pazocal{R}^{\sigma_B}_{A\rightarrow B}:\mathcal{L}(\mathcal{H}_A)\to \mathcal{L}(\mathcal{H}_B)$ be the replacer channel defined as
    \bb
        \pazocal{R}^{\sigma_B}_{A\rightarrow B}(\,\cdot\,)\coloneqq \Tr_A[\,\cdot\,]\,\sigma_B.
    \ee
    The umlaut information of a channel $\pazocal{N}_{A\to B}$ can be written as the minimal channel relative entropy between the family of replacer channels $\{\pazocal{R}^{\sigma_B}_{A\rightarrow B}\}_{\sigma_B\in\mathcal{D}(\mathcal{H}_B)}$ 
    and the channel $\pazocal{N}$ itself, namely
    \bb \label{eq:rep_channel}
    U(\pazocal{N}) &= \min_{\sigma_B}D\big(\pazocal{R}^{\sigma_B}_{A\rightarrow B}\,\big\|\,\pazocal{N}_{A\to B}\big)\, .
    \ee
\end{lemma}

Before proving the Lemma~\ref{lem:replacer}, we need to state a minimax lemma.

\begin{lemma}[{\cite[Theorem~5.2]{Farkas2006}}]\label{lem:minmax} Let $\XX$ be a compact and convex set in a Hausdorff topological vector space, and let $\YY$ be a convex set in any vector space. Let $f: \XX \times \YY \to (-\infty, +\infty]$ be lower semicontinuous on $\XX$ for all fixed $y\in \YY$, convex in the first variable, and concave in the second. Then
\bb
\sup_{y\in \YY}\inf_{x\in \XX} f(x,y) = \inf_{x\in \XX}\sup_{y\in \YY} f(x,y)\, .
\label{eq:minimax}
\ee
\end{lemma}

Due to the lower semi-continuity of $f$ in the first variable and to the compactness of $\XX$, the infimum on the left-hand side of~\eqref{eq:minimax} can be replaced with a minimum. Since the point-wise supremum of lower semi-continuous functions is again lower semi-continuous, the same is true for the infimum on the right-hand side of~\eqref{eq:minimax}. Thus, we can rewrite~\eqref{eq:minimax} as
\bb
\sup_{y\in \YY} \min_{x\in \XX} f(x,y) = \min_{x\in \XX} \sup_{y\in \YY} f(x,y)\, .
\label{eq:minimax_modified}
\ee

\begin{proof}[Proof of Lemma~\ref{lem:replacer}]
By Definition~\ref{def:channel_lautum}, we have
\bb
U(\pazocal{N})&=\sup_{\rho_{A'}}\min_{\sigma_B}D\left(\rho_{A'}\otimes\sigma_B\middle\|\rho_{A'B}^{(\pazocal{N})}\right)\\
&= \min_{\sigma_B}\sup_{\rho_{A'}}D\left(\rho_{A'}\otimes\sigma_B\middle\|\rho_{A'B}^{(\pazocal{N})}\right)\\
&=\min_{\sigma_B}\sup_{\rho_{A'}}D\left(\left(\mathrm{Id}_{A'}\otimes\pazocal{R}^{\sigma_B}_{A\rightarrow B}\right)\left(\rho_{A'}^{1/2}\,\Phi_{A'A}\,\rho_{A'}^{1/2}\right)\,\middle\|\,\left(\mathrm{Id}_{A'}\otimes\pazocal{N}_{A\to B}\right)\left(\rho_{A'}^{1/2}\,\Phi_{A'A}\,\rho_{A'}^{1/2}\right)\right)\\
&=\min_{\sigma_B} D\left(\pazocal{R}^{\sigma_B}_{A\rightarrow B}\middle\|\pazocal{N}\right) .
\ee
Here, the second equality is a consequence of Lemma~\ref{lem:minmax} in the form~\eqref{eq:minimax_modified}. This can be invoked due to Proposition~\ref{prop:concavity} and by observing that the lower semi-continuity in $\sigma_B$ follows from the lower semi-continuity of the relative entropy.
\end{proof}


\subsection{Additivity discussion}

A crucial property of the classical channel umlaut information was the additivity under tensor product of channels~\cite{Filippo25}. In the quantum setting we are going to prove three results about the additivity  property of the quantum channel umlaut information:
\begin{itemize}
    \item for generic quantum channels $\pazocal{N}$ and $\pazocal{M}$, it is super-additive: $U(\pazocal{N}\otimes\pazocal{M})\geq U(\pazocal{N})+U(\pazocal{M})$, see Lemma~\ref{lem:super_add};
    \item for classical-quantum channels $\pazocal{N}$ and $\pazocal{M}$, it is additive: $U(\pazocal{N}\otimes\pazocal{M})= U(\pazocal{N})+U(\pazocal{M})$, see Proposition~\ref{prop:add_CQ};
    \item there exists a quantum channel $\pazocal{N}$ such that $U(\pazocal{N}^{\otimes 2})> 2U(\pazocal{N})$, see Proposition~\ref{prop:non_add}.
\end{itemize}

\begin{lemma}[(Super-additivity of the channel umlaut information)]\label{lem:super_add}
Let $\pazocal{N}_{A_1\to B_1}$ and $\pazocal{M}_{A_2\to B_2}$ be two quantum channels. Then
    \bb
        U(\pazocal{N}\otimes \pazocal{M}) \geq U(\pazocal{N})+U(\pazocal{M}),
    \ee
    i.e.\ the channel umlaut information is super-additive under tensor products.
\end{lemma}

\begin{proof}
    By the specific ansatz $\rho_{A'_1A'_2}=\rho_{A'_1}\otimes\rho_{A'_2}$, we get
    \bb
        \rho_{A_1'B_1A_2'B_2}^{(\pazocal{N}\otimes\pazocal{M})}=\left(\rho_{A_1'}^{1/2}\otimes\rho_{A_2'}^{1/2}\right)\,\left(J^{(\pazocal{N})}_{A_1'B_1}\otimes J^{(\pazocal{M})}_{A_2'B_2}\right)\,\left(\rho_{A_1'}^{1/2}\otimes\rho_{A_2'}^{1/2}\right)=\rho_{A_1'B_1}^{(\pazocal{N})}\otimes \rho_{A_2'B_2}^{(\pazocal{M})}.
    \ee
    Therefore, by the additivity of the umlaut information under tensor products (Corollary~\ref{cor:additivity}),
    \bb
        U(\pazocal{N}\otimes\pazocal{M})&\geq \sup_{\rho_{A'_1}\otimes\rho_{A'_2}}U(A'_1A'_2;B_1B_2)_{\rho^{(\pazocal{N}\otimes\pazocal{M})}}\\
        &=\sup_{\rho_{A'_1}\otimes\rho_{A'_2}}U(A'_1A'_2;B_1B_2)_{\rho^{(\pazocal{N})}\otimes \rho^{(\pazocal{M})}}\\
        &= \sup_{\rho_{A'_1}}U(A'_1;B_1)_{\rho^{(\pazocal{N})}}+\sup_{\rho_{A'_2}}U(A'_2;B_2)_{\rho^{(\pazocal{M})}}\\
        &=U(\pazocal{N})+U(\pazocal{M}),
    \ee
    which concludes the proof.
\end{proof}

\begin{prop}[(Additivity of $U(\pazocal{N})$ for classical-quantum channels)]\label{prop:add_CQ}
Let $\pazocal{N}:\mathcal{P}(\mathcal{X})\to \mathcal{D}(\mathcal{H}_{B})$ be a classical-quantum channel:
\bb
    \pazocal{N}(P_{X})= \sum_{x\in\mathcal{X}} P_{X}(x)\,\rho_x^{B}.
\ee
Then
    \bb
        U(\pazocal{N})=-\log\min_{P_X}\Tr\left[\exp\left(\sum_{x\in\mathcal{X}}P_X(x)\log\rho_x^B\right)\right]
    \ee
In particular, if we consider two classical-quantum channels $\pazocal{N}:\mathcal{P}(\mathcal{X}_1)\to \mathcal{D}(\mathcal{H}_{B_1})$ and $\pazocal{M}:\mathcal{P}(\mathcal{X}_2)\to \mathcal{D}(\mathcal{H}_{B_2})$, then we have
    \bb
        U(\pazocal{N}\otimes \pazocal{M}) = U(\pazocal{N})+U(\pazocal{M}),
    \ee
    i.e.\ the channel umlaut information is additive under tensor products for classical-quantum channels.
\end{prop}

\begin{proof} 
The proof strategy is similar to the fully  classical case~\cite[Proposition 19 and Corollary 21]{Filippo25}. We recall that, using the quantum notation, $P_X\in\mathcal{P}(\mathcal{X})$ and $\pazocal{N}_{X\to B}$ are represented as
\bb
    P_X \coloneqq \sum_{x\in\mathcal{X}}P_X(x)\ketbra{x}_X,\quad
    \pazocal{N}_{X\to B}(\,\cdot\,)\coloneqq \sum_{x\in\mathcal{X}} \bra{x}\,\cdot\,\ket{x}_X\, \rho_x^B.
\ee
Then, we have
\bb\label{eq:choi_CQ}
    \rho_{X'B}^{(\pazocal{N})}&= \pazocal{N}_{X\to B}\left(\sum_{x_1\in\mathcal{X}}P^{1/2}_{X}(x_1)\ketbra{x_1}_{X'}\otimes\id_X \sum_{x_2,x_3\in\mathcal{X}}\ket{x_2\, x_2}\bra{x_3\, x_3}_{X'X}\sum_{x_4\in\mathcal{X}}P^{1/2}_{X}(x_4)\ketbra{x_4}_{X'}\otimes\id_X\right)\\
    &=\pazocal{N}_{X\to B}\left( \sum_{x_2,x_3\in\mathcal{X}}P^{1/2}_{X}(x_2)P^{1/2}_{X}(x_3)\ket{x_2\, x_2}\bra{x_3\, x_3}_{X'X}\right)\\
    &=\sum_{x\in\mathcal{X}}P_{X}(x)\ketbra{x}_{X'}\otimes\rho_x^B.
\ee
Noticing that the states $\{\ketbra{x}_{X'}\otimes\rho_x^B\}_{x\in\mathcal{X}}$ are orthogonal and using the closed-form expression provided by Proposition~\ref{prop:closed_lautum}, we compute
\bb\label{eq:formula_without_max}
        U(X';B)_{\rho^{(\pazocal{N})}}&=-S\left(\rho_{X'}^{(\pazocal{N})}\right)-\log\Tr\left[\exp\left(\Tr_{X'}[(P_{X'}\otimes\id_B)\log\rho_{X'B}^{(\pazocal{N})}]\right)\right]\\
        &=-S(P_X)-\log\Tr\left[\exp\left(\Tr_{X'}\sum_{x\in\mathcal{X}}P_X(x)\Big(\log P_X(x) \ketbra{x}_{X'}\otimes\id_B+\ketbra{x}_{X'}\otimes\log\rho_x^B\Big)\right)\right]\\
        &=-S(P_X)-\log\Tr\left[\exp\left(-S(P_X)\id_B+\sum_{x\in\mathcal{X}}P_X(x)\log\rho_x^B\Big)\right)\right]\\
        &=-\log\Tr\left[\exp\left(\sum_{x\in\mathcal{X}}P_X(x)\log\rho_x^B\Big)\right)\right],
    \ee
whence
    \bb
        U(\pazocal{N})=\sup_{P_X}-\log\Tr\left[\exp\left(\sum_{x\in\mathcal{X}}P_X(x)\log\rho_x^B\right)\right].
    \ee
   The additivity immediately follows. Let
    \bb
        \pazocal{N}(P_{X_1})= \sum_{x_1\in\mathcal{X}_1} P_{X_1}(x_1)\,\rho_{x_1}^{B_1},\qquad \pazocal{M}(P_{X_2})= \sum_{x_2\in\mathcal{X}_2} P_{X_2}(x_2)\,\sigma_{x_2}^{B_2}.
    \ee
    Then
    \bb
        U(\pazocal{N}\otimes\pazocal{M})
        &=\sup_{P_{X_1X_2}}-\log\Tr\exp\left(\sum_{\substack{x_1\in\mathcal{X}_1\\x_2\in\mathcal{X}_2}}P_{X_1X_2}(x_1,x_2)\log\left(\rho_{x_1}^{B_1}\otimes \sigma_{x_2}^{B_2}\right)\right)\\
        &=\sup_{P_{X_1X_2}}-\log\Tr\exp\left(\sum_{\substack{x_1\in\mathcal{X}_1\\x_2\in\mathcal{X}_2}}P_{X_1X_2}(x_1,x_2)\left(\log\rho_{x_1}^{B_1}\otimes\id_{B_2}+\id_{B_1}\otimes\log \sigma_{x_2}^{B_2}\right)\right)\\
        &=\sup_{P_{X_1X_2}}-\log\Tr\exp\left(\sum_{x_1\in\mathcal{X}_1}P_{X_1}(x_1)\log\rho_{x_1}^{B_1}\otimes\id_{B_2}+\id_{B_1}\otimes\sum_{x_2\in\mathcal{X}_2}P_{X_2}(x_2)\log \sigma_{x_2}^{B_2}\right)\\
        &=\sup_{P_{X_1X_2}}-\log\Tr\left[\exp\left(\sum_{x_1\in\mathcal{X}_1}P_{X_1}(x_1)\log\rho_{x_1}^{B_1}\right)\otimes\exp\left(\sum_{x_2\in\mathcal{X}_2}P_{X_2}(x_2)\log \sigma_{x_2}^{B_2}\right)\right]\\
        &=\sup_{P_{X_1}}-\log\Tr\exp\left(\sum_{x_1\in\mathcal{X}_1}P_{X_1}(x_1)\log\rho_{x_1}^{B_1}\right)+\sup_{P_{X_2}}-\log\Tr\exp\left(\sum_{x_2\in\mathcal{X}_2}P_{X_2}(x_2)\log \sigma_{x_2}^{B_2}\right)\\
        &=U(\pazocal{N})+U(\pazocal{M}),
    \ee
where $P_{X_1}$ and $P_{X_2}$ are the marginals of $P_{X_1X_2}$.
This concludes the proof.
\end{proof}

At this point, we might wonder if additivity holds for generic channels, at least for many copies of a single quantum channel. This amounts to asking whether $U(\pazocal{N}^{\otimes n})\eqt{?} nU(\pazocal{N})$ for all positive integers $n$. In the following subsection we are going to provide an explicit counterexample. Since numerical computations to estimate the maximum over $\rho_{A'}$ in the definition of the quantum umlaut information become harder as the dimension of the system $A'$ increases, it might prove hard to efficiently compute the umlaut information of a tensor product of two channels. It turns out that such an optimisation is easier when a channel is covariant under a group action.

\begin{prop}[(Computation of $U(\pazocal{N})$ for covariant channels)]\label{prop:group}
Let $\pazocal{N}_{A\to B}$ be a covariant channel under the action of a finite group $G$ represented by $\{U_A(g)\}_{g\in G}$ and $\{V_B(g)\}_{g\in G}$, i.e.
\bb
    \pazocal{N}\left(U_A(g) \,\cdot \, U_A^\dagger(g)\right)=
    V_B(g)\pazocal{N}(\,\cdot \, )V_B^\dagger(g)\qquad \forall\, g\in G.
\ee
Then
\bb\label{eq:maximisation_group}
    U(\pazocal{N})=\sup_{\rho_{A'}^{(G)}}\left(-S\left(\rho_{A'}^{(G)}\right)-\log\Tr\left[\exp\Tr_{A'}\left[\rho_{A'}^{(G)}\log\left((\rho_{A'}^{(G)})^{1/2}\, J^{(\pazocal{N})}_{A'A}\,(\rho_{A'}^{(G)})^{1/2}\right)\right]\right]\right)
\ee
where the maximisation is over all states $\rho_{A'}^{(G)}$ that are invariant under the action of $\{U_{A'}(g)\}_{g\in G}$, and $J^{(\pazocal{N})}_{A'A}$ is the Choi--Jamio\l kowski state of $\pazocal{N}$.
\end{prop}

The proof is deferred to Appendix~\ref{app:covariant_channels}.

\begin{figure}[h]
    \centering
    \includegraphics[width=0.7\linewidth]{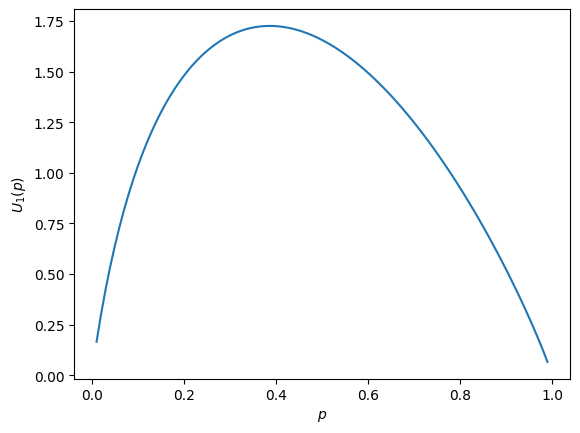}
    \caption{Estimation of the umlaut information of $\pazocal{A}_{\gamma,\beta}$ with $(\gamma,\beta)=(0.5,0.001)$ using the family of states $\rho_{A'}(p)$. The umlaut information of the channel is given by the maximum of the function in the plot.}
    \label{fig:L_p}
\end{figure}

\begin{prop}[(Additivity violation for channel quantum umlaut information)] \label{prop:non_add}
There exists a quantum channel $\pazocal{N}:\mathcal{L}(\mathbb{C}^2)\to \mathcal{L}(\mathbb{C}^2)$ such that
    \bb
        U(\pazocal{N}^{\otimes 2})>2U(\pazocal{N}).
    \ee
\end{prop}

\begin{proof}
Similarly to the strategy used in~\cite[Proposition~3.1]{Fang2020}, we can leverage Proposition~\ref{prop:group} in order to reduce the optimisation problem over a qubit state to a single parameter optimisation problem in the case of a $Z$-covariant channel, i.e.\ a channel $\pazocal{N}:\mathcal{L}(\mathbb{C}^2)\to \mathcal{L}(\mathbb{C}^2)$ such that
\bb
    \pazocal{N}Z=Z\pazocal{N},\qquad\text{with}\qquad Z(\,\cdot\,)\coloneqq \sigma_z \,\cdot \, \sigma_z.
\ee 
    Let us consider, as a class of $Z$-covariant channels, the parameterised family of generalised amplitude damping channels given by
    \bb
        \pazocal{A}_{\gamma,\beta}(\,\cdot\,)\coloneqq \sum_{i=1}^4A_{\gamma,\beta}^{(i)}\,\cdot\,A_{\gamma,\beta}^{(i)}
    \ee
    with 
    \bb
        A_{\gamma,\beta}^{(1)}&\coloneqq \sqrt{1-\beta}(\ketbra{0}+\sqrt{1-\gamma}\ketbra{1}), & A_{\gamma,\beta}^{(2)}&=\sqrt{\gamma(1-\beta)}\ket{0}\bra{1},\\
        A_{\gamma,\beta}^{(3)}&\coloneqq \sqrt{\beta}(\sqrt{1-\gamma}\ketbra{0}+\ketbra{1}), & A_{\gamma,\beta}^{(4)}&=\sqrt{\gamma\beta}\ket{1}\bra{0}.
    \ee
     For instance, choosing 
     $(\gamma, \beta)=(0.5,0.001)$, we have
    \bb
        J^{(\pazocal{A}_{\gamma,\beta})}=
        \begin{pmatrix}
        0.9995 &  &  & \sqrt{0.5} \\
         & 0.005 &  &  \\
         &  & 0.4995 &  \\
        \sqrt{0.5} &  &  & 0.5005
        \end{pmatrix}
    \ee
    in the basis $\{\ket{00}_{A'B}, \ket{01}_{A'B}, \ket{10}_{A'B}, \ket{11}_{A'B}\}$. We parameterise $\rho_{A'}$ as
    \bb
        \rho_{A'}(p)=p\ketbra{0}+(1-p)\ketbra{1}
    \ee
    and, using the \texttt{minimise} function from the SciPy library~\cite{2020SciPy-NMeth}, we get that the maximum is achieved for $p\simeq 0.386$, with
    \bb
        U(\pazocal{A})\simeq 1.725
    \ee
    In \autoref{fig:L_p} we show the plot of the function $U_1(p)$ to be maximised that we get plugging $\rho_{A'}(p)$ in~\eqref{eq:maximisation_group}. As proved in Proposition~\ref{prop:concavity}, such a function is concave in $p$.
    If we now consider $\pazocal{A}_{\gamma,\beta}^{\otimes 2}$, it is clear that any guess $\rho_{A'_1A'_2}^\ast$ provides a lower bound to $U(\pazocal{A}_{\gamma,\beta}^{\otimes 2})$ as follows:
    \bb
        U(\pazocal{A}_{\gamma,\beta}^{\otimes 2})\geq U^\ast_2 \coloneqq -S\left(\rho_{A'_1A'_2}^\ast\right)-\log\Tr\left[\exp\Tr_{A'}\left[\rho_{A'_1A'_2}^\ast\log\left(\left(\rho_{A'_1A'_2}^{\ast}\right)^{1/2}\, J_{\pazocal{A}_{\gamma,\beta}}^{\otimes 2}\,\left(\rho_{A'_1A'_2}^{\ast}\right)^{1/2}\right)\right]\right].
    \ee
     Considering
    \bb
        \rho_{A'_1A'_2}^\ast = 
        \begin{pmatrix}
        0.182 &  &  & \\
         & 0.213 &  &  \\
         &  & 0.213 &  \\
         &  &  & 0.392
        \end{pmatrix}
    \ee
    we get $U^\ast_2 \simeq 3.474$, which gives
    \bb
        \frac{U(\pazocal{A}_{\gamma,\beta}^{\otimes 2})}{U(\pazocal{A}_{\gamma,\beta})}\geq \frac{U^\ast_2}{U(\pazocal{A}_{\gamma,\beta})} \simeq 2.013,
    \ee
 and this concludes the proof. Most of the computations in this proof were performed using the package \texttt{toqito} for Python~\cite{toqito}.
\end{proof}

As a consequence, if we regularise the channel umlaut information for generic quantum channels, we get a new quantity.

\begin{Def}[(Regularised umlaut information of a channel)]
For $\pazocal{N}_{A\to B}$ we define
\bb\label{eq:reg_lautum_ch}
    U^\infty(\pazocal{N})\coloneqq\lim_{n\to\infty}\frac{1}{n}U(\pazocal{N}^{\otimes n}),
\ee
or more explicitly,
\bb
U^\infty(\pazocal{N})=\lim_{n\to\infty}\sup_{\rho_{A'\,^n}}\min_{\sigma_{B^n}}\frac{1}{n}D\left(\rho_{A'\,^n}\otimes\sigma_{B^n} \middle\|\, \rho^{(\pazocal{N}^{\otimes n})}_{A'\,^nB^n}\right).
\ee
\end{Def}

The existence of the limit~\eqref{eq:reg_lautum_ch} is ensured by Fekete's lemma, due to the super-additivity of the sequence $U(\pazocal{N}^{\otimes n})$ deduced from Lemma~\ref{lem:super_add}. By Lemma~\ref{lem:super_add} and Proposition~\ref{prop:add_CQ}, $U^\infty(\pazocal{N})\geq U(\pazocal{N})$, with equality for every CQ channel. 

It is possible to extend the boundedness condition given in Lemma~\ref{lem:channel_umlaut_prop} for the channel umlaut information to its regularised version.

\begin{lemma} \label{lem:boundedness_regularised_umlaut}
For a generic quantum channel $\NN = \NN_{A\to B}$, calling $Q_{A'B}$ the projector onto $\ker\big(J_{A'B}^{(\NN)}\big)$ and setting $Q_B \coloneqq \Tr_{A'} Q_{A'B}$, the following are equivalent:
\begin{enumerate}[(a)]
    \item $U(\NN)<\infty$;
    \item $U^\infty(\NN)< \infty$;
    \item $\ker\big(\Tr_{A'} Q_{A'B}\big)\neq \{0\}$.
\end{enumerate}
\end{lemma}

\begin{proof}
The equivalence between~(a) and~(c) is stated already in Lemma~\ref{lem:channel_umlaut_prop}(2). By super-additivity (Lemma~\ref{lem:super_add}), we only need to show that~(c) implies~(b). 

By assumption, the projector onto $\ker Q_B$, which we denote by $P_B$ in what follows, is non-zero. Set $\sigma_B \coloneqq \frac{P_B}{\dim\ker Q_B}$. Calling $\mu(\NN)$ the minimal non-zero eigenvalue of $J_{A'B}^{(\NN)}$, this gives us
\bb
\id_{A'}\otimes \sigma_B = \id_{A'} \otimes \frac{P_B}{\dim\ker Q_B} \leq \frac{1}{\dim\ker Q_B}\, \left( \id_{A'B} - Q_{A'B}\right) \leq \frac{\mu(\NN)^{-1}}{\dim\ker Q_B}\, J_{A'B}^{(\NN)}\, .
\label{tensorisable_inequality}
\ee

Here, the first inequality can be deduced by observing that
\bb
\ker\left(\id_{A'B} - Q_{A'B}\right) = \supp\big(Q_{A'B}\big) \subseteq \HH_{A'}\otimes \supp(Q_B) = \HH_{A'} \otimes \ker(P_B) = \ker\left(\id_{A'}\otimes P_B\right) ,
\ee
which implies that $\id_{A'}\otimes P_B \leq \id_{A'B} - Q_{A'B}$ because both sides are projectors. The last inequality in~\eqref{tensorisable_inequality} holds by definition of $\mu(\NN)$. 

Taking tensor powers, from~\eqref{tensorisable_inequality} we obtain that
\bb
\id_{A'}^{\otimes n} \otimes \sigma_B^{\otimes n} \leq \left(\frac{\mu(\NN)^{-1}}{\dim\ker Q_B}\right)^n\, \big(J_{A'B}^{(\NN)}\big)^{\otimes n} = \left(\frac{\mu(\NN)^{-1}}{\dim\ker Q_B}\right)^n\, J_{{A'}^nB^n}^{(\NN^{\otimes n})}
\ee
for all positive integers $n$. As detailed in the proof of Lemma~\ref{lem:channel_umlaut_prop}, this implies that
\bb
U\big(\NN^{\otimes n}\big) \leq \log \left(\frac{\mu(\NN)^{-1}}{\dim\ker Q_B}\right)^n = n \log \left(\frac{\mu(\NN)^{-1}}{\dim\ker Q_B}\right) .
\ee
Diving by $n$ and taking the limit shows that
\bb
U^\infty(\NN) \leq \log \left(\frac{\mu(\NN)^{-1}}{\dim\ker Q_B}\right) < \infty\, ,
\ee
establishing~(b).
\end{proof}


\subsection{Operational interpretation in quantum communication theory}\label{sec:operational_int}

\begin{figure}[h]
    \centering
    \includegraphics[width=0.45\linewidth]{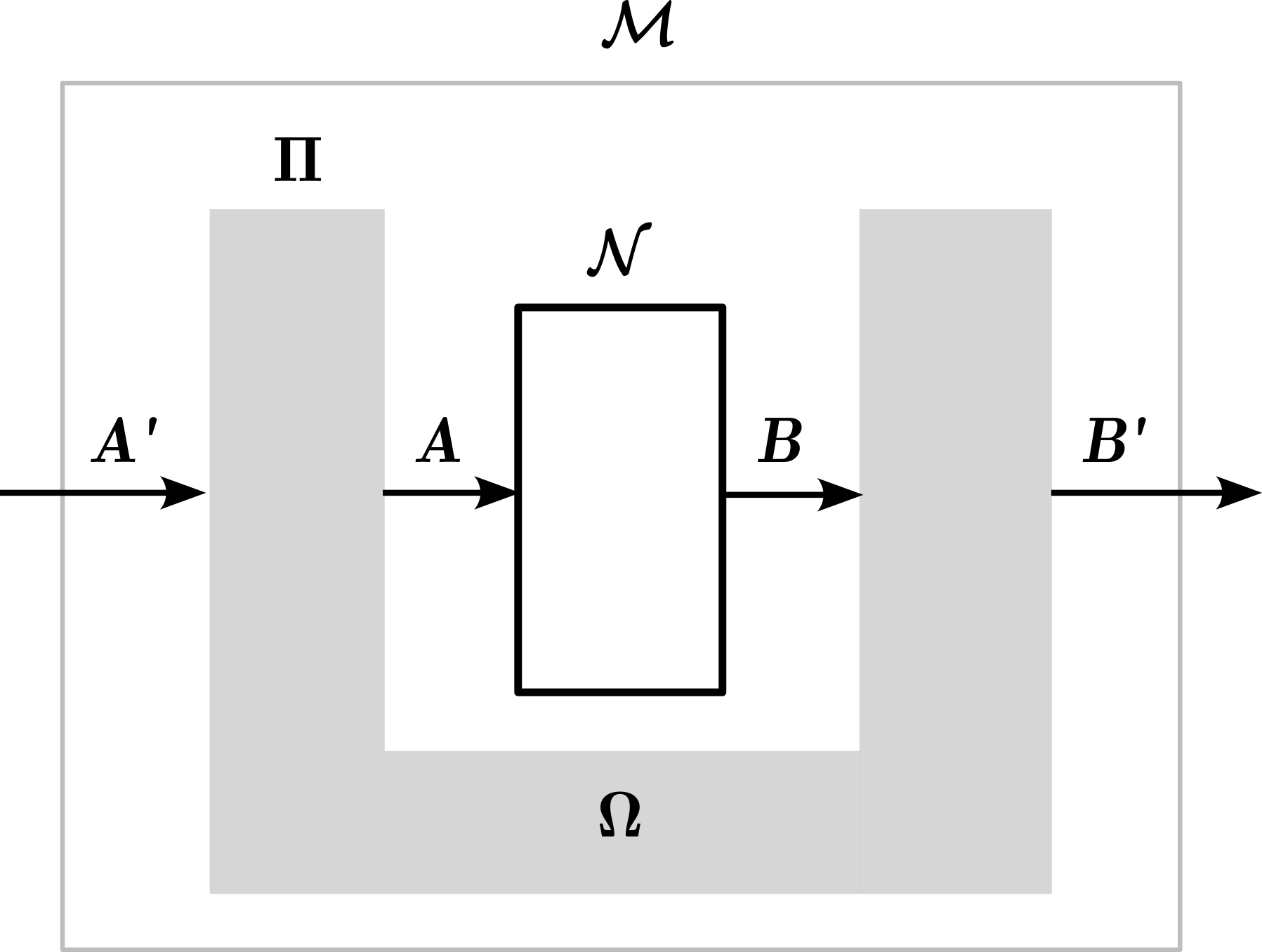}
    \caption{The effective channel $\pazocal{M}$ given by composition of an $\Omega$-assisted code $\Pi$ with a channel $\pazocal{N}$.}
    \label{fig:omega_assisted}
\end{figure}

An $\Omega$-assisted code is a two-input CPTP map $\Pi^\Omega_{A'B\to AB'}$ which sends any channel $\pazocal{N}_{A\to B}$ to an effective channel $\pazocal{M}_{A'\to B'}$ (see \autoref{fig:omega_assisted}) and that can be implemented using local operations with $\Omega$-assistance:
\begin{itemize}
    \item $\Omega=\emptyset$: unassisted codes, which are the product of an encoder $\mathcal{E}_{A'\to A}$ and a decoder $\mathcal{D}_{B\to B'}$, namely $\Pi_{A'B\to AB'}^\emptyset=\mathcal{E}_{A'\to A}\otimes\mathcal{D}_{B\to B'}$;
    \item $\Omega=\rm E$: entanglement-assisted codes;
    \item $\Omega=\rm NS$: non-signalling-assisted codes.
\end{itemize}

Given a channel $\pazocal{N}_{A\to B}$ and an $\Omega$-assisted code $\Pi_{A'B\to AB'}^\Omega$, Alice uses the composite channel $\pazocal{M}_{A'\to B'}=\Pi_{A'B\to AB'}^\Omega\circ \pazocal{N}_{A\to B}$.

\begin{Def}
Given a channel $\pazocal{N}_{A\to B}$ and a number of messages $M\in \mathbb{N}^+$, the minimum   average error probability that can be achieved by an $\Omega$-assisted code is defined as 
\bb
    \epsilon^{ \Omega}(M, \pazocal{N})\coloneqq \inf_{\Pi^{\Omega}}\left(1-\frac{1}{M}\sum_{k=1}^M\Tr\left[\pazocal{M}(\ketbra{k}_{A'})\ketbra{k}_{B'}\right]\right),
\ee
where $\pazocal{M}$ is the composite channel $\pazocal{M}_{A'\to B'}=\Pi^{\Omega}_{A'B\to AB'}\circ \pazocal{N}_{A\to B}$, with $\Pi^{\Omega}_{A'B\to AB'}$ belonging to the class of $\Omega$-assisted codes. The minimum   average error probability that can be achieved by an activated $\Omega$-assisted code is defined as 
\bb
\epsilon^{\Omega, \rm a}(M, \pazocal{N})\coloneqq \inf_{M'\in \mathbb{N}^+} \epsilon^{\Omega}(M\times M', \pazocal{N}\otimes {\rm Id}_{M'}), 
\ee
where  ${\rm Id}_{M}$ is the classical $M$-dimensional noiseless channel 
\bb
    {\rm Id}_M(\,\cdot\,)\coloneqq\sum_{k=1}^M \Tr[\,(\cdot)\,\ketbra{k}]\ketbra{k}.
\ee
The largest size $M$ of the set of messages that can be transmitted with error probability at most $\epsilon \in (0,1)$ using an activated $\Omega$-assisted code is 
\bb
    M^{\Omega,\rm a}(\epsilon,\pazocal{N})\coloneqq\sup\left\{ M : \epsilon^{\Omega,\rm a}(M, \pazocal{N})\le  \epsilon\right\}.
\ee
\end{Def}

The corresponding error exponents are then defined as follows.

\begin{Def}[($\Omega$-assisted error exponents)]
    Given a channel $\pazocal{N}_{A\to B}$ and a rate $r>0$, its activated $\Omega$-assisted error exponent with communication rate $r$ is defined as 
    \bb
    E^{\Omega,\rm a}(r, \pazocal{N}) \coloneqq\liminf_{n\to \infty } -\frac{1}{n}\log \epsilon^{\Omega,\rm a}(\exp(rn), \pazocal{N}^{\otimes n}).
    \ee
    The activated $\Omega$-assisted zero-rate error exponent of $\pazocal{N}$ is defined as 
    \bb
    E^{\Omega,a}(0^+, \pazocal{N}) \coloneqq\liminf_{r\to 0^+ }  E^{\Omega,a}(r, \pazocal{N}).
    \ee
\end{Def}

The largest size $M$ of the set of messages that can be transmitted using activated, non-signalling--assisted codes admits a formulation using the hypothesis testing divergence. This is the so-called Matthews--Wehner converse~\cite{Matthews_2014}, which is achievable via activated, non-signalling--assisted codes~\cite{wang2019-2}, and can be rephrased in terms of the error as follows.

\begin{prop}[(Meta-converse)]
\label{thm:rev_meta}
Let $\pazocal{N}_{A\to B}$ be a quantum channel and let $M> 1$. Then, we have
    \bb\label{eq:reversedmc1}
        -\log\epsilon^{\rm NS,a}(\new{M,}\pazocal{N}\old{,r})=\sup_{\rho_{A'}}\min_{\sigma_B}D^{\old{\exp(-r)}\new{1/M}}_H\left(\rho_{A'}\otimes\sigma_{B}\middle\|\rho_{A'B}^{(\pazocal{N})}\right),
    \ee
    where $\rho_{A'B}^{(\pazocal{N})}\coloneqq \left({\rm Id}_{A'}\otimes\pazocal{N}_{A\to B}\right)\left( \left(\rho_{A'}^{1/2}\otimes\id_A\right)\Phi_{A'A}\left(\rho_{A'}^{1/2}\otimes\id_A\right)\right)$.
\end{prop}

\begin{proof}
    The activated non-signalling error probability corresponds to the meta-converse error probability~\cite{wang2019-2}, where the latter admits the following SDP formulation~\cite[Proposition 24]{Matthews_2014}
    \bb
    1-\epsilon^{\rm NS, a}(M, \pazocal{N})&= \sup_{\rho_{A'}\in \mathcal{D}(A')}\sup_{\Lambda_{A'B}} \left\{ \frac{1}{M}\Tr\left[ J^{(\pazocal{N})}_{A'B} \Lambda_{A'B} \right] \,:\, \Lambda_{B}\le \id_B, 0\le \Lambda_{A'B}  \le M \rho_{A'}\otimes \id_B \right\} 
    \\&\eqt{(i)} \sup_{\rho_{A'}\in \mathcal{D}(A')} \sup_{0\,\le\, O_{A'B} \,\le\, \id_{A'B}} \left\{\Tr\left[\rho_{A'}^{1/2} J^{(\pazocal{N})}_{A'B} \rho_{A'}^{1/2}  O_{A'B}\right]  : \sup_{\sigma_B\in \mathcal{D}(B)} \Tr\left[ (\rho_{A'}\otimes \sigma_B)\ O_{A'B}\right]\le \frac{1}{M} \right\}
     \\&\eqt{(ii)} \sup_{\rho_{A'}\in \mathcal{D}(A')} \min_{\sigma_B\in \mathcal{D}(B)}\sup_{0\,\le\, O_{A'B} \,\le\, \id_{A'B}} \left\{\Tr\left[\rho_{A'}^{1/2} J^{(\pazocal{N})}_{A'B} \rho_{A'}^{1/2}  O_{A'B}\right]  :  \Tr\left[ (\rho_{A'}\otimes \sigma_B) O_{A'B}\right]\le \frac{1}{M} \right\}
     \\&= \sup_{\rho_{A'}\in \mathcal{D}(A')} \min_{\sigma_B\in \mathcal{D}(B)} 1-\exp\left(-D^{1/M}_H\left(\rho_{A'}\otimes\sigma_{B}\middle\|\rho_{A'B}^{(\pazocal{N})}\right)\right),
    \ee
    where in (i) we make the change of variable $O_{A'B} = \frac{1}{M}\rho_{A'}^{-1/2} \Lambda_{A'B} \rho_{A'}^{-1/2}$, and  (ii) can be proved using standard minimax arguments. Hence, we find
 \bb
      -\log \epsilon^{\rm NS, a}(M, \pazocal{N})
     &=\sup_{\rho_{A'}\in \mathcal{D}(A')} \min_{\sigma_B\in \mathcal{D}(B)} D^{1/M}_H\left(\rho_{A'}\otimes\sigma_{B}\middle\|\rho_{A'B}^{(\pazocal{N})}\right).
    \ee
\end{proof}

We are interested in the regime in which, no matter how low the rate is, the error is minimised. The focus is therefore on the quality of the communication regardless of the communication rate, and the figure of merit is the error exponent in the asymptotic limit of large number of uses of the channel:
\bb\label{eq:def_thoroughness}
 \errnsa{\pazocal{N}}\coloneqq
        \liminf_{r\to 0^+}\liminf_{n\to \infty} -\frac{1}{n}\log\epsilon^{\rm NS,a}\left(\exp(rn),\pazocal{N}^{\otimes n}\right).
\ee
As another quantifier, we could consider
\bb\label{eq:T_ast}
    \opnsa{\pazocal{N}}\coloneqq
        \liminf_{{M}\to\infty}\liminf_{n\to \infty}-\frac{1}{n}\log\epsilon^{\rm NS,a}(M,\pazocal{N}^{\otimes n}).
\ee
It might appear that $ \errnsa{\pazocal{N}}$ and  $\opnsa{\pazocal{N}}$ are two different error exponents as the former deals with the situation of sending a number of messages with a vanishing but non-zero rate, while the latter considers sending an arbitrary large but constant number of messages. However, it will turn out that they are actually equal and as such $\opnsa{\pazocal{N}}$ can be seen as an alternative definition of the activated, non-signalling error exponent of the channel $\pazocal{N}$ at zero rate. For the moment, it is easy to prove that $\opnsa{\pazocal{N}}\geq \errnsa{\pazocal{N}}$. For any fixed $M\geq 1$ and $r>0$,
\bb
    \epsilon^{\rm NS,a}(\new{\exp(rn),}\pazocal{N}^{\otimes n}\old{, rn})\geq \epsilon^{\rm NS,a}(\new{M},\pazocal{N}^{\otimes n}\old{, \log M})\qquad \forall\, n\geq \frac{\log M}{r}.
\ee
Therefore, we get
\bb
    \liminf_{n\to \infty} -\frac 1n  \log\epsilon^{\rm NS,a}(\new{\exp(rn),}\pazocal{N}^{\otimes n}\old{, rn})\leq \liminf_{n\to \infty}  -\frac 1n \log\epsilon^{\rm NS,a}(\new{M},\pazocal{N}^{\otimes n}\old{, \log M}),
\ee
whence, by arbitrariness of $M\geq 1$ and $r>0$,
\bb\label{eq:Tast}
    \errnsa{\pazocal{N}}\leq \opnsa{\pazocal{N}}.
\ee

\begin{boxed}{}
\begin{thm}[(Activated, non-signalling error exponent at zero rate of a quantum channel)]\label{thm:exact_tho} 
    Let $\pazocal{N}:\mathcal{D}(\mathcal{H}_A)\to \mathcal{D}(\mathcal{H}_B)$ be a quantum channel. Then, it holds that
    \bb\label{eq:interpr}
        \errnsa{\pazocal{N}} = \opnsa{\pazocal{N}} = U^\infty(\pazocal{N}).
    \ee
    That is, the activated, no signalling error exponent of $\pazocal{N}$ at zero rate is its regularised umlaut information. In particular, for classical-quantum channel $\pazocal{M}:\mathcal{P}(\mathcal{X})\to \mathcal{D}(\mathcal{H}_B)$ one has
    \bb\label{eq:interpr_cq}
        \errnsa{\pazocal{M}} = U(\pazocal{M}).
    \ee
\end{thm}
\end{boxed}

\begin{proof}
    The proof is divided into four steps.
    
    \noindent\textbf{First step: lower bound.} Let $\alpha\in(0,1)$. Then we can write
    \bb
        \errnsa{\pazocal{N}}&=
        \liminf_{r\to 0^+}\liminf_{n\to \infty}-\frac{1}{n}\log\epsilon^{\rm NS,a}(\new{\exp(rn)},\pazocal{N}^{\otimes n}\old{,rn})\\
        &\eqt{(i)}\liminf_{r\to 0^+}\liminf_{n\to \infty}\frac{1}{n}\sup_{\rho_{A'\,^n}}\min_{\sigma_{B^n}}D^{\old{\exp(-nr)}\new{\exp(-nr)}}_H\left(\rho_{A'\,^n}\otimes\sigma_{B^n}\middle\|\rho_{A'\,^nB^n}^{(\pazocal{N}^{\otimes n})}\right)\\
        &\geqt{(ii)}\liminf_{r\to 0^+}\liminf_{n\to \infty}\frac{1}{n}\min_{\sigma_{B^n}}D^{\old{\exp(-nr)}\new{\exp(-nr)}}_H\left(\rho_{A'}^{\otimes n}\otimes\sigma_{B^n}\middle\|\left(\rho_{A'B}^{(\pazocal{N})}\right)^{\otimes n}\right)\\
        &\geqt{(iii)}\liminf_{r\to 0^+}\liminf_{n\to \infty}\frac{1}{n}\min_{\sigma_{B^n}}\left(D_\alpha\left(\rho_{A'}^{\otimes n}\otimes\sigma_{B^n}\middle\|\left(\rho_{A'B}^{(\pazocal{N})}\right)^{\otimes n}\right)+\frac{\alpha nr}{1-\alpha} \right)\\
        &\eqt{(iv)} \liminf_{n\to \infty}\frac{1}{n} U_\alpha(A'\,^n;B^n)_{\left(\rho^{(\pazocal{N})}\right)^{\otimes n}}\\
        &\eqt{(v)} U_\alpha(A';B)_{\rho^{(\pazocal{N})}},
    \ee
    where in~(i) we have used  
    Proposition~\ref{thm:rev_meta}, in~(ii) we have made the ansatz $\rho_{A'\,^n}=\rho_{A'}^{\otimes n}$ for an arbitrary fixed state $\rho_{A'}\in\mathcal{D}(\mathcal{H}_{A'})$, noting that for this choice $\rho_{A'\,^nB^n}^{(\pazocal{N}^{\otimes n})}=\left(\rho_{A'B}^{(\pazocal{N})}\right)^{\otimes n}$, in~(iii) we have recalled the inequality~\eqref{eq:lowerDalpha}, 
    in~(iv) we have simply applied the definition of Petz--Rényi $\alpha$-umlaut information (see Definition~\ref{def:petz_lautum}) and in~(v) we have leveraged the additivity of $U_\alpha$, as in Proposition~\ref{prop:closed_alpha}.
    By the arbitrariness of the state $\rho_{A'}$ and of $\alpha\in(0,1)$, we can lower bound
    \bb
        \errnsa{\pazocal{N}}&\geq \sup_{\rho_{A'}}\limsup_{\alpha\to 1^-}U_\alpha(A';B)_{\rho^{(\pazocal{N})}}
        \eqt{(vi)}\sup_{\rho_{A'}}U(A';B)_{\rho^{(\pazocal{N})}}=U(\pazocal{N}),
    \ee
    where in~(vi) we have used Lemma~\ref{lem:alpha_to_1}.

    \noindent\textbf{Second step: upper bound.} By Proposition~\ref{thm:rev_meta}, we have
    \bb
        \opnsa{\pazocal{N}}
        &=\liminf_{\delta\to0}\liminf_{n\to \infty}-\frac{1}{n}\log\epsilon^{\rm NS,a}\left(\frac{1}{\delta},\pazocal{N}^{\otimes n}\right)
        \\&=
        \liminf_{\delta\to 0}\liminf_{n\to \infty}\frac{1}{n}\sup_{\rho_{A'\,^n}}\min_{\sigma_{B^n}} D^{\delta}_H\left(\rho_{A'\,^n}\otimes\sigma_{B^n}\middle\|\rho_{A'\,^nB^n}^{(\pazocal{N}^{\otimes n})}\right)\\
        & \leqt{(vii)} \lim_{\delta\to 0}\liminf_{n\to\infty}\sup_{\rho_{A'\,^n}}\min_{\sigma_{B^n}}\frac{1}{n}\frac{1}{1-\delta}\left(1+D\left(\rho_{A'\,^n}\otimes\sigma_{B^n}\middle\|\rho_{A'\,^nB^n}^{(\pazocal{N}^{\otimes n})}\right)\right)\\
        & =\liminf_{n\to\infty}\sup_{\rho_{A'\,^n}}\min_{\sigma_{B^n}}\frac{1}{n}D\left(\rho_{A'\,^n}\otimes\sigma_{B^n}\middle\|\rho_{A'\,^nB^n}^{(\pazocal{N}^{\otimes n})}\right)\\
        &=U^{\infty}(\pazocal{N}),
    \ee
    where in~(vii) we have used the upper bound~\cite{Hiai1991}
    \bb
        D_H^{\epsilon}(\rho\|\sigma)\leq \frac{1}{1-\epsilon}(1+D(\rho\|\sigma))\, ,
    \ee
    which is a simple consequence of the data processing inequality for the relative entropy, together with the definition of the hypothesis testing relative entropy in~\eqref{D_H}.
    
    \noindent\textbf{Third step: weak additivity of the error exponent.}
    
The third step is quite standard in (quantum) information theory: since  $\opnsa{\pazocal{N}}$ and $\errnsa{\pazocal{N}}$ are defined through a regularisation and are well-behaved with respect to the message size, they can be shown to be weakly additive:
\bb 
\opnsa{\pazocal{N}} = \lim_{n \to \infty} \frac{1}{n}\opnsa{\pazocal{N}^{\otimes n}},\quad \errnsa{\pazocal{N}} = \lim_{n \to \infty} \frac{1}{n}\errnsa{\pazocal{N}^{\otimes n}}\, .
\ee
We refer to Appendix~\ref{app:reg} for details on the proof of this claim.
    
    \noindent\textbf{Fourth step: regularisation.}
    Combining the lower and the upper bounds, we obtain that
    \bb\label{eq:chain_ineq}
        U^\infty(\pazocal{N})\, \geq\, \opnsa{\pazocal{N}}\ \geqt{(xiii)}\ \errnsa{\pazocal{N}}\, \geq\, U(\pazocal{N})\, ,
    \ee
    where~(xiii) is~\eqref{eq:Tast}. Regularising this inequality yields
    \bb
        U^\infty(\pazocal{N})\geq \opnsa{\pazocal{N}}\geq \errnsa{\pazocal{N}}\geq \lim_{n\to \infty}\frac{1}{n}U(\pazocal{N}^{\otimes n})\, ,
    \ee
    whence
    \bb
        U^\infty(\pazocal{N})\geq \opnsa{\pazocal{N}}\geq \errnsa{\pazocal{N}}\geq U^\infty(\pazocal{N})\, ,
    \ee
    and this proves~\eqref{eq:interpr}. Due to the additivity of the channel umlaut information of classical-quantum channels (Proposition~\ref{prop:add_CQ}), Eq.~\eqref{eq:interpr_cq} immediately follows.
\end{proof}

In order to prove directly~\eqref{eq:interpr_cq}, the third step in the previous proof is not needed because in~\eqref{eq:chain_ineq} the upper bound already coincides with the lower bound, as the umlaut information of classical-quantum channels is additive.

Before we move on, it is worth emphasising that the umlaut information can diverge to infinity in some cases --- indeed, this is fully expected, as it is known that channel coding with non-signalling assistance is possible with zero error for some channels~\cite{duan_2016,duan_2017}. Let us comment on this operational connection more.
If $U^\infty(\NN)=\infty$, the above Theorem~\ref{thm:exact_tho} indicates that classical communication assisted by activated non-signalling correlations can be carried out with an error that decreases super-exponentially in the number of channel uses.
However, comparing our exact characterisation of the set of channels for which $U^\infty(\NN)$ diverges in Lemma~\ref{lem:boundedness_regularised_umlaut} with the results of~\cite{duan_2016}, we can obtain a much closer relation with the setting of \emph{zero}-error coding --- that is, where decoding errors are not allowed whatsoever.
We summarise the connections in the following result. Since we do not focus on the zero-error setting in this paper, we refer to~\cite{duan_2016,duan_2017} for detailed definitions.
\begin{cor}
Let $C^{\rm NS}_0(\pazocal{N})$ denote the asymptotic zero-error classical capacity with non-signalling assistance~\cite{duan_2016}, and $C^{\rm NS,a}_0(\pazocal{N})$ the zero-error capacity with activated non-signalling assistance~\cite{duan_2017}. Then,
\bb
U^\infty(\pazocal{N}) = \infty \iff U(\pazocal{N}) = \infty \iff C^{\rm NS}_0(\pazocal{N}) > 0 \iff C^{\rm NS,a}_0(\pazocal{N}) > 0.
\ee
\end{cor}
\begin{proof}
The first equivalence follows directly from Lemma~\ref{lem:boundedness_regularised_umlaut}; the second follows by comparing the condition of Lemma~\ref{lem:boundedness_regularised_umlaut}(c) with~\cite[Theorem~25(i,iv)]{duan_2016}. This then implies that whenever $C^{\rm NS}_0(\pazocal{N}) = 0$, then $U^\infty(\NN) < \infty$, which by Theorem~\ref{thm:exact_tho} means that there cannot exist a zero-error activated non-signalling protocol: the best achievable error exponent, even for arbitrarily small rates, satisfies $\errnsa{\pazocal{N}} < \infty$.
Hence $C^{\rm NS,a}_0(\pazocal{N}) = 0$. As $C^{\rm NS,a}_0(\pazocal{N}) \geq C^{\rm NS}_0(\pazocal{N})$ in general, the result follows.
\end{proof}

This connection also complements~\cite[Theorem~8]{duan_2017}, where it was found that activation cannot increase the zero-error non-signalling--assisted capacity for channels such that $C^{\rm NS}_0(\pazocal{N}) > 0$ --- we now see that this is true also when the capacity is zero, implying that $C^{\rm NS,a}_0(\pazocal{N}) = C^{\rm NS}_0(\pazocal{N})$ for all channels $\pazocal{N}$.


\subsection{A family of lower bounds for the case of classical-quantum channels}

For classical channels, the following quantities attain an operational meaning in list decoding~\cite{Filippo25}. Here, we explore them for classical-quantum channels.

\begin{Def}\label{def:ell} Let $k\geq 1$ and let $q\in \mathbb{R}^k$ be any vector such that $\sum_{i=1}^kq_i=1$. Let $\pazocal{N}_{X\to B}$ be a classical-quantum channel and let $P_X\in\mathcal{P}(\mathcal{X})$. Then we define the $(k,q)$-\textit{lower umlaut information} of $\pazocal{N}$ with input distribution $P_X$ as
\bb\label{eq:def_ell}
    \ell_{k,q}(\pazocal{N},P_X)\coloneqq \sum_{x_1,\dots,x_k\in\mathcal{X}}P_{X}(x_1)\cdots P_{X}(x_k)\left(-\log \Tr\exp\left(\sum_{i=1}^kq_i\log\rho_{x_i}^B\right)\right).
\ee
\end{Def}

\begin{prop}[(Lower bound to the umlaut information for classical-quantum channels)] \label{prop:inequality} Let $k\geq 1$ and let $q\in \mathbb{R}^k$ be any vector such that $\sum_{i=1}^kq_i=1$. Let $\pazocal{N}_{X\to B}$ be a classical-quantum channel and let
\bb\label{eq:converse_l}
    \rho_{X'B}^{(\pazocal{N})}\coloneqq ({\rm Id}_{X'}\otimes\pazocal{N}_{X\to B})\left(\left(P_{X'}^{1/2}\otimes\id_A\right)\Phi_{X'X}\left(P_{X'}^{1/2}\otimes\id_A\right)\right).
\ee
Then, we have
    \bb\label{eq:ineq_lautum}
        U(X';B)_{\rho^{(\pazocal{N})}}\geq \ell_{k,q}(\pazocal{N},P_{X'}).
    \ee
    Moreover, let us now suppose that all the states $\{\rho_x^B\}_{x\in\mathcal{X}}$ have the same support. Calling $u_k\in\mathbb{R}^k$ the uniform vector $u_k=(1/k,\dots,1/k)$, we have
    \bb\label{eq:achievability_l}
        U(X';B)_{\rho^{(\pazocal{N})}}= \lim_{k\to\infty}\ell_{k,u_k}(\pazocal{N},P_{X'})
    \ee
and in particular
    \bb\label{eq:ineq_lautum_channel}
        U(\pazocal{N})\geq \sup_{P_{X'}}\ell_{k,q}(\pazocal{N},P_{X'})\,\quad\text{as well as}\quad
        U(\pazocal{N})= \sup_{P_{X'}}\lim_{k\to\infty}\ell_{k,u_k}(\pazocal{N},P_{X'}).
    \ee
    More precisely, the following quantitative estimate holds
    \bb
    U(X';B)_{\rho^{(\pazocal{N})}}-\ell_{k,u_k}(\pazocal{N},P_{X'})\leq 
    \begin{cases}
        \frac{3}{k^{1/3}}\left(\frac{1}{2}U(X';B)_{\rho_{X'B}^{(\pazocal{N})}}^2\lambda_{\max}^{\rm fin}(R_{P_{X'}})\right)^{2/3} & P^\ast<P_{\min}\, ,\\[2ex]
        \frac{1}{kP_{\min}^2}U(X';B)_{\rho_{X'B}^{(\pazocal{N})}}+P_{\min}\lambda_{\max}^{\rm fin}(R_{P_{X'}}) & P^\ast\geq P_{\min}\, ,
    \end{cases}
\ee
where $\lambda_{\max}^{\rm fin}(R_{P_{X'}})$ is the largest finite eigenvalue of
    \bb
        R_{P_{X'}}\coloneqq-\sum_{x\in{\rm supp}(P_{X'})}\log\rho_x^B
    \ee 
and
\bb
      P_{\min}&\coloneqq\min\{P_X(x)\,:\,x\in{\rm supp}(P_X) \}\qquad
    P^\ast \coloneqq \frac{1}{k^{1/3}}\sqrt[3]{ \frac{2U(X';B)_{\rho_{X'B}^{(\pazocal{N})}}}{\lambda_{\max}^{\rm fin}(R_{P_X})}}.
\ee 
\end{prop}

\begin{proof} This proof generalises the strategy used in the classical case~\cite[Propositions 31, 32 and 34]{Filippo25}. Let us compute
    \bb   U(X';B)_{\rho^{(\pazocal{N})}}&=\min_{\sigma_B}D(p_{X'}\otimes\sigma_B\|\rho_{X'B}^{(\pazocal{N})})\\
        &\eqt{(i)}\min_{\sigma_B}\sum_{x\in\mathcal{X}}P_{X'}(x)D\big(\sigma_B\|\rho^B_x\big)\\
        &\eqt{(ii)}\min_{\sigma_B}\sum_{x\in\mathcal{X}}P_{X'}(x)\sum_{i=1}^kq_iD\big(\sigma_B\|\rho^B_x\big)\\
        &\eqt{(iii)}\min_{\sigma_B}\sum_{x_1,\dots,x_k\in\mathcal{X}}P_{X'}(x_1)\cdots P_{X'}(x_k)\sum_{i=1}^kq_iD\big(\sigma_B\|\rho^B_{x_i}\big)\\
        &\geq \sum_{x_1,\dots,x_k\in\mathcal{X}}P_{X'}(x_1)\cdots P_{X'}(x_k)\min_{\sigma_B}\sum_{i=1}^kq_iD\big(\sigma_B\|\rho^B_{x_i}\big)\\
        &=\sum_{x_1,\dots,x_k\in\mathcal{X}}P_{X'}(x_1)\cdots P_{X'}(x_k)\min_{\sigma_B}\left(-S(\sigma_B)-\Tr\left[\sigma_B\sum_{i=1}^kq_i\log \rho_{x_i}^B\right]\right)\\
        &\eqt{(vi)}\sum_{x_1,\dots,x_k\in\mathcal{X}}P_{X'}(x_1)\cdots P_{X'}(x_k)\left(-\log \Tr\exp\left(\sum_{i=1}^kq_i\log\rho_{x_i}^B\right)\right),
    \ee
    where in (i) we have recalled the definition $\rho_{A'B}^{(\pazocal{N})}$, in (ii) we have used that $\sum_{i=1}^kq_i=1$, similarly in (iii) we have noticed that $\sum_{x\in\mathcal{X}}P_{X'}(x)=1$, and in (vi) we have used the Gibbs variational principle (Lemma~\ref{lemma:Gibbs}). This proves~\eqref{eq:ineq_lautum}. To get~\eqref{eq:ineq_lautum_channel}, we just take the maximum over $P_{X'}\in\mathcal{P}(\mathcal{X})$, according to Definition~\ref{def:channel_lautum}.\\
    Let us now assume that all $\{\rho_x^B\}_{x\in\mathcal{X}}$ have the same support. Since $\exp(-\infty)=0$, without loss of generality we can assume that the states $\{\rho_x^B\}_{x\in\mathcal{X}}$ have full support. Let us call now consider the uniform vector $u_k=(1/k,\dots,1/k)$ in $\mathbb{R}^k$.
     We can represent the lower umlaut information $\ell_{k,q}$ as the expectation value
     \bb
        \ell_{k,u_k}(\pazocal{N},P_X)= \mathbb{E}_{X^k\sim P_X^k}\left[-\log \Tr\exp\left(\sum_{i=1}^k\frac{1}{k}\log\rho_{X_i}^B\right)\right]
    \ee
   where $P_X^k(x_1,\dots,x_k)=P_X(x_1)\cdots P_X(x_1)$. It is easy to see that the result of the sum $\sum_{i=1}^k\frac{1}{k}\log\rho_{X_i}^B$ depends only on the number of occurrences of each symbol $x\in\XX$ in the string $X^n$, namely
    \bb
        \sum_{i=1}^k\frac{1}{k}\log\rho_{X_i}^B=\sum_{x\in\XX}\frac{N(x|X^k)}{k}\log\rho_{x}^B
    \ee
    where $N(x|X^k)\coloneqq\sum_{i=1}^k\id_{X_i=x}$ is the number of occurrences of $x$ in $X^k$. Let us introduce
    \bb
        \hat N_k^x\coloneqq \frac{N(x|X^k)}{k}=\frac{1}{k}\sum_{i=1}^k\id_{X_i=x}
    \ee
    which, by the weak law of the large numbers, converges in probability to its expectation value $\mathbb{E}\left[\hat N_k^x\right]=P_X(x)$. More precisely, let us fix $0<\epsilon<\min\{P_X(x)\,:\,x\in{\rm supp}(P_X) \}$, and let us call $\mathcal{E}_k^{(\epsilon)}$ the event $\left\{\exists\, x\in{\rm supp}(P_X):\big|\hat{N}_k^x-P_X(x)\big|>\epsilon \right\}$ and $\big(\mathcal{E}_k^{(\epsilon)}\big)^c$ its complement. Then
    \bb\label{eq:law_large_numbers}
        \lim_{k\to \infty}\mathbb{P}\left(\left(\mathcal{E}_k^{(\epsilon)}\right)^c\right)=1.
    \ee If $x\notin {\rm supp} (P_X)$, then $\mathbb{P}(\hat N^x_k\neq 0 )=0$. We can therefore rewrite
    \bb\label{eq:to_be_lower_bounded}
        \ell_{k,u_k}(\pazocal{N},P_X)&= \mathbb{E}_{X^k\sim P_X^k}\left[-\log \Tr\exp\left(\sum_{x\in{\rm supp}(P_X)}\hat N_k^x\log\rho_{x}^B\right)\right]
    \ee
    Lat us define, for $v\in [0,1]^{|\XX|}\setminus\{0\}$, where $0$ is the null vector of $\mathbb{R}^{|\XX|}$, the map
    \bb
        f(v)\coloneq -\log \Tr\exp\left(\sum_{x\in{\rm supp}(v)}v(x)\log\rho_x^B\right).
    \ee
    Let us call $u(x)=v(x)/\|v\|_1\in\mathcal{P}(\XX)$. Then
    \bb\label{eq:f_rel_ent}
        f(v)&= -\log \Tr\exp\left(\sum_{x\in{\rm supp}(v)}\frac{v(x)}{\|v\|_1}\log\rho_x^B\right)-\log \|v\|_1\\
        &\eqt{(a)}\min_{\sigma_B}D\left(u_{X'}\otimes\sigma_B\middle\|u_{X'}^{1/2}J_{X'B}^{(\pazocal{N})}u_{X'}^{1/2}\right)-\log \|v\|_1,
    \ee
    where in (a) we have used~\eqref{eq:formula_without_max} and introduced the Choi--Jamio\l kowski matrix $J_{X'B}^{(\pazocal{N})}$ of $\mathcal{N}$.
    The map $f$ has the following properties.
    \begin{enumerate}
        \item \textbf{Componentwise increasing monotonicity}. This is an immediate consequence of the combination of these observations: $\log \rho_x^B\leq 0$, so the linear combination $\sum_{x\in{\rm supp}(v)}v(x)\log\rho_x^B$ is a matrix which is antimonotone in the components of $v$; $\Tr\exp(\,\cdot\,)$ is a matrix monotone function and $-\log(\,\cdot\,)$ is an antimonotone function.
        \item \textbf{Continuity.} Since all the states $\{\rho_x^B\}_{x\in\mathcal{X}}$ have full support, $v\mapsto f(v)$ is manifestly a continuous function, being the composition of continuous functions.
        \item \textbf{Lower bound.} By~\eqref{eq:f_rel_ent}, we also notice that $f(v)\geq -\log\|v\|_1$ due to the positivity of the quantum relative entropy.
    \end{enumerate} 
    Now, considering $N_k=(N_k^x)_{x\in\XX}$ as an element of $[0,1]^{|\XX|}\setminus\{0\}$, we lower bound~\eqref{eq:to_be_lower_bounded}:
    \bb\label{eq:quant_est}
        \ell_{k,u_k}(\pazocal{N},P_X)&= \mathbb{E}_{X^k\sim P_X^k}\left[f\left(\hat N_k\right)\right]\\
        &\eqt{(v)}\mathbb{E}_{X^k\sim P_X^k}\left[f\left(\hat N_k\right)\,\middle|\,\mathcal{E}_k^{(\epsilon)}\,\right]\mathbb{P}\left(\mathcal{E}_k^{(\epsilon)}\right)
        +\mathbb{E}_{X^k\sim P_X^k}\left[f\left(\hat N_k\right)\,\middle|\,(\mathcal{E}_k^{(\epsilon)})^c\,\right]\mathbb{P}\left((\mathcal{E}_k^{(\epsilon)})^c\right)\\
        &\geqt{(vi)}f\left( P_X^{(\epsilon)}\right)\mathbb{P}\left((\mathcal{E}_k^{(\epsilon)})^c\right)
    \ee
    where 
    \begin{itemize}
        \item in (v) we condition on $\mathcal{E}_k^{(\epsilon)}$ and on its complement;
        \item in (vi) we lower bound the first term with $-\log\|\hat N_k\|_1=0$ and we introduce, for the second term,
        \bb
            P_X^{(\epsilon)}\coloneqq \begin{cases}
                P_X(x)-\epsilon & x\in{\rm supp}(P_X)\\
                0 & {\rm otherwise}
            \end{cases},
        \ee
        which is componentwise  smaller than $P_X$, so we can leverage the monotoniciy of $f$.
    \end{itemize}
    Taking the liminf, we get
    \bb
        \liminf_{k\to\infty}\ell_{k,u_k}(\pazocal{N},P_X)\geq f\left( P_X^{(\epsilon)}\right)\liminf_{k\to\infty}\mathbb{P}\left((\mathcal{E}_k^{(\epsilon)})^c\right)=f\left( P_X^{(\epsilon)}\right);
    \ee
    where the last identity holds due to~\eqref{eq:law_large_numbers}. By arbitrariness of $0<\epsilon<\min\{P_X(x)\,:\,x\in{\rm supp}(P_X) \}$, we let $\epsilon\to 0^+$:
    \bb
        \liminf_{k\to\infty}\ell_{k,u_k}(\pazocal{N},P_X)&\geq \liminf_{\epsilon\to 0^+}f\left( P_X^{(\epsilon)}\right)\\
        &\eqt{(vii)} f\left( P_X\right)= -\log \exp\left(\sum_{x\in\XX}P_X(x)\log\rho_x^B\right)=U(X';B)_{\rho^{(\pazocal{N})}},
    \ee
    where (vii) holds because of the continuity of $f$.
    Combining this result with the converse bound~\eqref{eq:converse_l}, we conclude that
    \bb
        U(X';B)_{\rho^{(\pazocal{N})}}\geq \liminf_{k\to\infty}\ell_{k,u_k}(\pazocal{N},P_X) \geq U(X';B)_{\rho^{(\pazocal{N})}}.
    \ee
    Now we could get a quantitative estimate by refining the previous argument. By Chebyshev's inequality
\bb
    \mathbb{P}\left(\big|N_k^x-P_X(x)\big|>\epsilon\right)\leq \frac{{\rm Var}[N_k^x]}{\epsilon^2}=\frac{{\rm Var}_{X\sim P_X}[\id_{X=x}]}{k\epsilon^2}=\frac{P_X(x)(1-P_X(x))}{k\epsilon^2}
\ee
and of the union bound
\bb
    \mathbb{P}\left(\exists\,x\in\XX : \big|N_k^x-P_X(x)\big|>\epsilon\right)\leq \sum_{x\in\XX}\frac{P_X(x)(1-P_X(x))}{k\epsilon^2}\leq \frac{1}{k\epsilon^2}.
\ee
Then we can further lower bound~\eqref{eq:quant_est}
    \bb
        \ell_{k,u_k}(\pazocal{N},P_X)&\geq f\left( P_X^{(\epsilon)}\right)\mathbb{P}\left((\mathcal{E}_k^{(\epsilon)})^c\right)\\
        &\geq f\left( P_X^{(\epsilon)}\right)\left(1-\frac{1}{k\epsilon^2}\right)\\
        &= \left(1-\frac{1}{k\epsilon^2}\right)\left(-\log\Tr\exp\left\{\sum_{x\in{\rm supp}(P_X)}P_X(x)\log\rho_x^B-\epsilon\sum_{x\in{\rm supp}(P_X)}\log\rho_x^B\right\}\right)\\
        &\geqt{(viii)} \left(1-\frac{1}{k\epsilon^2}\right)\left(-\log \Tr \exp\left(\sum_{x\in{\rm supp}(P_X)}P_X(x)\log\rho_x^B\right)-\epsilon\lambda_{\max}^{\rm fin}(R_{P_X})\right)\\
        &\geqt{(ix)} \left(1-\frac{1}{k\epsilon^2}\right)U(X';B)_{\rho^{(\pazocal{N})}}-\epsilon\lambda_{\max}^{\rm fin}(R_{P_X}).
    \ee
    In (viii) we have bounded
    \bb\label{eq:golden}
        \Tr[\exp[A+B]]&\leqt{(b)} \Tr[\exp[A]\exp[B]] \\
        &= \Tr[\exp[A/2]\exp[B]\exp[A/2]] \\
        &\leqt{(c)} \Tr[\exp[A/2](\lambda_{\max}(\exp[B])\id)\exp[A/2]]\\
        &= \exp(\lambda_{\max}(B))\Tr[\exp[A]],
    \ee
    with 
    \bb
        A=\sum_{x\in{\rm supp}(P_X)}P_X(x)\log\rho_x^B\qquad B=-\epsilon\sum_{x\in{\rm supp}(P_X)}\log\rho_x^B;
    \ee
    in (b) we have used the Golden--Thompson inequality~\cite{Golden1965, Thompson1965}, namely 
    \bb\label{eq:golden_t}
        \Tr[\exp[A+B]]\leq \Tr[\exp[A]\exp[B]]
    \ee for any pair of Hermitian matrices $A$ and $B$, and in (c) we have noticed that $\exp[B]\leq \lambda_{\max}(\exp[B])\id$. In (ix) we have introduced the largest finite eigenvalue $\lambda_{\max}^{\rm fin}$ of
    \bb
        R_{P_X}\coloneqq-\sum_{x\in{\rm supp}(P_X)}\log\rho_x^B.
    \ee
    The reason why we need to consider only finite eigenvalues is the following: it is clearly the case if the states $\{\rho_x^B\}_{x\in\mathcal{X}}$ have full support. If they have the same support, restricting to their support means that we need to consider only the finite eigenvalues of $R_{P_X}$.
    Therefore, we get
    \bb\label{eq:bound_U_l2}
        U(X';B)_{\rho^{(\pazocal{N})}}-\ell_{k,u_k}(\pazocal{N},P_X)\leq \frac{1}{k\epsilon^2}U(X';B)_{\rho^{(\pazocal{N})}}+\epsilon\lambda_{\max}^{\rm fin}(R_{P_X})
    \ee
    
    Now, for any $a,b>0$, it easy to prove that the function $g(\epsilon)\coloneqq a/\epsilon^2 +b\epsilon$, defined for $\epsilon>0$, has a unique minimum for $\epsilon=\epsilon^\ast$, where
\bb
    \epsilon^\ast = \sqrt[3]{\frac{2a}{b}}>0
\ee
Therefore, by a simple argument about the monotonicity of $g$, we have
\bb
    \inf\{g(\epsilon): 0<\epsilon<\epsilon_0\}=
    \begin{cases}
        g(\epsilon^\ast)=3a^{1/3}\left(\frac{b}{2}\right)^{2/3} & \epsilon^\ast<\epsilon_0\\
        g(\epsilon_0) & \epsilon^\ast\geq \epsilon_0
    \end{cases},
\ee
whence, minimising the RHS of~\eqref{eq:bound_U_l2} and calling 
\bb
    P_{\min}&\coloneqq\min\{P_X(x)\,:\,x\in{\rm supp}(P_X) \}\qquad
    P^\ast \coloneqq \sqrt[3]{ \frac{2U(X';B)_{\rho_{X'B}^{(\pazocal{N})}}}{k\lambda_{\max}^{\rm fin}(R_{P_X})}}
\ee 
we get
\bb
    U(X';B)_{\rho^{(\pazocal{N})}}-\ell_{k,u_k}(\pazocal{N},P_X)\leq 
    \begin{cases}
        \frac{3}{k^{1/3}}\left(\frac{1}{2}U(X';B)_{\rho^{(\pazocal{N})}}^2\lambda_{\max}^{\rm fin}(R_{P_X})\right)^{2/3} & P^\ast<P_{\min}\, ,\\[2ex]
        \frac{1}{kP_{\min}^2}U(X';B)_{\rho^{(\pazocal{N})}}+P_{\min}\lambda_{\max}^{\rm fin}(R_{P_X}) & P^\ast\geq P_{\min}\, ,
    \end{cases}
\ee
and this completes the proof.
\end{proof}

The functions $\ell_{k, q}$ are closely related to the zero-rate unassisted error exponent. Indeed, for a classical channel $\pazocal{W}$ and uniform $q=u_{L+1}=(\frac{1}{L+1}, \dots, \frac{1}{L+1})$, the quantity $\sup_{P_X} \ell_{L+1, u_{L+1}}(\pazocal{W}, P_X)$ precisely corresponds to the zero-rate unassisted error exponent of the channel $\pazocal{W}$ when the decoder is allowed to return a list of size $L$ and the error occurs only if the original message is not included in that list~\cite{Blinovsky2001Oct}. Moreover, the case  $L=1$ corresponds to the standard channel coding in which case $\sup_{P_X} \ell_{2, u_2}(\pazocal{N}, P_X)$ can be lower bounded with the zero-rate unassisted error exponent of a classical-quantum channel $\pazocal{N}$~\cite{Holevo2000,Dalai_2013} by means of the Golden--Thompson inequality~\cite{Golden1965, Thompson1965}. 

\begin{prop}[\cite{Holevo2000,Dalai_2013}]
\label{thm:cq-0-rate}
    Let $\pazocal{N}_{X\to B}(\cdot) = \sum_{x\in \mathcal{X}} \bra{x}\cdot \ket{x} \rho_x^B$ be a classical-quantum channel. Then, we have
\bb
\errpl{\pazocal{N}} = \sup_{P_X\in \mathcal{P}(\mathcal{X})} \sum_{x_1,x_2\in\mathcal{X}}P_{X}(x_1)P_{X}(x_2)\left(-\log \Tr\left[\sqrt{\rho_{x_1}^B}\sqrt{\rho_{x_2}^B}\right]\right).
\ee 
\end{prop}

We include a slightly simplified, self-contained proof in terms of modern quantum information theory tools in Appendix~\ref{app:cq-0-rate}.

Now, since non-signalling--assisted coding strategies are at least as powerful as their unassisted counterparts, the corresponding zero-rate exponents inherit this ordering. However, proving such an inequality directly turns out to be less straightforward. 

\begin{cor}\label{cor:golden} Let $\pazocal{N}_{X\to B}$ be a classical-quantum channel. Then, we have
\bb
    U(\pazocal{N}) = \errnsa{\pazocal{N}}\geq \errpl{\pazocal{N}}.
\ee 
\end{cor}

\begin{proof}
    Let us set  $k=2$ and $q=(1/2,1/2)$ in Proposition~\ref{prop:inequality}. Then 
    \bb
        U(X;B)_{\rho^{(\pazocal{N})}}&\geq \sum_{x_1,x_2\in\mathcal{X}}P_{X}(x_1)P_{X}(x_2)\left(-\log \Tr\exp\left(\frac{1}{2}\log\rho_{x_1}^B+\frac{1}{2}\log\rho_{x_2}^B\right)\right)\\
    &\geqt{(i)}\sum_{x_1,x_2\in\mathcal{X}}P_{X}(x_1)P_{X}(x_2)\left(-\log \Tr\left[\exp\left(\frac{1}{2}\log\rho_{x_1}^B\right)\exp\left(\frac{1}{2}\log\rho_{x_2}^B\right)\right)\right]\\
    &=\sum_{x_1,x_2\in\mathcal{X}}P_{X}(x_1)P_{X}(x_2)\left(-\log \Tr\left[\sqrt{\rho_{x_1}^B}\sqrt{\rho_{x_2}^B}\right]\right),
    \ee
    where in (i) we have used the Golden--Thompson inequality~\cite{Golden1965, Thompson1965}, see \eqref{eq:golden_t}.
    By taking the maximum over $P_{X}\in\mathcal{P}(\mathcal{X})$, we conclude the proof.
\end{proof}

\subsection{Umlaut information and list decoding}

We discuss the extension of channel coding to list decoding in order to derive a further result regarding the channel umlaut information. Given a set of messages $\mathcal{M}=\{1,\dots, M\}$ and a quantum channel $\pazocal{N}_{A\to B}$, an $L$-list unassisted code is given by an encoder $\pazocal{E}:\mathcal{M}\to \mathcal{D}(\mathcal{H}_A)$ and a decoder $\pazocal{D}: \mathcal{D}(\mathcal{H}_B)\to \pazocal{P}_L(\mathcal{M})$, where $\pazocal{P}_L(\mathcal{M})$ the family of subsets of $\mathcal{M}$ with cardinality $L$:
\bb
    \pazocal{P}_L(\mathcal{M})\coloneqq\{ A\subseteq \mathcal{M}\,:\, |A|=L\}.
\ee
Calling $W$ the source of messages, i.e.~a random variable taking values in $\mathcal{M}$, which is assumed to be uniform, let $\hat W_L$ be the random output of the code taking values in $\pazocal{P}_L(\mathcal{M})$. An error is declared only if the original message does not belong to that list, i.e.\ $W\notin \hat W_L$. Therefore, the error probability is
\bb
    \epsilon_L^{\emptyset}(M, \pazocal{N})\coloneqq \min_{\substack{(\pazocal{E},\pazocal{D})\; L\text{-list}\\ \text{unassisted code}}}\mathbb{P}\left(W\notin \pazocal{D}\circ\pazocal{N}^{\otimes n}\circ\pazocal{E}(W)\right). 
\ee
The $L$-list unassisted error exponent with communication rate $r$ is defined as
\bb
E_L^\emptyset(r,\pazocal{N})\coloneqq \liminf_{n\to\infty}-\frac{1}{n}\log\epsilon_L^{\emptyset}(L\exp(rn), \pazocal{N})\, .
\ee
Note that we require a message size of size $L\exp(rn)$ rather than simply $\exp(rn)$, as we want to interpret $r$ as a communication rate, and it is common to define rates as the logarithm of the message size divided by $L$ when considering list decoding schemes. See for example~\cite{SHANNON196765}. However, since we anyway keep the list size finite, this has no impact on the value of $E_L^\emptyset(r,\pazocal{W})$ (see also~\cite[footnote on p.~1]{bondaschi2021}).

Now, the $L$-list unassisted zero-rate error exponent of $\pazocal{N}$ is defined as
\bb\label{eq:zero-rate-list}
    \errL{\pazocal{N}}&\coloneqq \lim_{r\to 0^+} E_L^\emptyset(r,\pazocal{N}).
\ee

Differently from the classical case, a closed-form expression for the zero-rate unassisted error exponent of a classical-quantum channel is not known for $L>1$. In the fully classical case, the sequence $\ell_{L+1,u_{L+1}}$ matches such exponents for all $L\geq 1$~\cite{Filippo25} and, in the classical-quantum case, a trace inequality --- specifically the Golden--Thompson inequality --- is sufficient to recover the right exponent for $L=1$ (see the proof of Corollary \ref{cor:golden}). The main result of this section is a converse bound on the list decoding error exponent in terms of the non-signalling--assisted, activated error exponent without the knowledge of an explicit form expression for the former. The proof is completely based on an operational argument.

\begin{boxed}{}
\begin{thm}[(Umlaut information as a converse for list decoding)]\label{thm:converse_umlaut}
    Let $\pazocal{N}_{A\to B}$ be an arbitrary quantum channel and let $L\geq 1$. Then,
    \bb
        E_L^\emptyset(r,\pazocal{N})\leq E^{\rm NS,a}(r,\pazocal{N}).
    \ee
    In particular,
    \bb
        \errL{\pazocal{N}}\leq U^{\infty}(\pazocal{N}),
    \ee
    which implies, for any arbitrary CQ channel $\pazocal{M}$,
    \bb\label{eq:list_CQ}
        \errL{\pazocal{M}}\leq U(\pazocal{M}).
    \ee
\end{thm}
\end{boxed}

The proof of the theorem relies on a pioneering result in communication complexity by Wim van Dam \cite{implausible}, establishing that the availability of a special type of non-signalling correlation known as `Popescu--Rohrlich box' (PR-box)~\cite{PR-boxes} induces a collapse in the communication complexity of an arbitrary bipartite function.

\begin{prop}[(Collapse of communication complexity via PR-boxes \cite{implausible})]\label{prop:van_dam}
    Let $k\geq 1$ and let \bb f:\{0,1\}^k\times\{0,1\}^k\to \{0,1\}\ee  be an arbitrary Boolean function. Suppose that Alice and Bob receive $x\in \{0,1\}^k$ and $y\in \{0,1\}^k$, respectively. If they are allowed to use a PR-box as many times as they want and Alice can communicate exactly one bit to Bob, then Bob can compute $f(x,y)$ with success probability 1.
\end{prop}

Now we have all the ingredients to prove Theorem \ref{thm:converse_umlaut}.

\begin{proof}[Proof of Theorem \ref{thm:converse_umlaut}.]
    In the setting of the $L$-list decoding task, let $r>0$ and let $M\coloneqq L\exp(\ceil{nr})$. Let $(\pazocal{E},\pazocal{D})$ be an arbitrary $(M,\epsilon)$-code for $\pazocal{N}^{\otimes n}$, namely
    \bb
        \epsilon^\emptyset_L(M,\pazocal{N})\leq \epsilon
    \ee
    where, calling $\mathcal{M}\coloneqq\{1,\dots, M\}$,
    $\pazocal{E}:\mathcal{M}\to \mathcal{D}(\mathcal{H}_A)$ is the encoder and $\pazocal{D}:\mathcal{D}(\mathcal{H}_B)\to \pazocal{P}_L(\mathcal{M})$ is the $L$-list decoder.
    Let $k=\ceil{\log M}$ and let $x_m\in\{0,1\}^k$ be the binary representation of $m\in\mathcal{M}$. We define the Boolean function $f:\{0,1\}^{k}\times \{0,1\}^{k}\to \{0,1\}$ as
    \bb
        f(x,y)=\begin{cases}
            1 & \text{if}\quad x=y,\\
            0 & \text{if}\quad x\neq y.
        \end{cases}
    \ee

    By Proposition \ref{prop:van_dam}, given any arbitrary (fixed) $x,y\in \{0,1\}^k$, Bob can compute $f(x,y)$ just sharing a PR-box with Alice and receiving a noiseless bit from her. The following procedure describes a non-signalling--assisted, activated $(\exp(\ceil{nr}),\epsilon')$-code for $\pazocal{N}^{\otimes n}$, with $\epsilon'\leq\epsilon$.
    \begin{enumerate}
        \item Alice chooses a message $m\in\mathcal{M}$ and encodes it using $\pazocal{E}$. Then she sends the encoded message through the channel $\pazocal{N}^{\otimes n}$. Furthermore, she decides that she will always choose $x_m$ as the first input of $f$ when Bob computes the function.
        \item Bob decodes the output of the channel using $\pazocal{D}$. Now, Bob has a list $\hat W_L\subseteq \mathcal{M}$ of $L$ guesses of the original message $m$.
        \item Let us suppose that $m\in \hat W_L$. Then, if Bob has access to the function $f$, he can identify $m$ with at most $L$ queries of $f$: he just needs to compute $f(x_m,y_{m'})$ for all $m'\in\hat W_L$ and to see when $f(x_m,y_{m'})=1$. So, we can leverage Proposition \ref{prop:van_dam} to ensure that, allowing Alice to send $L$ noiseless bits to Bob, Bob can compute $f(x_m,y_{m'})$ for $L$ different $m'\in\mathcal{M}$. In this case, Bob correctly guesses $m$.
        \item If $m\notin \hat W_L$, Bob can do the same procedure of the previous point, but he will notice that $f(x_m,y_{m'})=0$ at each of the $L$ queries. In this case, Bob chooses as a guess for $m$ a uniform element in $\mathcal{M}$.
    \end{enumerate}
    This procedure shows that it is possible to build a non-signalling--assisted, activated $(\ceil{2^{nr}},\epsilon')$-code --- with $\epsilon'\leq\epsilon$ --- for $\pazocal{N}^{\otimes n}$ starting from a $(L\ceil{2^{nr}},\epsilon)$ code for the $L$-list decoding task. In particular,
    \bb
       \epsilon^\emptyset_L(L\exp(\ceil{rn}),\pazocal{N})\geq  \epsilon^{\rm NS}( \exp(\ceil{rn})\times L,\pazocal{N}\otimes\text{Id}_{L})\geq \epsilon^{\rm NS,a}(\exp(\ceil{rn}),\pazocal{N}),
    \ee
    whence
    \bb
    E_L^\emptyset(r,\pazocal{W})&= \liminf_{n\to\infty}-\frac{1}{n}\log\epsilon_L^{\emptyset}(L\exp(\ceil{rn}), \pazocal{N}^{\otimes n})\\
    &\leq \liminf_{n\to\infty}-\frac{1}{n}\log\epsilon^{\rm NS,a}(\exp(\ceil{rn}), \pazocal{N}^{\otimes n})=E^{\rm NS,a}(r,\pazocal{N}).
    \ee
    In particular, by taking the limit $r\to 0^+$, we have
    \bb
    E_L^\emptyset(0^+,\pazocal{W})\leq E^{\rm NS,a}(0^+,\pazocal{N})\eqt{(i)}U^\infty(\pazocal{N}),
    \ee
    where (i) follows from Theorem \ref{thm:exact_tho}. Finally, \eqref{eq:list_CQ} follows from Proposition \ref{prop:add_CQ}.
\end{proof}

\begin{cj} For any arbitrary fully classical channel, we also have \cite{Filippo25}
\bb
    \sup_{L\geq 1}\errL{\pazocal{W}}=\lim_{L\to \infty}\errL{\pazocal{W}}= U(\pazocal{W}).
\ee
We conjecture that $\lim_{L\to \infty}\errL{\pazocal{M}}\eqt{?} U(\pazocal{M})$ also for CQ channels $\mathcal{M}$.
\end{cj}


\subsection{Geometric umlaut information}

In the Section \ref{sec:operational_int}, we have proved that the error exponent at zero rate of CQ channels is given by the umlaut information of the channel, while for generic channels it is given by the regularised umlaut information of the channel: a single-letter formula is not known yet. 
In this section we are going to introduce a different notion of quantum umlaut information, which relies on the Belavkin--Staszewski quantum relative entropy. Such a quantity, in the setting of quantum communication, provides an upper bound to the error exponent at zero rate of generic quantum channels. We recall that the Belavkin--Staszewski (BS) quantum relative entropy is defined as
\bb
    D_{BS}(\rho\|\sigma)\coloneqq \Tr\left[\rho\log\left(\rho^{1/2}\sigma^{-1}\rho^{1/2}\right)\right],
\ee
which, similarly to Umegaki relative entropy, reduces to the Kullback--Leibler divergence for classical states. Leveraging this property, it is interesting to consider the following definition.

\begin{Def}[(Geometric umlaut information)] Given a bipartite state $\rho_{AB}\in\mathcal{D}(\mathcal{H}_A\otimes\mathcal{H}_B)$, we define its {\rm geometric umlaut information} as 
\bb
U_{BS}(A;B)_{\rho}&\coloneqq \min_{\sigma_B}D_{BS}(\rho_A\otimes\sigma_B\|\rho_{AB}),
\ee
where $\sigma_B\in \mathcal{D}(\mathcal{H}_B)$ and $\rho_A=\Tr_B[\rho_{AB}]$.
\end{Def}

\begin{lemma}\label{lem:BS_ineq}
Given any bipartite state $\rho_{AB}\in\mathcal{D}(\mathcal{H}_A\otimes\mathcal{H}_B)$,
\bb
    U(A;B)_{\rho}\leq U_{BS}(A;B)_{\rho}
\ee
\end{lemma}

\begin{proof} Follows immediately from the inequality $D_{BS}(\rho\|\sigma)\geq D(\rho\|\sigma)$, which holds for any pair of states $\rho$ and $\sigma$~\cite[Corollary~2.6]{Hiai1991}.
\end{proof} 

\begin{lemma}[(Data processing inequality)] \label{lem:data_processing_bs} Let $\Lambda_{A\to A'}:\mathcal{L}(\mathcal{H}_A)\to \mathcal{L}(\mathcal{H}_{A'})$ and $\Lambda'_{B\to B'}:\mathcal{L}(\mathcal{H}_B)\to \mathcal{L}(\mathcal{H}_{B'})$ be two quantum channels and let $\rho_{AB}\in\mathcal{D}(\mathcal{H}_A\otimes\mathcal{H}_B)$ be a bipartite state. Then,
\bb
    U_{BS}(A';B')_{(\Lambda\otimes\Lambda')(\rho)}\leq U_{BS}(A;B)_{\rho}.
\ee
That is, the geometric umlaut information satisfies data processing under local operations.
\end{lemma}

\begin{proof}
The proof is identical to the one of Lemma~\ref{lem:quantum_lautum_prop}(3), with the only difference that now we need to use the data processing for the Belavkin--Staszewski quantum relative entropy~\cite{Fujii1989} (see also~\cite[Proposition~2.5(2)]{Hiai1991}).
\end{proof}

\begin{prop}[(Subadditivity of $U_{BS}$)] Given two bipartite states $\rho_{AB}\in\mathcal{D}(\mathcal{H}_A\otimes\mathcal{H}_B)$ and $\sigma_{A'B'}\in\mathcal{D}(\mathcal{H}_{A'}\otimes\mathcal{H}_{B'})$, it holds that
    \bb
        U_{BS}(AA';BB')_{\rho\otimes\sigma}\leq U_{BS}(A;B)_{\rho}+U_{BS}(A';B')_{\sigma}.
    \ee
\end{prop}

\begin{proof}
    It is easy to compute
    \bb
        U_{BS}(AA';BB')_{\rho\otimes\sigma}&=\min_{\xi_{BB'}}D_{BS}(\rho_A\otimes\sigma_{A'}\otimes\xi_{BB'}\|\rho_{AB}\otimes\sigma_{A'B'})\\
        &\leqt{(i)} \min_{\xi_{B}\otimes\chi_{B'}}D_{BS}\big((\rho_A\otimes\xi_{B})\otimes(\sigma_{A'}\otimes\chi_{B'})\|\rho_{AB}\otimes\sigma_{A'B'}\big)\\
        &\eqt{(ii)} \min_{\xi_{B}}D_{BS}(\rho_A\otimes\xi_{B}\|\rho_{AB})+\min_{\chi_{B'}}D_{BS}(\sigma_{A'}\otimes\chi_{B'}\|\sigma_{A'B'})\\
        &=U_{BS}(A;B)_{\rho}+U_{BS}(A';B')_{\sigma},
    \ee
    where (i) follows from the choice of the ansatz $\xi_{BB'}=\xi_{B}\otimes\chi_{B'}$ and (ii) is due to the additivity of the BS relative entropy.
\end{proof}

In the setting of quantum comm\"unication theory, we can consequently define the $BS$ version of the umlaut information of a channel as follows. 

\begin{Def}[(Geometric umlaut information of a channel)]\label{def:geom_channel_lautum}
Let $\pazocal{N}:\mathcal{L}(\mathcal{H}_A)\to \mathcal{L}(\mathcal{H}_B)$ be a quantum channel and let $\mathcal{H}_{A'}$ be an auxiliary Hilbert space isomorphic to $\mathcal{H}_{A}$. Then the quantum umlaut information of $\pazocal{N}$ is defined as
\bb
U(\pazocal{N}) \coloneqq \sup_{\Psi_{A'A}} U(A';B)_{(\mathrm{Id}\otimes\NN)(\Psi)} = \sup_{\rho_{A'}} \min_{\sigma_B} D\left(\rho_{A'}\otimes \sigma_B \,\middle\|\, \rho_{A'}^{1/2} J_{A'B}^{(\pazocal{N})} \rho_{A'}^{1/2}\right) ,
\ee
where $\sigma_B$ is an arbitrary mixed state of $\mathcal{H}_B$ and $\Psi_{A'A}$ is an arbitrary input state of $\mathcal{H}_{A'}\otimes\mathcal{H}_A$, which can be assumed to be pure, without loss of generality, due to Lemma~\ref{lem:data_processing_bs}; the state $\rho_{A'}$ in the rightmost expression is the $A'$ marginal of $\Psi_{A'A}$, and $J_{A'B}^{(\pazocal{N})}$ is the (un-normalised) Choi--Jamio\l kowski matrix of $\pazocal{N}$.
\end{Def}

\begin{prop}[(Subadditivity of $U_{BS}(\pazocal{N})$)]\label{lem:subaddBS}
Let $\pazocal{N}_{A_1\to B_1}$ and $\pazocal{M}_{A_2\to B_2}$ be two quantum channels. Then
    \bb
        U_{BS}(\pazocal{N}\otimes\pazocal{M})\leq U_{BS}(\pazocal{N})+U_{BS}(\pazocal{M})
    \ee
    i.e.\ the channel geometric  umlaut information is sub-additive under tensor products.
\end{prop}

\begin{proof} We compute
\begin{align}
U_{BS}(\pazocal{N}\otimes\pazocal{M})
&=\sup_{\rho_{A_1'A_2'}}U_{BS}(A_1'A_2';B_1B_2)_{\rho^{(\pazocal{N}\otimes\pazocal{M})}} \nonumber \\
&=\sup_{\rho_{A_1'A_2'}}\min_{\sigma_{B_1B_2}} D_{BS}\left(\rho_{A_1'A_2'}\otimes\sigma_{B_1B_2}\,\middle\|\,\rho_{A_1'A_2'}^{1/2}\left(J^{(\pazocal{N})}_{A_1'B_1}\otimes J^{(\pazocal{M})}_{A_2'B_2}  \right)\rho_{A_1'A_2'}^{1/2}\right) \nonumber \\
&=\sup_{\rho_{A_1'A_2'}}\min_{\sigma_{B_1B_2}} \Tr\left[(\rho_{A_1'A_2'}\otimes\sigma_{B_1B_2})\log\left(\sigma_{B_1B_2}^{1/2}\left(J^{(\pazocal{N})}_{A_1'B_1}\otimes J^{(\pazocal{M})}_{A_2'B_2}  \right)^{-1}\sigma_{B_1B_2}^{1/2}\right)\right] \nonumber \\
&\leqt{(i)}
\sup_{\rho_{A_1'A_2'}}\min_{\sigma_{B_1}\otimes\sigma_{B_2}} \Tr\left[(\rho_{A_1'A_2'}\otimes\sigma_{B_1}\otimes\sigma_{B_2})\log\left(\sigma_{B_1}^{1/2}\otimes\sigma_{B_2}^{1/2}\left(J^{(\pazocal{N})}_{A_1'B_1}\otimes J^{(\pazocal{M})}_{A_2'B_2}  \right)^{-1}\sigma_{B_1}^{1/2}\otimes\sigma_{B_2}^{1/2}\right)\right] \nonumber \\
&\eqt{(ii)}\sup_{\rho_{A'_1}}\min_{\sigma_{B_1}}\Tr\left[\left(\rho_{A_1'}\otimes\sigma_{B_1}\right)\log\left(\sigma_{B_1}^{1/2}\left(J_\pazocal{N}^{-1}\right)_{A_1'B_1}\sigma_{B_1}^{1/2}\right)\right] \nonumber \\
&\qquad + \sup_{\rho_{A'_2}} \min_{\sigma_{B_2}} \Tr\left[\left(\rho_{A_2'}\otimes\sigma_{B_2}\right)\log\left(\sigma_{B_2}^{1/2}\left(J_\pazocal{M}^{-1}\right)_{A_2'B_2}\sigma_{B_2}^{1/2}\right)\right] \nonumber \\
&=U_{BS}(\pazocal{N})+U_{BS}(\pazocal{M}),
\end{align}
where~(i) follows from the ansatz $\sigma_{B_1B_2}=\sigma_{B_1}\otimes\sigma_{B_2}$ and~(ii) is due to the additivity of the logarithm under tensor products combined with the fact that
\bb
\Tr\left[\left(\rho_{A_1'A_2'}\otimes\sigma_{B_1}\right)\log\left(\sigma_{B_1}^{1/2}\left(J_\pazocal{N}^{-1}\right)_{A_1'B_1}\sigma_{B_1}^{1/2}\right)\right]=\Tr\left[\left(\rho_{A_1'}\otimes\sigma_{B_1}\right)\log\left(\sigma_{B_1}^{1/2}\left(J_\pazocal{N}^{-1}\right)_{A_1'B_1}\sigma_{B_1}^{1/2}\right)\right]
\ee
and similarly for the other term, so that the maximisation over $\rho_{A_1'A_2'}$ reduces to two independent maximisations over the reduced states $\rho_{A_1'}$ and $\rho_{A_2'}$.
\end{proof}

\begin{Def}[(Regularised geometric umlaut information of a channel)] Let $\pazocal{N}:\mathcal{L}(\mathcal{H}_A)\to \mathcal{L}(\mathcal{H}_B)$ be a quantum channel. Its regularised geometric umlaut information is defined as
\bb\label{eq:reg_lautum_ch_BS}
    U^\infty_{BS}(\pazocal{N})\coloneqq\lim_{n\to\infty}\frac{1}{n}U_{BS}(\pazocal{N}^{\otimes n})\, .
\ee
\end{Def}

The existence of the limit~\eqref{eq:reg_lautum_ch_BS} is ensured by Fekete's lemma, as the sequence $U_{BS}(\pazocal{N}^{\otimes n})$ is subadditive according to Proposition~\ref{lem:subaddBS}.

\begin{lemma}\label{lem:infty_norm}
 Let $\pazocal{N}:\mathcal{L}(\mathcal{H}_A)\to \mathcal{L}(\mathcal{H}_B)$ be a quantum channel. Its geometric umlaut information can be written as
    \bb
        U_{BS}(\pazocal{N})=\min_{\sigma_B}\left\|\,\Tr_B\left[\sigma_B\log\left(\sigma_B^{1/2}\left(J^{(\pazocal{N})}_{A'B}\right)^{-1}\sigma_B^{1/2}\right)\right]\,\right\|_\infty,
    \ee
    where $J^{(\pazocal{N})}_{A'B}$ is the Choi--Jamio\l kowski state of $\pazocal{N}$.
\end{lemma}

\begin{proof}
    As in the proof of Lemma~\ref{lem:replacer} we apply the minimax principle to the definition of $U_{BS}(\pazocal{N})$ and we use~\cite[Theorem 3]{Fang2019} to conclude, noticing that the Choi--Jamio\l kowski matrix of the replacement  channel $\pazocal{R}^{\sigma_B}$ is $J_{A'B}^{(\pazocal{R}^{\sigma_B})}=\id_{A'}\otimes\sigma_B$.
\end{proof}

\begin{prop}[(An alternative formula for CQ-channels)]\label{prop:alternativa_geometric_cq}
Let $\pazocal{N}_{X\to B}(\,\cdot\,) = \sum_{x\in \mathcal{X}} \bra{x}\cdot \ket{x} \rho_x^B$ be a classical-quantum channel. Then, we have
    \bb
        U(\pazocal{N})=\min_{\sigma_B}\sup_{x\in\mathcal{X}}D\left(\sigma_{B}\middle\|\rho_x^{B}\right),\quad
        U_{BS}(\pazocal{N})=\min_{\sigma_B}\sup_{x\in\mathcal{X}}D_{BS}\left(\sigma_{B}\middle\|\rho_x^{B}\right).
    \ee    
\end{prop}

\begin{proof}
Let us compute the Choi--Jamio\l kowski state of $\pazocal{N}$:
    \bb
        J^{(\pazocal{N})}_{X'B}&=\sum_{x,x'}\ket{x}\bra{x'}_{X'}\otimes\pazocal{N}_{X\to B}\left(\ket{x}\bra{x'}_{X}\right)\\
        &=\sum_{x\in\mathcal{X}}\ketbra{x}_{X'}\otimes\rho_{x}^{B},
    \ee
By Lemma~\ref{lem:replacer},
\bb
    U(\pazocal{N})&= \min_{\sigma_B}D(\pazocal{R}^{\sigma_B}_{X\rightarrow B}\|\pazocal{N})\\
    &=\min_{\sigma_B}\sup_{P_{X'}}D\left(P_{X'}^{1/2}\,J_{X'B}^{(\pazocal{R}^{\sigma_B})}\,P_{X'}^{1/2}\,\middle\|\,P_{X'}^{1/2}\,J_{X'B}^{(\pazocal{N})}\,P_{X'}^{1/2})\right)\\
    &\eqt{(i)}\min_{\sigma_B}\sup_{P_{X'}} D\left(P_{X'}\otimes\sigma_B \middle \|\sum_{x\in\mathcal{X}}P_{X'}(x)\ketbra{x}_{X'}\otimes\rho_x^B\right)\\
    &\eqt{(ii)}\min_{\sigma_B}\sup_{P_{X'}}\left(-S(\sigma_B)-\sum_{x\in\mathcal{X}}P_X(x)\Tr\left[\sigma_B\log\rho_x^{B}\right]\right)\\
    &=\min_{\sigma_B}\sup_{P_{X'}}\sum_{x\in\mathcal{X}}P_{X'}(x)D(\sigma_B\|\rho_x^B)\\
    &=\min_{\sigma_B}\sup_{x\in\mathcal{X}}D(\sigma_B\|\rho_x^B),
\ee
where in (i) we have used that $J_{X'B}^{(\pazocal{R}^{\sigma_B})}=\id_{X'}\otimes\sigma_B$ and in (ii) we have expanded the expression of the relative entropy, getting a cancellation of the entropy of $P_X$. 
Now, we want to get a similar identity for the geometric umlaut information of $\pazocal{N}$ using Lemma~\ref{lem:infty_norm}. Let us compute
    \bb
        \log\left(\sigma_{B}^{1/2}\left(J^{(\pazocal{N})}_{X'B}\right)^{-1}\sigma_{B}^{1/2}\right)=-\log\left(\sigma_{B}^{-1/2}J^{(\pazocal{N})}_{X'B}\sigma_{B}^{-1/2}\right)&=
        \sum_{x\in\mathcal{X}}\ketbra{x}_{X'}\otimes\log\left(\sigma_{B}^{1/2}\left(\rho_{x}^{B}\right)^{-1}\sigma_{B}^{1/2}\right).
    \ee
    Therefore,
    \bb
    \Tr_{B}\left[\sigma_{B}\log\left(\sigma_{B}^{1/2}\left(J^{(\pazocal{N})}_{X'B}\right)^{-1}\sigma_{B}^{1/2}\right)\right]&=
        \sum_{x\in\mathcal{X}}\ketbra{x}_{X'} \otimes \Tr_{B}\left[\sigma_{B}\log\left(\sigma_{B}^{1/2}\left(\rho_{x}^{B}\right)^{-1}\sigma_{B}^{1/2}\right)\right],
    \ee
    so, by Lemma~\ref{lem:infty_norm}, we conclude that
    \bb
        U_{BS}(\pazocal{N})&=\min_{\sigma_B}\sup_{x\in\mathcal{X}}\left\|\Tr_{B}\left[\sigma_{B}\log\left(\sigma_{B}^{1/2}\left(\rho_{x}^{B}\right)^{-1}\sigma_{B}^{1/2}\right)\right]\right\|_{\infty}\\
        &=\min_{\sigma_B}\sup_{x\in\mathcal{X}}D_{BS}(\sigma_{B}\|\rho_x^{B}).
    \ee
    This concludes the proof.
\end{proof}

\begin{prop} \label{prop:hierarchy_umlauts}
Let $\pazocal{N}:\mathcal{L}(\mathcal{H}_A)\to \mathcal{L}(\mathcal{H}_B)$ be a quantum channel. Then, we have
    \bb\label{eq:U_bs1}
        U(\pazocal{N})\leq U^\infty(\pazocal{N})\leq U_{BS}^\infty(\pazocal{N}) \leq  U_{BS}(\pazocal{N}).
    \ee
In particular, $U(\pazocal{N})$ and $U_{BS}(\pazocal{N})$ provide single-letter lower and upper bounds, respectively, to the activated, non-signalling error exponent of $\pazocal{N}$ at zero rate:
\bb\label{eq:U_bs2}
    U(\pazocal{N})\leq \errnsa{\pazocal{N}} \leq  U_{BS}(\pazocal{N}).
\ee
\end{prop}

\begin{proof} By Lemma~\ref{lem:BS_ineq}, for any choice of $\rho_{A'}$,
\bb
        &U(A';B)_{\rho_{A'B}^{(\pazocal{N})}}\leq U_{BS}(A';B)_{\rho_{A'B}^{(\pazocal{N})}}\\
        &\qquad\Longrightarrow\quad \sup_{\rho_{A'}}U(A';B)_{\rho_{A'B}^{(\pazocal{N})}}\leq \sup_{\rho_{A'}} U_{BS}(A';B)_{\rho_{A'B}^{(\pazocal{N})}}\\
        &\quad\;\iff \quad U(\pazocal{N})\leq U_{BS}(\pazocal{N}).
    \ee
    Regularising the previous inequality we get
    \bb
         U^\infty(\pazocal{N})\leq U_{BS}^\infty(\pazocal{N}).
    \ee
    Since $U_{BS}^\infty(\pazocal{N})\leq U_{BS}(\pazocal{N})$ by Proposition~\ref{lem:subaddBS} and $U(\pazocal{N})\leq U^\infty(\pazocal{N})$, we have completed the proof of~\eqref{eq:U_bs1}. Then~\eqref{eq:U_bs2} follows from Theorem~\ref{thm:exact_tho}.
\end{proof} 

\begin{figure}
    \centering
\begin{tikzpicture}
    \draw[thick, ->] (2,0) -- (15,0);
    \draw[thick, gray, ->] (8.5,1.6) -- (15,1.6);
    \draw[thick,gray] (7.5,0) -- (8.5,1.6);

    \draw[thick] (2.5,0.2) -- (2.5,-0.2);
    \draw[thick, blue] (4.5,0.2) -- (4.5,-0.2); 
    \draw[thick] (6.5,0.2) -- (6.5,-0.2); 
    \draw[thick, gray] (10.5,1.4) -- (10.5,1.8); 
    \draw[thick, gray] (12.5,1.4) -- (12.5,1.8);
    \draw[thick, gray] (14.5,1.4) -- (14.5,1.8);
    \draw[thick] (9.5,0.2) -- (9.5,-0.2); 
    \draw[thick] (11.5,0.2) -- (11.5,-0.2);
    \draw[thick] (13.5,0.2) -- (13.5,-0.2);

    \node[blue, above] at (4.5,0.2) {$\ell_{2,u_{2}}(\pazocal{N})$};
    \node[above] at (6.5,0.2) {$U(\pazocal{N})$};
    \node[above] at (9.5,0.2) {$U^\infty(\pazocal{N})$};
    \node[above] at (11.5,0.2) {$U^\infty_{BS}(\pazocal{N})$};
    \node[above] at (13.5,0.2) {$U_{BS}(\pazocal{N})$};
    \node[above, gray] at (10.5,1.8) {$L(\pazocal{N})$};
    \node[above, gray] at (12.5,1.8) {$L_{BS}(\pazocal{N})$};
    \node[above, gray] at (14.5,1.8) {$L^{\infty}_{BS}(\pazocal{N})$};
    \node[blue, below] at (6.5,-0.2) {\errnsa{\pazocal{N}}};
    \node[below] at (9.5,-0.2) {\errnsa{\pazocal{N}}};
    \node[below] at (2.5,-0.2) {\errpl{\pazocal{N}}};

    \draw[thick, blue] (6.5,1.3) -- (6.5,0.9);
    \draw[thick, blue] (2,1.1) -- (6.5,1.1);
    \node[above, blue] at (4.25,1.1) {$\ell_{k,q}(\pazocal{N})$};


    \node[above] at (16,0.8) {information};
    \node[above] at (16,0.4) {measures};
    \node[above] at (16,-0.6) {error};
    \node[above] at (16,-1.1) {exponents};

\end{tikzpicture}
 
    \caption{A pictorial comparison of the information measures and of the error exponents that were discussed in this paper. $\pazocal{N}_{A\to B}$ is a quantum channel; for any $k\geq 1$, $q\in\mathbb{R}^k$ is a vector such that $\sum_{i=1}^kq_i=1$ and $u_2 =(1/2, 1/2)$. The inequalities and identities involving objects represented in blue have to be considered only if $\pazocal{N}$ is a classical-quantum channel.}
    \label{fig:recap}
\end{figure}

We did not explore in detail the quantum generalisation of Palomar and Verd\'u's lautum information~\cite{Lautum_08} to the quantum case, but a few elementary upper bounds can be immediately established. 
Given a bipartite state $\rho_{AB}\in\mathcal{D}(\mathcal{H}_A\otimes\mathcal{H}_B)$, we call
    \bb
        L(A\!:\!B)_{\rho}&\coloneqq D(\rho_A\otimes\rho_B\|\rho_{AB})\\
        L_{BS}(A\!:\!B)_{\rho}&\coloneqq D_{BS}(\rho_A\otimes\rho_B\|\rho_{AB})
    \ee
where $\rho_A=\Tr_B[\rho_{AB}]$ and $\rho_B=\Tr_A[\rho_{AB}]$. Due to the additivity of the Umegaki and Belavkin--Staszewski relative entropies under tensor products, it is clear that each corresponding lautum information is additive under tensor product, namely, given two bipartite states $\rho_{AB}\in\mathcal{D}(\mathcal{H}_A\otimes\mathcal{H}_B)$ and $\sigma_{A'B'}\in\mathcal{D}(\mathcal{H}_{A'}\otimes\mathcal{H}_{B'})$, it holds that
\bb
    L(AA'\!:\!BB')_{\rho\otimes\sigma}=L(A\!:\!B)_{\rho}+L(A'\!:\!B')_{\sigma},\qquad  L_{BS}(AA'\!:\!BB')_{\rho\otimes\sigma}=L_{BS}(A\!:\!B)_{\rho}+L_{BS}(A'\!:\!B')_{\sigma}.
\ee
The corresponding formulation for a quantum channel $\mathcal{N}_{A\to B}$ would be
\bb
        L(\pazocal{N})&\coloneqq\sup_{\rho_{A'}}L(A'\!:\!B)_{\rho_{A'B}^{(\mathcal{N})}},\\
        L_{BS}(\pazocal{N})&\coloneqq\sup_{\rho_{A'}}L_{BS}(A'\!:\!B)_{\rho_{A'B}^{(\mathcal{N})}},
    \ee
where, as usual, $\rho_{A'B}^{(\mathcal{N})}\coloneqq\rho_{A'}^{1/2}\,J_{A'B}^{(\pazocal{N})}\,\rho_{A'}^{1/2}$. In particular, if we consider $n$ uses of the channel and we choose $\rho_{A'\,^n}=\rho_A^{\otimes n}$ for any arbitrary $\rho_A\in\mathcal{D}(\mathcal{H}_A)$, then we have $\rho_{A'\,^nB^n}^{(\mathcal{N}^{\otimes n})}\coloneqq\left(\rho_{A'}^{1/2}\,J_{A'B}^{(\pazocal{N})}\,\rho_{A'}^{1/2}\right)^{\otimes n}$, yielding the lower bound
\bb
    L(\pazocal{N}^{\otimes n})=\sup_{\rho_{A'\,^n}}L(A'\,^n\!:\!B^n)_{\rho_{A'\,^nB^n}^{(\mathcal{N}^{\otimes n})}}\geq \sup_{\rho_{A'}}L(A'\,^n\!:\!B^n)_{\left(\rho_{A'B}^{(\mathcal{N})}\right)^{\otimes n}}\eqt{(a)}nL(\pazocal{N}),
\ee
where in (a) we have leveraged the additivity of $L$. The same holds for the BS formulation of $L$. In particular, calling $L^\infty(\mathcal{N})$ and $L^\infty_{BS}(\mathcal{N})$ the regularisations of $L(\mathcal{N})$ and $L_{BS}(\mathcal{N})$, respectively, we have
\bb
    L(\mathcal{N}) \leq L^\infty(\mathcal{N}) \coloneqq \lim_{n\to\infty}\frac{1}{n}L(\mathcal{N}^{\otimes n}),\qquad  L_{BS}(\mathcal{N}) \leq L^\infty_{BS}(\mathcal{N}) \coloneqq \lim_{n\to\infty}\frac{1}{n}L_{BS}(\mathcal{N}^{\otimes n}).
\ee
Since $D(\rho\|\sigma)\leq D_{BS}(\rho\|\sigma)$, we immediately have that every (Umegaki) lautum information measure is smaller than the corresponding Belavkin--Staszewski information measure. 
Furthermore, by their very definition, every umlautum information measure is smaller than the corresponding lautum information measure. Therefore, we conclude with the following inequalities
\bb
    U(\pazocal{N})\leq L(\pazocal{N})\leq 
    \begin{cases}
        L_{BS}(\pazocal{N})\\
        L^\infty(\pazocal{N})
    \end{cases}
    \hspace{-0.8em}\leq L_{BS}^\infty(\pazocal{N}),\qquad U^\infty(\pazocal{N})\leq L^\infty(\pazocal{N})\leq L^\infty_{BS}(\pazocal{N}).
\ee
Some of them are schematically represented in Figure~\ref{fig:recap}, together with the other information measures and error exponents discussed in this paper.


\section{Outlook}\label{sec:outlook} 

In this work, we generalised the umlaut information~\cite{Filippo25} to the quantum setting. The quantum umlaut information bears many similarities to its classical counterpart, admitting a closed-form expression and being additive under tensor product. We showed that this resemblance extends to the case of classical-quantum channels, where the channel umlaut information is additive for parallel channel uses, and gains an operational interpretation as the zero-rate error exponent in activated non-signalling--assisted communication.
The similarities with the additive and well-behaved classical umlaut information break down for general quantum channels, however. We showed in particular that the channel umlaut information is not additive in general, necessitating the use of its regularised form to recover an operational interpretation as a zero-rate error exponent in activated non-signalling--assisted coding. This latter scenario is known to correspond to the well-known sphere packing bound in the classical setting~\cite{Polyanskiy13} and in the classical-quantum setting~\cite{Oufkir2024Oct}. The findings of this work indicate that regularisation is needed in order to prove a general quantum sphere packing bound, which is an interesting problem to be explored in future work.

An intriguing aspect of our findings is that we do \emph{not} obtain a single-letter solution for the zero-rate exponent for general quantum channels, even though the corresponding channel capacities (asymptotic rates with vanishing error) \emph{are} given by the single-letter quantum mutual information~\cite{Bennett2002,leung_2015,wang2019-2} --- assuming entanglement or non-signalling assistance. This goes in a completely different direction than the recent result from~\cite{lami2024asymptotic}: there, a problem that only admits a multi-letter, regularised solution at the level of asymptotic rates became single letter for the zero-rate error exponent, 
which suggested that looking at such exponents could make asymptotic problems simpler to characterise and more easily yield single-letter expressions. This curious discrepancy may mean that such intuition does not carry over to quantum channel manipulation questions.

A property of the umlaut information --- and, more generally, zero-rate exponents --- is that it can be infinite for some channels, meaning that an arbitrarily large number of messages can be sent using activated non-signalling strategies with error vanishing super-exponentially fast. 
As discussed in Section~\ref{sec:operational_int} following the development from \cite{duan_2016,duan_2017}, a diverging (regularised) umlaut information in fact even implies that there is non-vanishing zero-error non-signalling capacity. This means that there is a strong and precise dichotomy between the two settings, and it would be interesting to further explore these connections to zero-error information theory. 

Yet another open problem is whether umlaut information also governs the zero-rate error exponents in the more practically meaningful setting of entanglement assistance~\cite{Bennett2002}, and whether the need for activation in the non-signalling--assisted strategies used in our work can be removed, following e.g.\ recent ideas in~\cite{Oufkir2024Oct}.
 
Finally, classically~\cite{Filippo25}, the task of list decoding provides an operational interpretation of the umlaut information in the \emph{unassisted} zero-rate large-list setting. A quantum generalisation of list decoding is still an open problem, even in the classical-quantum case. We conjecture that the family of lower bounds we provide in Proposition~\ref{prop:inequality} constitute a family of multivariate fidelities that can be used in order to prove list decoding error exponents.

\acknowledgements
We thank Mark M.\ Wilde for correspondence on the topics of this paper. MB acknowledges funding from the European Research Council (ERC Grant Agreement No.~948139) and the Excellence Cluster Matter and Light for Quantum Computing (ML4Q). LL acknowledges financial support from the European Union under the European Research Council (ERC Grant Agreement No.~101165230) and from MIUR (Ministero dell'Istruzione, dell'Universit\`a e della Ricerca) through the project `Dipartimenti di Eccellenza 2023--2027' of the `Classe di Scienze' department at the Scuola Normale Superiore. MT is supported by the Ministry of Education through grant T2EP20124-0005 and by the National Research Foundation, Singapore through the National Quantum Office, hosted in A*STAR, under its Centre for Quantum Technologies Funding Initiative (S24Q2d0009). AO  acknowledges funding by the European Research Council (ERC Grant Agreement No. 948139).

\bibliography{biblio}


\appendix

\section{Proof of Proposition~\ref{prop:closed_alpha}}\label{app:proof_closed_alpha}
    Let us rewrite
\bb
    U_\alpha(A;B)_{\rho}&=\min_{\sigma_B}\frac{1}{\alpha-1}\log\Tr\left[(\rho_A\otimes\sigma_B)^\alpha\rho_{AB}^{1-\alpha}\right]\\
    &=\frac{1}{\alpha-1}\log \max_{\sigma_B}\Tr\left[\sigma_B^\alpha\Tr_A\left[\rho_A^\alpha\rho_{AB}^{1-\alpha}\right]\right].
\ee
Let $X\coloneqq \sigma_B^\alpha$ and $Y\coloneqq \Tr_A\left[\rho_A^\alpha\rho_{AB}^{1-\alpha}\right]$. Then
\bb
    X,Y\geq 0, \qquad \|X\|_{1/\alpha}=\left(\Tr[X^{1/\alpha}]\right)^\alpha=1.
\ee 
Then, by Hölder's inequality
\bb
    U_\alpha(A;B)_{\rho}
    &=\frac{1}{\alpha-1}\log \max_{\substack{X\geq 0\\ \|X\|_{1/\alpha}}=1}\Tr\left[XY\right]\eqt{(i)}\frac{1}{\alpha-1}\log \|Y\|_{\frac{1}{1-\alpha}}=-\log \big\|\Tr_A\left[\rho_A^\alpha\rho_{AB}^{1-\alpha}\right]\big\|_{\frac{1}{1-\alpha}}^{\frac{1}{1-\alpha}},
\ee
where the equality in (i) is achieved by
\bb
X = \frac{Y^{\frac{\alpha}{1-\alpha}}}{\|Y\|_{\frac{1}{1-\alpha}}^{\frac{\alpha}{1-\alpha}}} \geq 0\, ,
\ee
i.e.
\bb
\sigma_B=\frac{\left(\Tr_A\left[\rho_A^\alpha\rho_{AB}^{1-\alpha}\right]\right)^{\frac{1}{1-\alpha}}}{\Tr\left[\left(\Tr_A\left[\rho_A^\alpha\rho_{AB}^{1-\alpha}\right]\right)^{\frac{1}{1-\alpha}}\right]}\eqcolon \uml{\rho}_{B}^{(\alpha)}.
\ee
Let us call $\uml{\rho}_{B}^{(\alpha)}$ $\alpha$-umlaut-marginal of $\rho_{AB}$.
If we consider any product state $\rho_{AB}\otimes\sigma_{A'B'}$, it is easy to verify that its $\alpha$-umlaut reduced state is the tensor product of $\uml{\rho}_{B}^{(\alpha)}$ and $\uml{\sigma}_{B'}^{(\alpha)}$. Since the Petz--R\'{e}nyi quantum relative entropies are additive, this implies the additivity of $U_\alpha$.


\section{Covariant channels} \label{app:covariant_channels}

This appendix is devoted to the proof of Proposition~\ref{prop:group}, which plays an important role in simplifying the numerical calculation of the single-copy umlaut information of group covariant channels, in turn a key step in establishing additivity violations.

\begin{proof}[Proof of Proposition~\ref{prop:group}] Let $\{U_A(g)\}_{g\in G}$ and $\{V_B(g)\}_{g\in G}$ denote two unitary representations of a finite group $G$ in $\mathcal{H}_A$ and $\mathcal{H}_B$. Let $\pazocal{U}_{g}:\mathcal{L}(\mathcal{H}_A)\to\mathcal{L}(\mathcal{H}_A)$ be the pure unitary channel
    \bb
        \pazocal{U}_{g}(\,\cdot\,)\coloneqq U_A (g) \,\cdot \, U_A^\dagger(g),
    \ee
and, similarly, let $\pazocal{V}_{g}:\mathcal{L}(\mathcal{H}_B)\to\mathcal{L}(\mathcal{H}_B)$ be the pure unitary channel
    \bb
        \pazocal{V}_{g}(\,\cdot\,)\coloneqq V_B(g) \,\cdot \, V_B^\dagger(g).
    \ee
Now, whenever we consider a channel $\pazocal{N}:\mathcal{L}(\mathcal{H}_A)\to\mathcal{L}(\mathcal{H}_B)$ which is $G$-covariant according to the above representations, i.e.\ 
$\pazocal{N}\pazocal{U}_g=\pazocal{V}_g\pazocal{N}$, we have
    \bb
        D(\pazocal{R}^{\sigma_B}\|\pazocal{N})\eqt{(i)}D(\pazocal{V}_{g}\pazocal{R}^{\sigma_B}\pazocal{U}_{-g}\|\pazocal{V}_{g}\pazocal{N}\pazocal{U}_{-g})\eqt{(ii)}D(\pazocal{V}_{g}\pazocal{R}^{\sigma_B}\pazocal{U}_{-g}\|\pazocal{V}_{g}\pazocal{V}_{-g}\pazocal{N})\eqt{(iii)}D(\pazocal{V}_{g}\pazocal{R}^{\sigma_B}\|\pazocal{N}),
    \ee
    for any $g\in G$ and $\sigma_B\in D(\mathcal{H}_B)$, where
    \begin{itemize}
        \item (i) holds since $D(\pazocal{U}_2\pazocal{M}\pazocal{U}_1\|\pazocal{U}_2\pazocal{N}\pazocal{U}_1)=D(\pazocal{M}\|\pazocal{N})$ is true for arbitrary channels $\pazocal{M}$, $\pazocal{N}$ and pure unitary channels $\pazocal{U}_1$, $\pazocal{U}_2$;
        \item (ii) holds since $\pazocal{N}$ is assumed to be covariant with respect to $G$;
        \item (iii) holds since $\pazocal{R}^{\sigma_B}\pazocal{M}=\pazocal{R}^{\sigma_B}$ for any channel $\pazocal{M}$.
    \end{itemize}
    As a consequence, by the convexity of the channel relative entropy,
    \bb\label{eq:convexity}
        D\left(\,\frac{1}{|G|}\sumno_{g\in G}\pazocal{V}_g\pazocal{R}^{\sigma_B}\,\middle\|\,\pazocal{N}\,\right)\leq \frac{1}{|G|}\sum_{g\in G}D(\pazocal{V}_g\pazocal{R}^{\sigma_B}\|\pazocal{N})=D(\pazocal{R}^{\sigma_B}\|\pazocal{N}).
    \ee 
    In particular, defining $\sigma^{(G)}_B$ as the symmetrised version of any state $\sigma_B$, i.e.
    \bb
        \sigma_B^{(G)}\coloneqq\frac{1}{|G|}\sum_{g\in G}\pazocal{V}_g\sigma_B,
    \ee
    which is $G$-invariant according to the representation $\{U_B(g)\}_{g\in G}$
    \bb
        \pazocal{V}_{g}\sigma_B^{(G)}=\sigma_B^{(G)}\qquad \forall g\in G,
    \ee
    we immediately see that
    \bb
        \frac{1}{|G|}\sum_{g\in G}\pazocal{V}_g\pazocal{R}^{\sigma_B}=\pazocal{R}^{\frac{1}{|G|}\sum_{g\in G}\pazocal{V}_g\sigma_B}=\pazocal{R}^{\sigma_B^{(G)}},
    \ee
    Eq.~\eqref{eq:convexity} ensures that the minimisation over a generic $\sigma_B$ can be restricted to a minimisation over the set of $G$-invariant states $\sigma_B^{(G)}$:
    \bb\label{eq:minG}
        U(\pazocal{N})=\min_{\sigma_B^{(G)}}D(\pazocal{R}^{\sigma_B^{(G)}}\|\pazocal{N}).
    \ee
    Such family of $G$-invariant states engenders replacer channels which are $G$-covariant. Indeed, since $\pazocal{R}^{\sigma_B^{(G)}}\pazocal{U}_g=\pazocal{R}^{\sigma_B^{(G)}}$ and $\pazocal{V}_g\pazocal{R}^{\sigma_B^{(G)}}=\pazocal{R}^{\pazocal{V}_g\sigma_B^{(G)}}=\pazocal{R}^{\sigma_B^{(G)}}$,
    it clearly holds that
    \bb
        \pazocal{R}^{\sigma_B^{(G)}}\pazocal{U}_g=\pazocal{V}_g\pazocal{R}^{\sigma_B^{(G)}}
    \ee
    for any $g\in G$.
    This means that, in Eq.~\eqref{eq:minG}, both the channel $\pazocal{R}^{\sigma_B^{(G)}}$ and $\pazocal{N}$ are $G$-covariant, so we can apply~\cite[Proposition II.4]{Leditzky2018} to rewrite the definition of channel relative entropy restricting the maximisation to the set of $G$-invariant states $\rho_{A'}^{(G)}$:
    \bb\label{eq:maxG}
        D(\pazocal{R}^{\sigma_B^{(G)}}\|\pazocal{N})=\sup_{\rho_{A'}^{(G)}}D\big((\mathrm{Id}_{A'}\otimes\pazocal{R}^{\sigma_B^{(G)}})(\phi_{A'A}^{(G)})\|(\mathrm{Id}_{A'}\otimes\pazocal{N}_{A\to B})(\phi_{A'A}^{(G)})\big),
    \ee
    where 
    \bb
        \phi_{A'A}^{(G)}\coloneqq \left(\rho_{A'}^{(G)}\right)^{1/2}\Phi_{A'A}\left(\rho_{A'}^{(G)}\right)^{1/2},\qquad \Phi_{A'A}\coloneqq\sum_{ij}\ket{ii}_{A'A}\bra{jj}_{A'A}.
    \ee
    Therefore,
    \bb
        U(\pazocal{N})&\eqt{(iv)}\min_{\sigma_B^{(G)}}\sup_{\rho_{A'}^{(G)}}D\big((\mathrm{Id}_{A'}\otimes\pazocal{R}^{\sigma_B^{(G)}})(\phi_{A'A}^{(G)})\|(\mathrm{Id}_{A'}\otimes\pazocal{N}_{A\to B})(\phi_{A'A}^{(G)})\big)\\
        &\eqt{(v)}\sup_{\rho_{A'}^{(G)}}\min_{\sigma_B^{(G)}}D\big((\mathrm{Id}_{A'}\otimes\pazocal{R}^{\sigma_B^{(G)}})(\phi_{A'A}^{(G)})\|(\mathrm{Id}_{A'}\otimes\pazocal{N}_{A\to B})(\phi_{A'A}^{(G)})\big)\\
        &\eqt{(vi)}\sup_{\rho_{A'}^{(G)}}\min_{\sigma_B}D\big(\rho_{A'}^{(G)}\otimes\sigma_B\|(\rho_{A'}^{(G)})^{1/2}\, J^{(\pazocal{N})}_{A'A}\,(\rho_{A'}^{(G)})^{1/2} \big)\\
        &\eqt{(vii)}\sup_{\rho_{A'}^{(G)}}\left(-S(\rho_{A'}^{(G)})-\log\Tr\left[\exp\Tr_{A'}\left[\rho_{A'}^{(G)}\log\left((\rho_{A'}^{(G)})^{1/2}\, J^{(\pazocal{N})}_{A'A}\,(\rho_{A'}^{(G)})^{1/2}\right)\right]\right]\right),\\
    \ee
    where
    \begin{itemize}
        \item in (iv) we have combined~\eqref{eq:minG} with~\eqref{eq:maxG};
        \item in (v) we have used the minimax theorem as in the proof of Lemma~\ref{lem:replacer};
        \item be enlarging the set of the minimisers to any generic $\sigma_B$, in (vi) we should have an inequality, but we will soon show that the equality holds; in (vi) we have also introduced the Choi--Jamio\l kowski state of $\pazocal{N}$.
        \item in (vii) we have used the closed-form expression provided by Proposition~\ref{prop:closed_lautum} and we have skipped a few $\id_B$ in order to simplify the notation.
    \end{itemize}
    The only remaining step to conclude is to prove that (vi) is indeed an equality. Let us start by noticing that, for any fixed $\rho_{A'}^{(G)}$, $\sigma_B$ and $g\in G$, we have 
    \bb\label{eq:rel_ent_group}
        D\big(\rho_{A'}^{(G)}\otimes \pazocal{V}_g\sigma_B\|(&\mathrm{Id}_{A'}\otimes\pazocal{N}_{A\to B})(\phi_{A'A}^{(G)})\big)\\
        &\eqt{(a)}D\big(\rho_{A'}^{(G)}\otimes \sigma_B\|(\mathrm{Id}_{A'}\otimes \pazocal{V}_{-g}\pazocal{N}_{A\to B})(\phi_{A'A}^{(G)})\big)\\
        &\eqt{(b)}D\big(\rho_{A'}^{(G)}\otimes \sigma_B\|(\mathrm{Id}_{A'}\otimes \pazocal{N}_{A\to B}\pazocal{U}_{-g})(\phi_{A'A}^{(G)})\big)\\
        &\eqt{(c)}D\big(\rho_{A'}^{(G)}\otimes \sigma_B\|(\rho_{A'}^{(G)})^{1/2}(\mathrm{Id}_{A'}\otimes \pazocal{N}_{A\to B}\pazocal{U}_{-g})(\Phi_{A'A})(\rho_{A'}^{(G)})^{1/2}\big),
    \ee
    where
    \begin{itemize}
        \item in (a) we have used the invariance of the quantum relative entropy under joint unitary conjugation, i.e.\ $D(\rho\|\sigma)=D(U\rho U^\dagger\|U\sigma U^\dagger)$;
        \item (b) follows from the $G$-covariance of $\pazocal{N}$;
        \item (c) is simply a restatement of the definition of $\phi_{A'A}^{(G)}$.
    \end{itemize}

Let $\{\ket{\psi_\lambda}_{A'}\}_{\lambda=1,\dots,\dim\mathcal{H}_{A'}}$ be a the basis diagonalising $\rho^{(G)}_{A'}$, namely
\bb\label{eq:diagonal_form}
    \rho^{(G)}_{A'}=\sum_{\lambda=1}^{\dim\mathcal{H}_{A'}}\mu_\lambda \ketbra{\psi_\lambda}_{A'}
\ee
for suitable real numbers $\mu_\lambda\geq 0$.
Let us apply the transpose trick in that basis:
    \bb\label{eq:transpose_trick}
        (\id_{A'}\otimes U_{A}(-g))\sum_\lambda\sum_{i=1}^{\dim \mathcal{H}_\lambda}\ket{\psi_\lambda}_{A'}\ket{\psi_\lambda}_{A}
        &=(U^\intercal_{A'}(-g)\otimes \id_{A} )\sum_\lambda\sum_{i=1}^{\dim \mathcal{H}_\lambda}\ket{\psi_\lambda}_{A'}\ket{\psi_\lambda}_{A}\\
        &= (U^\ast_{A'}(g)\otimes \id_{A} )\sum_\lambda\sum_{i=1}^{\dim \mathcal{H}_\lambda}\ket{\psi_\lambda}_{A'}\ket{\psi_\lambda}_{A}.
    \ee
    Calling $\pazocal{U}^\ast_{g}:\mathcal{L}(\mathcal{H}_{A'})\to\mathcal{L}(\mathcal{H}_{A'})$ be the pure unitary channel
    \bb
        \pazocal{U}^\ast_{g}(\,\cdot\,)\coloneqq U_{A'}^\ast (g) \,\cdot \, U_{A'}^\intercal(g),
    \ee
    we immediately see that
    \bb\label{eq:just_d}
        (\mathrm{Id}_{A'}\otimes \pazocal{N}_{A\to B}\pazocal{U}_{-g})(\Phi_{A'A})=(\pazocal{U}^\ast_{g}\otimes \pazocal{N}_{A\to B})(\Phi_{A'A})
    \ee
    as a consequence of~\eqref{eq:transpose_trick}.
    Taking the complex conjugate and transpose in the eigenbasis of $\rho_{A'}^{(G)}$, it is clear that $\left(\rho_{A'}^{(G)}\right)^\ast = \rho_{A'}^{(G)}$ due to the form~\eqref{eq:diagonal_form}, and therefore
    \bb\label{eq:just_g}
        U_{A'} (g) \rho_{A'}^{(G)} U_{A'}^\dagger(g)=\rho_{A'}^{(G)}\qquad\xrightarrow{\ast}\qquad
        U_{A'}^\ast (g) \rho_{A'}^{(G)} U_{A'}^\intercal(g)=\rho_{A'}^{(G)} \qquad\forall\,g\in G,
    \ee
    i.e.\ $\pazocal{U}_g(\rho_{A'}^{(G)})=\rho_{A'}^{(G)}$. As a consequence, we have
    \bb
        U_{A'}^\ast (g) (\rho_{A'}^{(G)})^{1/2} =(\rho_{A'}^{(G)})^{1/2}U_{A'}^\ast(g),\qquad U_{A'}^\intercal (g) (\rho_{A'}^{(G)})^{1/2} =(\rho_{A'}^{(G)})^{1/2}U_{A'}^\intercal(g)\qquad\forall\,g\in G,
    \ee
    and therefore
    \bb\label{eq:just_e}
        (\rho_{A'}^{(G)})^{1/2}(\pazocal{U}^\ast_{g}\otimes \pazocal{N}_{A\to B})(\Phi_{A'A})(\rho_{A'}^{(G)})^{1/2}&=(\pazocal{U}^\ast_{g}\otimes \mathrm{Id}_{A})\left((\rho_{A'}^{(G)})^{1/2}(\mathrm{Id}_{A'}\otimes \pazocal{N}_{A\to B})(\Phi_{A'A})(\rho_{A'}^{(G)})^{1/2}\right)\\
        &=(\pazocal{U}^\ast_{g}\otimes \mathrm{Id}_{A})\left((\rho_{A'}^{(G)})^{1/2}J^{(\pazocal{N})}_{A'A}(\rho_{A'}^{(G)})^{1/2}\right).
    \ee
    Eq.~\ref{eq:rel_ent_group} becomes
    \bb\label{eq:group_invariance}
    D\big(\rho_{A'}^{(G)}\otimes \pazocal{V}_g\sigma_B\|(&\mathrm{Id}_{A'}\otimes\pazocal{N}_{A\to B})(\phi_{A'A}^{(G)})\big)\\
    &\eqt{(d)}D\big(\rho_{A'}^{(G)}\otimes \sigma_B\|(\rho_{A'}^{(G)})^{1/2}(\pazocal{U}^\ast_{g}\otimes \pazocal{N}_{A\to B})(\Phi_{A'A})(\rho_{A'}^{(G)})^{1/2}\big)\\
    &\eqt{(e)}D\left(\rho_{A'}^{(G)}\otimes \sigma_B\middle\|(\pazocal{U}^\ast_{g}\otimes \mathrm{Id}_{A})\left((\rho_{A'}^{(G)})^{1/2}J^{(\pazocal{N})}_{A'A}(\rho_{A'}^{(G)})^{1/2}\right)\right)\\
    &\eqt{(f)}D\big((\pazocal{U}^\ast_{-g}\otimes \mathrm{Id}_{A})(\rho_{A'}^{(G)}\otimes \sigma_B)\|(\rho_{A'}^{(G)})^{1/2}J^{(\pazocal{N})}_{A'A}(\rho_{A'}^{(G)})^{1/2}\big)\\
    &\eqt{(g)}D\big(\rho_{A'}^{(G)}\otimes \sigma_B\|(\rho_{A'}^{(G)})^{1/2}J^{(\pazocal{N})}_{A'A}(\rho_{A'}^{(G)})^{1/2}\big).
     \ee
     where
     \begin{itemize}
         \item (d) follows from~\eqref{eq:just_d} and (e) follows from~\eqref{eq:just_e};
         \item in (f) we have used the invariance of the quantum relative entropy under joint unitary conjugation;
         \item in (g) we have used the invariance of $\rho_{A'}^{(G)}$ under the action of $\pazocal{U}^\ast_g$, as in~\eqref{eq:just_g}.
     \end{itemize}
     So, by convexity of the quantum relative entropy, for any $\sigma_B$ 
     \bb
        D\big(\rho_{A'}^{(G)}\otimes \sigma_B^{(G)}\|(\mathrm{Id}_{A'}\otimes\pazocal{N}_{A\to B})(\phi_{A'A}^{(G)})\big)&\leq \frac{1}{|G|}\sum_{g\in G}D\big(\rho_{A'}^{(G)}\otimes \pazocal{V}_g\sigma_B\|(\mathrm{Id}_{A'}\otimes\pazocal{N}_{A\to B})(\phi_{A'A}^{(G)})\big)\\
        &\eqt{(h)}D\big(\rho_{A'}^{(G)}\otimes \sigma_B\|(\rho_{A'}^{(G)})^{1/2}J^{(\pazocal{N})}_{A'A}(\rho_{A'}^{(G)})^{1/2}\big),
     \ee
     where in (h) we have used~\eqref{eq:group_invariance}. This means that enlarging the set of candidate minimisers as in (vi) does not produce a strict inequality, since $G$-invariant states yield the same minimum. 
\end{proof}


\section{Proof of the weak additivity of the regularised exponents $\opnsa{\pazocal{N}}$ and $\errnsa{\pazocal{N}}$}
\label{app:reg}
In this appendix we prove that $\opnsa{\pazocal{N}}$ and $\errnsa{\pazocal{N}}$ are weakly additive and in particular:
\bb
\opnsa{\pazocal{N}}=\lim_{k\to \infty}\frac{1}{k}\opnsa{\pazocal{N}^{\otimes k}}\, , \quad  \errnsa{\pazocal{N}}=\lim_{k\to \infty}\frac{1}{k}\errnsa{\pazocal{N}^{\otimes k}}\, .
\ee

Let $M>0$ and let $m\leq n$. Then $\epsilon^{\rm NS,a}(\new{M,}\pazocal{N}^{\otimes m}\old{, R})\geq \epsilon^{\rm NS,a}(\new{M,}\pazocal{N}^{\otimes n}\old{, R})$ since any code using $m$ copies of $\pazocal{N}$ with message size at least $M$ can be seen as a code using $n$ copies of $\pazocal{N}$ --- acting non-trivially just on the first $m$ copies of the channel --- having message size at least $M$. Let $(n_\ell)_{\ell\in\mathbb{N}}$ be the subsequence achieving the liminf in~\eqref{eq:T_ast}:
    \bb
        \opnsa{\pazocal{N}}=
        \liminf_{M\to\infty}\lim_{\ell\to \infty}-\frac{1}{n_\ell}\log\epsilon^{\rm NS,a}(\new{M,}\pazocal{N}^{\otimes n_\ell}\old{,\log M}).
    \ee
    Fixed $k\geq 1$, for any $\ell\geq 1$ let us write $n_\ell=km^{(k)}_\ell +r^{(k)}_\ell$, where $0\leq r^{(k)}_\ell \leq k-1$ and $m_\ell^{(k)} \in \mathbb{N}$. Such representation is unique for any fixed $k\geq 1$. Then
    \bb\label{eq:ineqRk}
        -\frac{1}{km^{(k)}_\ell}\log\epsilon^{\rm NS,a}(\new{M,}\pazocal{N}^{\otimes km^{(k)}_\ell}\old{,\log M})
        &\leq -\frac{1}{km^{(k)}_\ell}\log\epsilon^{\rm NS,a}\left(\new{M,}\pazocal{N}^{\otimes km^{(k)}_\ell+r^{(k)}_\ell}\old{,\log M}\right)\\
        &=-\frac{km^{(k)}_\ell +r^{(k)}_\ell}{km^{(k)}_\ell}\frac{1}{n_\ell}\log\epsilon^{\rm NS,a}\left(\new{M,}\pazocal{N}^{\otimes n_\ell}\old{,\log M}\right).
    \ee
    Therefore,
    \bb
        \frac{1}{k}\opnsa{\pazocal{N}^{\otimes k}}&=\liminf_{M\to\infty}\liminf_{m\to \infty}-\frac{1}{km}\log\epsilon^{\rm NS,a}\left(\new{M,}\pazocal{N}^{\otimes km}\old{,\log M}\right)\\
        &\leq \liminf_{M\to\infty}\liminf_{\ell\to \infty}-\frac{1}{km_\ell^{(k)}}\log\epsilon^{\rm NS,a}\left(\new{M,}\pazocal{N}^{\otimes km_\ell^{(k)}}\old{,\log M}\right)\\
        &\leq \liminf_{M\to\infty}\liminf_{\ell\to \infty} -\frac{km^{(k)}_\ell +r^{(k)}_\ell}{km^{(k)}_\ell}\frac{1}{n_\ell}\log\epsilon^{\rm NS,a}\left(\new{M,}\pazocal{N}^{\otimes n_\ell}\old{,\log M}\right)\\
        &\eqt{(viii)}\,\opnsa{\pazocal{N}}
    \ee
    where in (viii) we have used that both the limits
    \bb
        \lim_{\ell\to \infty} \frac{km^{(k)}_\ell +r^{(k)}_\ell}{km^{(k)}_\ell}&=1  
    \ee
    and
    \bb
        \liminf_{M\to\infty}\lim_{\ell\to \infty} -\frac{1}{n_\ell}\log\epsilon^{\rm NS,a}\left(\new{M,}\pazocal{N}^{\otimes n_\ell}\old{,\log M}\right)&=\opnsa{\pazocal{N}}
    \ee
    exist and their product is well defined. The converse bound is simpler, since 
    \bb\label{eq:converseT}
        \frac{1}{k}\opnsa{\pazocal{N}^{\otimes k}}&=\liminf_{M\to\infty}\liminf_{m\to \infty}-\frac{1}{km}\log\epsilon^{\rm NS,a}\left(\new{M,}\pazocal{N}^{\otimes km}\old{,\log M}\right)\\
        &\geqt{(ix)} \liminf_{M\to\infty}\liminf_{n\to \infty}-\frac{1}{n}\log\epsilon^{\rm NS,a}\left(\new{M,}\pazocal{N}^{\otimes n}\old{,\log M}\right)\\
        &=\opnsa{\pazocal{N}},
    \ee
    where in (ix) we have relabeled $mk =n$ and we have relaxed the constraint that $n$ must be a multiple of $k$, obtaining a smaller liminf.
    The two inequalities imply that $\opnsa{\,\cdot\,}$ is weakly additive and, in particular,
    \bb
        \opnsa{\pazocal{N}}=\lim_{k\to \infty}\frac{1}{k}\opnsa{\pazocal{N}^{\otimes k}}.
    \ee
    To deal with $\errnsa{\,\cdot\,}$, let us write as before
    \bb
        \opnsa{\pazocal{N}}=
        \liminf_{R\to 0}\lim_{\ell\to \infty}-\frac{1}{n_\ell}\log\epsilon^{\rm NS,a}(\new{\exp(Rn_\ell),}\pazocal{N}^{\otimes n_\ell}\old{,Rn_\ell}).
    \ee
    for a suitable subsequence $(n_\ell)_{\ell\in\mathbb{N}}$ achieving the liminf, and, for any fixed $k\geq 1$, we introduce $n_\ell=km^{(k)}_\ell +r^{(k)}_\ell$, where $0\leq r^{(k)}_\ell \leq k-1$ and $m_\ell^{(k)} \in \mathbb{N}$. We will call
    \bb
        \xi^{(k)} \coloneqq \sup_{\ell}\frac{n_\ell}{k m^{(k)}_\ell} = \sup_{\ell}\frac{km^{(k)}_\ell +r^{(k)}_\ell}{k m^{(k)}_\ell} \in [1,2).
    \ee
    Then
    \bb
        \frac{1}{k}\errnsa{\pazocal{N}^{\otimes k}}&=\liminf_{R\to 0}\liminf_{m\to \infty}-\frac{1}{km}\log\epsilon^{\rm NS,a}(\new{\exp(Rm),}\pazocal{N}^{\otimes km}\old{, Rm})\\
        &\leq \liminf_{R\to 0 }\liminf_{\ell\to \infty}-\frac{1}{km_\ell}\log\epsilon^{\rm NS,a}\left(\new{\exp\left(Rkm^{(k)}_\ell\right),}\pazocal{N}^{\otimes km^{(k)}_\ell}\old{,Rkm^{(k)}_\ell}\right)\\
        &\eqt{(x)}\liminf_{R' \to 0 }\liminf_{\ell\to \infty}-\frac{1}{km_\ell}\log\epsilon^{\rm NS,a}\left(\new{\exp\left(R'\xi^{(k)} km^{(k)}_\ell\right),}\pazocal{N}^{\otimes km^{(k)}_\ell}\old{,R'\xi^{(k)} km^{(k)}_\ell}\right)\\
        &\leqt{(xi)} \liminf_{R' \to 0 }\liminf_{\ell\to \infty}-\frac{1}{km_\ell}\log\epsilon^{\rm NS,a}\left(\new{\exp\left(R'n_\ell\right),}\pazocal{N}^{\otimes km^{(k)}_\ell}\old{,R'n_\ell}\right)\\
        &\leqt{(xii)} \liminf_{R'\to 0}\liminf_{\ell\to \infty} -\frac{km^{(k)}_\ell +r^{(k)}_\ell}{km^{(k)}_\ell}\frac{1}{n_\ell}\log\epsilon^{\rm NS,a}\left(\new{\exp(R'n_\ell),}\pazocal{N}^{\otimes n_\ell}\old{,R'n_\ell}\right)\\
        &=\errnsa{\pazocal{N}},
    \ee
    where in (x) we have reparameterised $R=R'\xi^{(k)}$, in (xi) we have used that $\xi^{(k)}  km^{(k)}_\ell\geq n_\ell$ by the very definition of $\xi^{(k)}$ and then we have leveraged the monotonicity of $\epsilon^{\rm NS,a}$ with respect to the \new{message size}\old{communication rate}, which stems from its definition, and in (xii) we have used~\eqref{eq:ineqRk}. The converse bound is analogous to~\eqref{eq:converseT}. This proves the additivity of $\errnsa{\,\cdot\,}$. In particular,
    \bb
        \errnsa{\pazocal{N}}=\lim_{k\to \infty}\frac{1}{k}\errnsa{\pazocal{N}^{\otimes k}}\, .
    \ee


\section{Proof of the unassisted zero-rate error exponent of CQ-channels (Proposition~\ref{thm:cq-0-rate})}
\label{app:cq-0-rate}

In this section we present a proof (similar to the one of~\cite{Holevo2000} and~\cite{Dalai_2013}) of Proposition~\ref{thm:cq-0-rate} which we restate.

\begin{thm*}[\cite{Holevo2000, Dalai_2013}]
    Let $\pazocal{N}_{X\to B}(\,\cdot\,) = \sum_{x\in \mathcal{X}} \bra{x}\cdot \ket{x} \rho_x^B$ be a classical-quantum channel. Then, we have
\bb
\errpl{\pazocal{N}} = \sup_{P_X} \sum_{x_1,x_2\in\mathcal{X}}P_{X}(x_1)P_{X}(x_2)\left(-\log \Tr\left[\sqrt{\rho_{x_1}^B}\sqrt{\rho_{x_2}^B}\right]\right).
\ee 
\end{thm*}

\begin{proof}
Let $\pazocal{N}_{X\to B}(\cdot) = \sum_{x\in \mathcal{X}} \bra{x}\cdot \ket{x} \rho_x^B$ be a classical-quantum channel. For positive semi-definite operators $A$ and $B$, we denote by $A\wedge B = A-(A-B)_+$ their non commutative minimum. 

\textbf{Achievability.} We start with a simplified proof of the achievability. Compared to the proof of~\cite{Holevo2000}, the main difference lies in the proof of the following upper bound on the probability of error of sending a message $m\in [M]$:
\bb
\epsilon_m \le \sum_{m'\neq m} \Tr\left[\sqrt{\rho^B_{x(m)}}\sqrt{\rho^B_{x(m')}}\right],
\ee
for any encoding map $m\in [M]\mapsto x(m)\in \mathcal{X}$. We will use the following decoder 
\bb
Y(m) = \frac{\sqrt{\rho^B_{x(m)}}}{\sqrt{\rho^B_{x(m)}}+ \sum_{m'\neq m} \sqrt{\rho^B_{x(m')}}},
\ee
with the quotient of Beigi and Tomamichel~\cite{Beigi2024Aug}, defined as
\bb
    \frac AB\coloneqq\int_0^\infty \frac{1}{\lambda+B}A\frac{1}{\lambda+B}d\lambda.
\ee
 This quotient allows to use the operator inequality $\frac{A}{A+B}\le \frac{A}{B}$ for any $A\ge 0$ and $B>0$. The error probability can be bounded as follows 
\bb\label{app-eq-eps-m}
\epsilon_m &= \Tr\left[ \rho^B_{x(m)} (\id - Y(m))\right]
\\&= \Tr\left[ \rho^B_{x(m)} \frac{\sum_{m'\neq m} \sqrt{\rho^B_{x(m')}}}{\sqrt{\rho^B_{x(m)}}+ \sum_{m'\neq m} \sqrt{\rho^B_{x(m')}}}\right]
\\&\le  \Tr\left[ \rho^B_{x(m)} \frac{\sum_{m'\neq m} \sqrt{\rho^B_{x(m')}}}{\sqrt{\rho^B_{x(m)}}}
\right]
\\&= \sum_{m'\neq m} \Tr\left[ \sqrt{\rho^B_{x(m)}}  \sqrt{\rho^B_{x(m')}}
\right].
\ee  
The remaining of the proof uses  the expurgation technique~\cite{gallager1968information} as in~\cite{Holevo2000}. 

Let $P_X \in \mathcal{P}(\mathcal{X})$ be a probability distribution. Let $n\in \mathbb{N}^+$,  $r>0$, $M=\floor{\exp(rn)}$ and   $s>1$. 
We consider a random coding where each $x^n(m) \sim P_X^{\times n}$ for $m\in [M]$.  From the initial code, we can find a set $\mathcal{C}\subset [M]$ of size $|\mathcal{C}|\ge \frac{M}{2}-1$ such that for all $m\in \mathcal{C}$ we have that~\cite{gallager1968information}
\bb
\epsilon_m \le 2^s\,\mathbb{E}^s\left[ \epsilon_m^{1/s}\right].
\ee
We use the notation $\rho^{\otimes n}_{x^n(m)} = \rho^B_{x_1(m)}\otimes \rho^B_{x_2(m)}\otimes \cdots \otimes \rho^B_{x_n(m)}$ and the inequality~\eqref{app-eq-eps-m} to further upper bound the error probability as follows 
\bb
\epsilon_m &\le 2^s\,\mathbb{E}^s\left[ \epsilon_m^{1/s}\right]
\\&\le 2^s\,\mathbb{E}^s\left[\sum_{m'\neq m} \Tr^{1/s}\left[ \sqrt{\rho^{\otimes n}_{x^n(m)}}  \sqrt{\rho^{\otimes n}_{x^n(m')}}\right]\right]
\\&= 2^s\,(M-1)^s \left[\sum_{x,y} P_X(x)P_X(y)\left(\Tr\left[ \sqrt{\rho^{B}_{x}}  \sqrt{\rho^{B}_{y}}\right]\right)^{1/s} \right]^{ns}.
\ee 
Hence 
\bb
\liminf_{n\to \infty}-\frac{1}{n}\log\frac{1}{|\mathcal{C}|}\sum_{m\in \mathcal{C}}\epsilon_m &\ge -sr -s\log\sum_{x,y} P_X(x)P_X(y)\left(\Tr\left[ \sqrt{\rho^{B}_{x}}  \sqrt{\rho^{B}_{y}}\right]\right)^{1/s}.
\ee 
Taking $r\to 0^+$ then  $s\to \infty$, we deduce that: 
\bb 
 \errpl{\pazocal{N}} \ge \sup_{P_X} \sum_{x, y}P_{X}(x)P_{X}(y)\left(-\log \Tr\left[\sqrt{\rho_{x}^B}\sqrt{\rho_{y}^B}\right]\right). 
\ee

\textbf{Converse.} The converse proof have  similarities with the proof of~\cite[Theorem 4]{Dalai_2013} (see also~\cite{Nussbaum2009}). In particular, we use the Nussbaum-Szko\l a probability distributions to lift the problem to the classical setting. Moreover,~\cite{Dalai_2013} proved a similar inequality to~\eqref{app-eq-aud} which was proved previously by~\cite{Audenaert2008}. We employ Chebyshev's inequality in order to turn the minimum between probability distributions to a minimum between scalars, a method used by Blahut~\cite{blahut74}.  This allows us to use a variational formulation of the $\alpha$-R\'enyi divergence while~\cite{Dalai_2013} employed directly the tilted probability distribution which is actually the optimiser of this variational formulation~\cite{vanErven2014}.

In the block-length setting, we code over $\pazocal{N}^{\otimes n}$ where $n\in \mathbb{N}^+$.  We consider an unassisted  strategy with a codebook $\{x^{n}(m)\}_{m\in [M]}$ and a decoding POVM $\{Y(m)\}_{m\in [M]}$. For any $m, m'\in [M]$, we have that 
\bb 
\epsilon_{m} + \epsilon_{m'} &= \Tr\left[\rho^{\otimes n}_{x^n(m)}  \cdot ( \id -Y(m) )\right] + \Tr\left[\rho^{\otimes n}_{x^n(m')}  \cdot (\id-Y(m') )\right]
\\&\ge \Tr\left[\rho^{\otimes n}_{x^n(m)}  \cdot (\id-Y(m) )\right] + \Tr\left[ \rho^{\otimes n}_{x^n(m')}  \cdot Y(m) \right]
\\&\ge \Tr\left[ \rho^{\otimes n}_{x^n(m)} \wedge \rho^{\otimes n}_{x^n(m')}\right],
\ee
where we use  $\Tr[A\wedge B] \coloneqq \min_{0\le O\le \id} \Tr[AO] + \Tr[B(\id - O)]$ and the notation $\rho^{\otimes n}_{x^n(m)} = \rho^B_{x_1(m)}\otimes \rho^B_{x_2(m)}\otimes \cdots \otimes \rho^B_{x_n(m)}$. 

Given the  eigenvalue decomposition of $\rho_x= \sum_{a} \lambda^{x}_a \ketbra{\phi_a^x}$, we define  
$(p^{x,y},q^{x,y})$ as the Nussbaum-Szko\l a probability distributions for the pair of states $(\rho_x, \rho_{y})$:
\begin{align}\label{eq:NS-one-shot}
p^{x,y}_{a,b} = \lambda^{x}_a  |\spr{\phi^{x}_a}{\phi^{y}_b}|^2 \; \text{ and }\; q^{x,y}_{a,b}  = \lambda^{y}_b  |\spr{\phi^{x}_a}{\phi^{y}_b}|^2, \quad \forall (a,b).
\end{align}
For the pair of states $(\rho^{\otimes n}_{x_n(m)}, \rho^{\otimes n}_{x_n(m')})$, the Nussbaum-Szko\l a probability distributions are 
\bb
P_{a^n, b^n}^{m,m'} &= p^{x_1(m),x_1(m')}_{a_1,b_1}\times  p^{x_2(m),x_2(m')}_{a_2,b_2} \times \dots \times  p^{x_n(m),x_n(m')}_{a_n,b_n}, \quad \forall (a^n,b^n),
\\
Q_{a^n, b^n}^{m,m'} &= q^{x_1(m),x_1(m')}_{a_1,b_1}\times  q^{x_2(m),x_2(m')}_{a_2,b_2} \times \dots \times  q^{x_n(m),x_n(m')}_{a_n,b_n}, \quad \forall (a^n,b^n).
\ee
Using~\cite[Proposition 2]{Audenaert2008}, we have that 
\bb\label{app-eq-aud}
\Tr\left[ \rho^{\otimes n}_{x^n(m)} \wedge \rho^{\otimes n}_{x^n(m')}\right] \ge \frac{1}{2}\Tr\left[ P^{m,m'} \wedge Q^{m,m'}\right].
\ee
For a set of probability distributions $r = \{r^{x, y}\}_{x, y \in \mathcal{X}}$, we define 
\bb
R^{m,m'} = r^{x_1(m),x_1(m')}\times  r^{x_2(m),x_2(m')} \times \dots \times  r^{x_n(m),x_n(m')}, 
\ee 
and the good events 
\bb
\mathcal{G}^{m,m'}_{r,p}&=\left\{(a^n,b^n) : \log\left(\tfrac{P^{m,m'}_{a^n,b^n} }{R^{m,m'}_{a^n,b^n} } \right)\ge - {D(R^{m,m'}\|P^{m,m'})}-\sqrt{4{ \operatorname{Var}(R^{m,m'}\|P^{m,m'})} } \right\}, 
\\\mathcal{G}^{m,m'}_{r,q}&=\left\{(a^n,b^n) : \log\left(\tfrac{Q^{m,m'}_{a^n,b^n} }{R^{m,m'}_{a^n,b^n} } \right)\ge - {D(R^{m,m'}\|Q^{m,m'})}-\sqrt{4{ \operatorname{Var}(R^{m,m'}\|Q^{m,m'})} } \right\}.
\ee
By Chebyshev's inequality and union bound, we have that 
\bb\label{Chebyshev 2}
\sum_{ a^n,b^n } R^{m,m'}_{a^n,b^n} \id\{ (a^n,b^n)\in \mathcal{G}^{m,m'}_{r,p} \cap \mathcal{G}^{m,m'}_{r,q} \}\ge \frac{1}{2}.
\ee
Hence 
\begin{align}
\epsilon_{m} + \epsilon_{m'} &\ge  \Tr\left[ \rho^{\otimes n}_{x^n(m)} \wedge \rho^{\otimes n}_{x^n(m')}\right] \ge \frac{1}{2}\Tr\left[ P^{m,m'} \wedge Q^{m,m'}\right] 
\\&\ge \sup_{r} \frac{1}{2} \sum_{ a^n,b^n } R^{m,m'}_{a^n,b^n} \id\{ (a^n,b^n)\in \mathcal{G}^{m,m'}_{r,p} \cap \mathcal{G}^{m,m'}_{r,q} \} \min\left( \frac{P^{m,m'}_{a^n,b^n}}{R^{m,m'}_{a^n,b^n}}, \frac{Q^{m,m'}_{a,b}}{R^{m,m'}_{a,b}}\right)
\\&\ge  \sup_{r\in \pazocal{R}} \frac{1}{4}\exp\min\bigg(-{D(R^{m,m'}\|P^{m,m'})}-\sqrt{4{ \operatorname{Var}(R^{m,m'}\|P^{m,m'})}} , \label{eq:var1}
\\&\qquad  - {D(R^{m,m'}\|Q^{m,m'})}-\sqrt{4{ \operatorname{Var}(R^{m,m'}\|Q^{m,m'})}}\bigg), \label{eq:var2}
\end{align}
where we restrict the maximisation over $r\in \pazocal{R}$ that satisfy $\supp(r^{x,y}) \subset \supp(p^{x,y})\cap  \supp(q^{x,y})$ for all $x,y\in \mathcal{X}$. For such $r\in \pazocal{R}$, the variance terms in~\eqref{eq:var1} and~\eqref{eq:var2} can be bounded as follows:
\bb 
\operatorname{Var}(R^{m,m'}\|P^{m,m'})&= \sum_{x,y \in \mathcal{X}}  n_{x,y}^{x^n(m), x^n(m')}\operatorname{Var}(r^{x,y}\|p^{x,y})\le nA, \\
\operatorname{Var}(R^{m,m'}\|Q^{m,m'})&= \sum_{x,y \in \mathcal{X}} n_{x,y}^{x^n(m), x^n(m')}\operatorname{Var}(r^{x,y}\|q^{x,y})\le nA,
\ee
where $n_{x,y}^{x^n(m), x^n(m')} = \sum_{i=1}^n \id\{(x_i^n(m), x_i^n(m'))=(x,y)\}$ and $A$ is a constant that depends only on the channel $\pazocal{N}$. Noticing that we have a similar decomposition of the Kullback-Leibler divergence: 
    \bb 
D(R^{m,m'}\|P^{m,m'})&= \sum_{x,y \in \mathcal{X}}  n_{x,y}^{x^n(m), x^n(m')}D(r^{x,y}\|p^{x,y}), \\
D(R^{m,m'}\|Q^{m,m'})&= \sum_{x,y \in \mathcal{X}} n_{x,y}^{x^n(m), x^n(m')}D(r^{x,y}\|q^{x,y}),
\ee
we deduce that 
\bb
&\epsilon_{m} + \epsilon_{m'} 
\\&\ge \frac{\exp(-2\sqrt{nA}) }{4}\sup_{r \in \pazocal{R}}\exp\min\bigg(-\sum_{x,y}n_{x,y}^{x^n(m), x^n(m')} {D(r^{x,y}\|p^{x,y})}, -\sum_{x,y}n_{x,y}^{x^n(m), x^n(m')} {D(r^{x,y}\|q^{x,y})}\bigg)
\\&\eqt{(i)} \frac{\exp(-2\sqrt{nA})}{4}\sup_{r \in \pazocal{R}}\exp\inf_{0< \alpha< 1}\bigg(-\alpha\sum_{x,y}n_{x,y}^{x^n(m), x^n(m')} {D(r^{x,y}\|p^{x,y})} -(1-\alpha)\sum_{x,y}n_{x,y}^{x^n(m), x^n(m')} {D(r^{x,y}\|q^{x,y})}\bigg)
\\&\eqt{(ii)} \frac{\exp(-2\sqrt{nA})}{4}\exp\inf_{0< \alpha< 1}\sup_{r \in \pazocal{R}}\bigg(-\alpha\sum_{x,y}n_{x,y}^{x^n(m), x^n(m')} {D(r^{x,y}\|p^{x,y})} -(1-\alpha)\sum_{x,y}n_{x,y}^{x^n(m), x^n(m')} {D(r^{x,y}\|q^{x,y})}\bigg)
\\ &\eqt{(iii)}  \frac{\exp(-2\sqrt{nA})}{4}\exp\inf_{0< \alpha< 1} \bigg( -(1-\alpha) \sum_{x,y}n_{x,y}^{x^n(m), x^n(m')} {D_{\alpha}(p^{x,y} \| q^{x,y}  )}\bigg)
\\ &= \frac{\exp(-2\sqrt{nA})}{4}\exp\inf_{0 < \alpha < 1} \bigg( -(1-\alpha) \sum_{x,y}n_{x,y}^{x^n(m), x^n(m')} {D_{\alpha}(\rho_x  \| \rho_y  )}\bigg),
\ee
where in (i) we used $\min(a,b) = \inf_{0< \alpha< 1}( \alpha a+ (1-\alpha) b)$; in (ii), we used Sion's minimax theorem~\cite{Sion}; in (iii) we used the variational property of the $\alpha$-R\'enyi divergence~\cite{vanErven2014}. 

At this stage we can follow either~\cite{SHANNON1967522} or~\cite{Blinovsky2001} 
to show that there exists a probability distribution $P_X\in \mathcal{P}(\mathcal{X})$ and  a set $\mathcal{C}\subset [M]$ such that for all $m, m'\in \mathcal{C}$:
\bb
\inf_{0 < \alpha < 1} \bigg( -(1-\alpha) \sum_{x,y}n_{x,y}^{x^n(m), x^n(m')} {D_{\alpha}(\rho_x  \| \rho_y  )}\bigg)
\ge 
-n\sum_{x,y}   P_{X}(x)P_{X}(y)  \frac{1}{2} D_{1/2}(\rho_x\| \rho_y) - o(n),
\ee
which implies 
\bb 
\errpl{\pazocal{N}} &\le \liminf_{r\to 0^+}\liminf_{n\to \infty}-\frac{1}{n}\log \frac{1}{\floor{\exp(rn)}}\sum_{m=1}^{\floor{\exp(rn)}}\epsilon_m
\\&\le \liminf_{r\to 0^+}\liminf_{n\to \infty}-\frac{1}{n}\log \frac{1}{2\floor{\exp(rn)}}\sum_{m,m'\in \mathcal{C}}\epsilon_m+\epsilon_{m'}
\\&\le \sup_{P_X} \sum_{x,y}   P_{X}(x)P_{X}(y)  \frac{1}{2} D_{1/2}(\rho_x\| \rho_y)
\\&= \sup_{P_X} \sum_{x, y}P_{X}(x)P_{X}(y)\left(-\log \Tr\left[\sqrt{\rho_{x}^B}\sqrt{\rho_{y}^B}\right]\right).
\ee 
This concludes the proof.
\end{proof}


\end{document}